\newif\ifDissLarge
\newif\ifDissDim
\newif\ifDissHead

\DissLargetrue  % choices are \DissLargetrue or \DissLargefalse
\DissDimfalse     % choices are \DissDimfalse or \DissDimtrue
\DissHeadfalse    % choices are \DissHeadfalse or \DissHeadtrue

  \documentclass[12pt,a4paper,twoside]{report}
\newcommand{\MaquaHeadings}
{
  \ifDissHead
    \pagestyle{fancyplain}
 \lhead[\fancyplain{}{\bfseries \thepage}] {\fancyplain{}{\bfseries\sffamily\let\uppercase\relax\rightmark}}
\chead[\fancyplain{}{}]                 {\fancyplain{}{}}
\rhead[\fancyplain{}{\bfseries\sffamily\let\uppercase\relax\rightmark}]       {\fancyplain{}{\bfseries\thepage}}
 \else
    \pagestyle{headings}
  \fi
}

\ifDissDim
   \IfFileExists{typearea.sty}{\relax}{\DissDimfalse}
\fi
\ifDissDim
  \usepackage[a4paper,headinclude,DIV14,BCOR5mm]{typearea}
\else
\fi

\usepackage{latexsym}
\usepackage[intlimits]{amsmath}       % improved mathematical typesetting
\usepackage{amsfonts}
\usepackage{amsthm,amssymb,amscd}
\usepackage[utf8]{inputenc}

\newcommand{\tr}{\mbox{tr}}

\def\be{\begin{equation}}
\def\ee{\end{equation}}
\def\ve{\varepsilon}
\def\a{{\alpha}}

\def\l{{\lambda}}
\def\lb{{\overline\lambda}}
\def\b{{\beta}}

\def\g{{\gamma}}

\def\d{{\delta}}
\def\e{{\epsilon}}
\def\s{{\sigma}}
\def\k{{\kappa}}

\def\half{{1\over 2}}
\def\p{{\partial}}

\def\t{{\theta}}

\def\bar{\overline}
\def\({\left(}
\def\){\right)}
\def\cF{ {\cal F} }
\def\zb{ {\bar z}}
\def\wb{ {\bar w}}

\def \bra#1{\left\langle #1\right|}
\def \ket#1{\left| #1\right\rangle}

\def \vev#1{\left\langle #1 \right\rangle}

\def \tr{{\rm tr}}

\def \pbar{\bar{\partial}}

\def \be{\begin{equation}}
\def \ee{\end{equation}}

\def \la{\lambda^{\alpha}}
\def \t{\theta}

\renewcommand{\baselinestretch}{1.3}

\begin{document}

\renewcommand{\qedsymbol}{$\blacksquare$}
\newtheorem{def.}{Definition}
\newtheorem{teo}{Theorem}
\newtheorem{prop.}{Proposition}
\newtheorem{lema}{Lemma}
\newtheorem{corolario}{Corollary}
\newtheorem{postulado}{Postulate}

\pagenumbering{roman}

\def\baselinestretch{1.2}
%\setlength{\textheight}{25 true cm}
%\setlength{\textwidth}{15. true cm}
%\hoffset=-1.0 true cm
%\voffset=-2 true cm
%\textheight=21cm
%\topmargin=1.0cm
 \setlength{\textheight}{220mm}         % 217mm ... 223mm
  \setlength{\textwidth}{160mm}          % 157mm ... 160mm 
  \setlength{\evensidemargin}{2mm}       % 0mm ..10mm, also depends on  
  \setlength{\oddsidemargin}{2mm}       % the binding
  \setlength{\topmargin}{2mm}            % center page vertically, a few mm
  \setlength{\headsep}{8mm}              % 8mm ... 10mm
  \addtolength{\headheight}{2mm}         % for the hline in the header

\thispagestyle{empty}
%\def\thefootnote{\fnsymbol{footnote}}

%\begin{titlepage}
\thicklines
\begin{picture}(370,60)(0,0)
\setlength{\unitlength}{1pt}
\put(40,53){\line(2,3){15}}
\put(40,53){\line(5,6){19}}
\put(40,53){\line(1,1){27}}
\put(40,53){\line(6,5){33}}
\put(40,53){\line(3,2){25}}
\put(40,53){\line(2,1){19}}
\put(40,53){\line(5,-6){17}}
\put(40,53){\line(1,-1){22}}
\put(40,53){\line(6,-5){30}}
\put(40,53){\line(3,-2){22}}
\put(40,53){\line(-2,1){15}}
\put(40,53){\line(-3,1){23}}
\put(40,53){\line(-4,1){26}}
\put(40,53){\line(-6,1){36}}
\put(40,53){\line(-1,0){40}}
\put(40,53){\line(-6,-1){32}}
\put(40,53){\line(-3,-1){20}}
\put(40,53){\line(-2,-1){10}}
\put(75,45){\Huge \bf IFT}
\put(180,56){\small \bf Instituto de F\'\i sica Te\'orica}
\put(165,42){\small \bf Universidade Estadual Paulista} 
%\put(0,0){\line(1,0){375}}
\put(-25,2){\line(1,0){480}}
\put(-25,-2){\line(1,0){480}}
\end{picture}  

%%%%%%%%%%%%%%%%%%%%%%%%%%%%%%%%
%% Escolha a opcao apropriada %%
%%%%%%%%%%%%%%%%%%%%%%%%%%%%%%%%

%\vskip .3cm
%\noindent
%{DISSERTA\c C\~AO DE DOUTORADO}
%%%%%%%%%%% Numero da Dissertacao %%%%%%%%%%%%%
%\hfill    IFT--D.003/05\\

\noindent
{TESE DE DOUTORAMENTO}
%%%%%%%%%%%%%%% Numero da Tese %%%%%%%%%%%%% 
\hfill    IFT--T.007/08\\                  %%%mudar o numero

\vspace{3cm}
\begin{center}
%%%%%%%% TITULO DA DISSERTACAO OU TESE %%%%%%
{\large \bf Superstring Scattering Amplitudes with the Pure Spinor Formalism}
%%%%%%%%%%%%%%%%%%%%%%%%%%%%%%%%%%

\vspace{1.2cm}
%%%%%%%%%%%%%%% NOME DO AUTOR %%%%%%%%
Carlos Roberto Mafra
%%%%%%%%%%%%%%%%%%%%%%%%%%%%%%%%%
\end{center}

\vskip 3cm
\hfill Orientador
\vskip 0.4cm
%%%%%%%%%%%%% NOME DO ORIENTADOR %%%%%%%%%
\hfill {\em  Nathan Berkovits}
%%%%%%%%%%%%%%%%%%%%%%%%%%%%%%%%%
\vskip 4cm
\vfill
\begin{center}
%%%%%%%%%%%%%%%%%% DATA %%%%%%%%%%%%
Setembro de 2008
%%%%%%%%%%%%%%%%%%%%%%%%%%%%%%%%%
\end{center}

\newpage
\phantom{a}
\thispagestyle{plain}
\newpage
\pagenumbering{roman}

\begin{center}
{\Large \bf Agradecimentos}
\end{center}
\vskip 2.0cm
%%%%%%%%%%%% AGRADECIMENTOS %%%%%%%%%%%%%
\noindent 
%gratiarum actio
Primeiramente eu gostaria de agradecer ao Nathan pela 
orientação \emph{sui generis} ao longo dos últimos cinco anos.
Agora que as coisas se aproximam do seu final há uma tendência
natural de se recapitular a experiência de tê-lo como orientador.
Ainda não consegui assimilar as nuances que definem essa
situação, mas posso atestar para o fato de que elas realmente
mudaram minha maneira de aprender. Agradeço sinceramente 
por todos os momentos em que estive em sua sala e pude
aprender através da observação como é que se deve 
resolver um problema. Agora eu sei que envolve olhar
através da janela, mas depois de já ter tentado muitas vezes
acho que nunca conseguirei reproduzir o método completo.
Para piorar a situação, a minha sala no IFT não tinha janela.

Agradeço também ao meu amigo Oscar Bedoya por ser meu companheiro
científico. Eu poderia tentar enumerar os motivos que tornaram
a nossa convivência uma experiência engrandecedora para mim, 
mas estou triste demais por saber que nossos dias sob o mesmo
instituto estão prestes a terminar e talvez não seria
saudável ter uma listagem completa do que vou perder quando
partirmos. Señor Bedoya, muchas gracias por haberme enseñado
tantas cosas y en especial por su amistad! Nunca
me olvidaré que usted siempre estuvo listo para ayudarme.

Agradeço também à Dáfni e ao Clóvis, pelos companheiros que são.
Com eles aprendi coisas e foi um prazer enorme tê-los como 
amigos, pois nossas conversas nem sempre versaram sobre Física.

Ao Geová por suas perguntas, que me ensinaram mais do que ele imagina.
E ao Anderson por sempre ter ido a minha sala me chamar para almoçar,
mesmo sabendo que escutaria meus lamentos ou progressos sobre a conta
em dois loops. Não duvidaria que ele soubesse fazer a conta em sua
cabeça, pelo tanto que escutou e pelo inteligente que é.

Gostaria também de agradecer ao ICTP, pela incrível disposição
de complementar minha formação através das escolas em Trieste. 
Ao Pierre Vanhove por
ter me convidado para a escola em Cargèse. Finalmente à FAPESP por ter
sido a agência que financiou tanto meu mestrado quanto doutorado.

Um agradecimento canino também tem seu lugar aqui. Quero
agradecer à Mariposa por ter me acompanhado, debaixo da mesa,
enquanto estava sozinho em Cali escrevendo a tese e pensando 
na conta de 5-pontos ao mesmo tempo.
Não sei se ela sabia o motivo pelo qual ficávamos
acordados até as quatro da manhã, mas ela nunca saiu dos meus
pés.

Ao meu irmão, por nossas enumeráveis experiências que
são testemunhos indeléveis do quão felizes somos quando
estamos juntos. E também pela constante lembrança de nunca
abandonar os gibis.

À minha mãe, pelo amor, compreensão, estímulo e por ser uma pessoa
maravilhosa com tudo e com todos.

Ao meu pai, pelo exemplo de retidão e honestidade.
Foi ele também quem despertou meu interesse pela ciência, quando há
mais de vinte anos me contou o que tinha lido sobre a velocidade
da luz. Sua personalidade analítica e jeito matemático de pensar
são exemplos diários de como gostaria de ser. A ele dedico essa tese.

Finalmente agradeço a minha adorável esposa Maria del Pilar, pelo
amor incondicional.
Seu apoio constante tem sido de fundamental importância. Além das
nossas discussões científicas, ela nunca me disse que estava perdendo tempo
quando escrevia \emph{cleanup patches} para o kernel do linux. Mas também
soube me estimular a continuar estudando no momento certo, talvez 
por me conhecer melhor do que eu mesmo.

%%%%%%%%%%%%%%%%%%%%%%%%%%%%%%%%%%%%%

\newpage
%\phantom{a}
%\newpage

\begin{center}
{\Large \bf Resumo}
\end{center}
\vskip 2.0cm

%%%%%%%%%%%%%%%%%% RESUMO %%%%%%%%%

Esta tese discute como o formalismo de espinores puros
pode ser utilizado para calcular amplitudes de espalhamento
eficientemente. A enfâse recai sobre as expressões
dos fatores cinemáticos no superespaço de espinores puros, onde
as características simplificadoras inerentes dessa linguagem
nos permitiram relacionar explicitamente
as amplitudes de quatro-pontos em nível de árvore, um-loop e 
dois-loops. Enfatizamos como essas identidades simplificam 
de maneira elegante a tarefa de calcular as amplitudes de
quatro-pontos para
todas as possíveis combinações de partículas externas. Em particular,
as amplitudes envolvendo férmions em dois-loops nunca antes haviam sido
calculadas.

Também demonstramos a equivalência das amplitudes de um e dois-loops
entre os formalismos mínimo e não-mínimo. A
a variação de gauge da amplitude de seis-pontos dos glúons é calculada
para obter o fator
cinemático relacionado com o cancelamento da anomalia.
Alguns resultados parciais obtidos para a amplitude de cinco-pontos 
também serão discutidos.

%%%%%%%%%%%%%%%%%%%%%%%%%%%%%%%%

\vskip 1.0cm
%%%%%%%%%%%%% PALAVRAS CHAVE %%%%%%%%
%%%%%% SEPARADAS POR PONTO-E-VIRGULA %%%%%
\noindent

{\bf Palavras Chaves}: Supercordas; Supersimetria; Formalismo de Green-Schwarz;
Espinores Puros
%%%%%%%%% AREAS DO CONHECIMENTO %%%%%%%
%%%%%% SEPARADAS POR PONTO-E-VIRGULA %%%%%

\noindent

{\bf \'Areas do conhecimento}: Supersimetria; Teoria de Campos

\newpage
\begin{center}
{\Large \bf Abstract}
\end{center}
\vskip 2.0cm

%%%%%%%%%%%%%%%% ABSTRACT %%%%%%%%%%

This thesis discusses how the pure spinor formalism
can be used to efficiently compute superstring scattering
amplitudes. We emphasize the pure spinor superspace form
of the kinematic factors, where the
simplifying features of this language have allowed
an explicit relation among the massless four-point
amplitudes at tree-level, one- and two-loops to be found.
We show how these identities elegantly simplify the task
of computing the amplitudes for all possible external 
state combination related by supersymmetry. In particular,
the two-loop amplitudes involving fermionic states had never
been computed before.

By explicit calculation we show that the one- and two-loop 
amplitudes computed with the minimal and non-minimal formalisms are equivalent. 
Furthermore
we compute the gauge variation of the massless six-point open string amplitude
and obtain the kinematic factor related to the anomaly cancellation.
We also discuss some preliminary results regarding the massless 
five-point amplitude at one-loop.

%%%%%%%%%%%%%%%%%%%%%%%%%%%%%%%%

\vfill \eject
\newpage
\phantom{a}
\vspace{12cm}

\begin{flushright}
\textit{Praeterea, non debet poni superfluum aut aliqua distinctio sine causa,\\
quia frustra fit per plura quod potest fieri per pauciora.
}
\end{flushright}
\begin{flushright}
Petrus Aureolus, \textit{Scriptum super primum Sententiarum}
\end{flushright}

\begin{flushright}
\textit{At some point, ``performance'' is just more than a question of how fast\\
things are, it becomes a big part of usability.
}
\end{flushright}
\begin{flushright}
Linus Torvalds, \textit{linux-kernel mailing list}
\end{flushright}

\tableofcontents

%**************************************
\chapter{Introduction}
\pagenumbering{arabic}
%**************************************

In his later years Einstein struggled to find a unified theory describing both gravity
and electromagnetism and met failure. Nowadays superstring theory is the most promising
candidate to fulfill what Einstein envisioned in the last century. It unifies in a
quantum framework not only
gravity and electromagnetism but also the electroweak and strong forces. It is even
more than a simple construction which is able to handle all interactions together,
as it \emph{requires} them to be pieces of a whole setup which breaks down if
one of its parts is absent.

Among other things, superstring theory has provided us with a consistent quantum
description of the gravitational force. One particular oscillation mode of the closed
string has the right properties to be the quantum messenger of the gravitational force,
the graviton. And its interactions are described precisely by the Einstein-Hilbert
action,
\be
\label{EH}
S_{\rm EH} = \frac{1}{16\pi G_N}\int d^{10}x \sqrt{-g} R
\ee
plus quantum and superstring corrections to be described below \cite{sannan}.

Also present in the open superstring spectrum is a massless string
with spin one which describes the Yang-Mills gluons (or photons), whose
interactions in the low energy limit are described by the standard 
Yang-Mills action,
\be
\label{ym}
S_{\rm YM} = \frac{1}{g_{YM}^2}\int d^{10} x{\rm Tr}(F_{mn} F^{mn})
\ee
together with other quantum or superstring corrections.

One of the most fundamental questions which naturally arise 
when studying the low energy 
properties of the superstring interactions is to understand what
are the perturbative 
corrections to these two actions predicted by the theory.
That question automatically leads us to contemplate the fact that
superstring perturbation theory is finite to all loop orders \cite{finite}.
Therefore besides unifying all forces of nature, superstring theory does
it in such a way as to be \emph{finite}. No renormalization is ever needed
when deriving quantum corrections to the effective action.

One of the standard procedures to obtain these quantum corrections is through the
computation of scattering amplitudes. For example, the information needed
to derive higher-derivative terms in the Yang-Mills action \eqref{ym} is encoded 
in the scattering of the string counterparts of the gluons, \emph{i.e.},
the massless open strings with spin one. Analogously, quantum corrections to
the Einstein-Hilbert action are determined by the scattering of massless 
closed strings with spin two.

The tree-level scattering of three gluons, for example, can be used to find 
the three point vertex in the expansion of the Yang-Mills action \eqref{ym}.
Higher-point scatterings in string theory probe higher-order vertices in the
low energy effective action and so forth. But the
first true superstring corrections are obtained from the massless four-point
scattering at tree-level \cite{gross_witten}, and are of quartic order
in the field-strength $F_{mn}$ or Riemann tensor $R_{mnpq}$,
\be
\label{F4}
S \propto \frac{\a'^2}{g_{YM}^2} \int d^{10}x \,{\cal F}^4, \quad 
S \propto \frac{\a'^3}{16\pi G_N} \int d^{10} x \sqrt{-g}\,{\cal R}^4
\ee
where ${\cal F}^4$ and ${\cal R}^4$ are abbreviations for
\[
{\cal F}^4 = t_8^{mnpqrstu}F_{mn}F_{pq}F_{rs}F_{tu},
\]
\[
{\cal R}^4 = t_8^{mnpqrstu}t_8^{abcdefgh}R_{mnab}R_{pqcd}R_{rsef}R_{tugh},
\]
and the $t_8$ tensor is described in the Appendix \ref{t8_ap}.

Superstring theory -- no wonder -- is  supersymmetric, so there are many more
interactions in the effective actions than 
those of 
\eqref{EH} and \eqref{ym}.
In
fact their actions in the low energy limit are given by the ten-dimensional 
supergravity
and super-Yang-Mills actions, describing also their
fermionic superpartners; the gravitino and gluino. Furthermore,
all these extra terms are related by supersymmetry and 
also receive quantum and superstring corrections. Computing these
corrections to all those terms has proven to be a challenging task over
the years.

The computation of these various scattering amplitudes have been traditionally 
done using two different prescriptions, encompassed in the so-called
Ramond-Neveu-Schwarz \cite{RNS_formalism} 
or Green-Schwarz formalisms \cite{GS_formalism}\cite{GS_light_cone}.

\subsubsection{The Ramond-Neveu-Schwarz formalism}

The Ramond-Neveu-Schwarz formalism \cite{RNS_formalism} is based on
spacetime vectors $X^m(\s,\tau)$ and $\psi^m(\s,\tau)$ which are scalars 
(the $X^m$'s) or spinors (the $\psi^m$'s) of the two-dimensional
wordlsheet with coordinates $\s,\tau$. The lack of spacetime
spinors is the major source of difficulty in this formalism, as the
computation of scattering amplitudes for fermionic strings is not
natural in this framework. It has to be done using a
clever construction of vertex operators for the spacetime spinors which
uses \emph{spin fields} $\Sigma_{\a}$ \cite{FMS}
\[
\Sigma_{\a} = {\rm e}^{\pm i \phi_1}{\rm e}^{\pm i \phi_2}{\rm e}^{\pm i \phi_3}
{\rm e}^{\pm i \phi_4}{\rm e}^{\pm i \phi_5}
\]
and the bosonization of the $\psi^m$'s
\[
\psi^1 \pm i \psi^2 = {\rm e}^{\pm i \phi_1}, \quad
\psi^3 \pm i \psi^4 = {\rm e}^{\pm i \phi_2}, \quad
\psi^5 \pm i \psi^6 = {\rm e}^{\pm i \phi_3},
\]
\[
\psi^7 \pm i \psi^8 = {\rm e}^{\pm i \phi_4}, \quad
\psi^9 \pm i \psi^{10} = {\rm e}^{\pm i \phi_5}.
\]
Furthermore, because the $\psi^m$'s are spinors in the worldsheet,
the computation of higher-loop scattering amplitudes requires a
sum over different \emph{spin structures}. The fact that each
term can have divergences which are cancelled only
after the sum is performed also leads to difficulties.

So if one uses the scattering amplitude prescription of the Ramond-Neveu-Schwarz
formulation each scattering involving fermionic partners
has to be computed in isolation, and the computation of the fermionic
state is much more difficult due to the complicated nature of the
vertex operator. The formalism is said to lack manifest supersymmetry.

\subsubsection{The Green-Schwarz formalism}

In contrast, the Green-Schwarz formulation is manifestly 
supersymmetric \cite{GS_formalism}\cite{GS_light_cone}.
It is based on the worldsheet fields $X^m$ and $\t^{\a}$, which are spacetime
vectors and spinors, respectively. The drawback in this
formalism comes from the fact that it has a complicated action,
\[
S =  \frac{1}{\pi}\int d^2z \left[ 
		\frac{1}{2}\p X^m \bar{\p}X_m 
          	- i \p X^m \theta_L\gamma_m \bar{\p}\theta_L
                - i \bar{\p} X^m \theta_R\gamma_m\p\theta_R \right.
\]		
\[ - \left. \frac{1}{2}(\theta_L\gamma^m \bar{\p}\theta_L)(\theta_L\gamma_m\p\theta_L 
          +  \theta_R\gamma_m \p\theta_R)
        -\frac{1}{2}(\theta_R\gamma^m \p\theta_R)(\theta_L\gamma_m \bar{\p}\theta_L 
	  + \theta_R\gamma_m \bar{\p}\theta_R) 
	  \right],
\]
which is impossible to quantize preserving manifest Lorentz covariance.
By breaking $SO(1,9)$ covariance to $SO(8)$ with the light cone
gauge choice the action simplifies \cite{schwarz_phys_rep}
\[
S = \frac{1}{4\pi}\int d^2z \left(
	\p X^i\bar{\p}X^i + S^a_L\bar{\p}S^a_L + S^b_R\p S^b_R
    \right).
\]
In this gauge the construction of vertex operators is possible
and the computation of scattering amplitudes can be done. For
example, the gluon and gluino vertices are given, 
in a Lorentz frame where $k^+=0$, $\zeta^+ = \zeta^- = 0$,
by
\[
V_B(\zeta,k) = \zeta^i (\dot{X}^i - \frac{1}{4} S^a S^b k^j \g^{ij}_{ab})
{\rm e}^{ikX},
\]
\[
V_F(\zeta,k) = ( u^a F^a + u^{\dot{a}}F^{\dot{a}})
{\rm e}^{ikX},
\]
where 
\[
F^a = \sqrt{\frac{p^+}{2}}S^a, \quad
F^{\dot{a}} = (2p^+)^{-1/2}\Big[
(\g\cdot \dot{X} S)^{\dot{a}} + \frac{1}{3}:(\g^i S)^{\dot{a}}(S\g^{ij} S):k^j
\Big].
\]
However, the need of a non-covariant gauge and restricted kinematics
are features which reduce the power of this manifestly supersymmetric
approach. For example, in the light cone gauge one looses the
conformal symmetry of the original theory and therefore can not
use the powerful methods of conformal field theory.
Furthermore it is not always possible to impose those
restrictions simultaneously.

So up to the year 2000 the computations of superstring scattering
amplitudes were done using these two different formalisms.
The results were equivalent
but required different amounts of work to be performed. Due to
the issues mentioned above, however, there was little progress
in computing higher-loop and/or higher-point amplitudes.
Furthermore,
either spacetime supersymmetry or Lorentz covariance was hidden
in the middle steps. 
Nevertheless both symmetries
are fundamental requirements of superstring
theory and as such the results must respect them. The fact that the end result
has all these symmetries while they are not obvious in the middle steps
means that the formalisms were introducing spurious difficulties were there
should be none.

\subsubsection{The Pure Spinor formalism}

The pure spinor formalism was born at the dawn of the new millennium \cite{nathan_pure}, as
a successful attempt to solve this long-standing problem of finding a manifestly 
supersymmetric
and covariant superstring formalism. 
It has already been used to study several aspects of string theory, for example
the propagation of strings in curved backgrounds\footnote{For detailed computations
see the theses \cite{back_review} and \cite{back_review2}.} \cite{berk_howe}\cite{ps_background} 
where among other things it has been used to derive
the non-linear Born-Infeld equations of motion \cite{born_infeld}. Various aspects
of strings in $AdS_5\times S^5$ were also studied \cite{ps_ads}. It has also been used to derive 
the Chern-Simons correction required by the anomaly cancellation \cite{ps_chern}. Furthermore,
there is also research related to its inner workings \cite{partition}\cite{cohomology}\cite{ps_superspace},
including studies of its own origins \cite{origin}\cite{mazz}. The focus of this thesis, however,
is to show how it can be
used in the computations of scattering amplitudes\footnote{After this thesis was
finished a mixed open-closed amplitude has been computed by Alencar in \cite{geova}.}, 
highlighting the virtues and elegance of manifest Lorentz covariance and spacetime
supersymmetry in the results obtained.

The amplitudes computed so far in the pure spinor formalism turned
out to be easier to obtain. As a sounding example of
how simpler computations can be, a good measure is to compare
the hundred-pages long calculation of the four-point amplitude
at two-loops in the RNS formalism \cite{dhoker}\cite{zhu}
versus the ten-pages-long computation
using pure spinors \cite{twoloop}\cite{twolooptwo}.
Of course 
the results were shown to be equivalent \cite{twolooptwo},
as well as for all other amplitudes computed so 
far \cite{mafra_one}\cite{nmps_two}\cite{mafra_tree} (see \cite{nathan_brenno}
for a general tree-level proof), proving
that the pure spinor formalism produces the
same results while being simpler.

Right after the formalism came into light, the tree-level amplitudes
were shown to be equivalent with the RNS computations in \cite{nathan_brenno}, 
for amplitudes containing any number of bosons and 
up to four fermions. Years later, Berkovits spelled out the multiloop prescription 
\cite{nathan_multiloop}\cite{twoloop} and
paved the way to show the equivalence of his formalism up to the two-loop level,
which is the state-of-the-art situation as of 2008.

In the computations of massless four-point amplitudes 
the results can be written down in terms of a supersymmetric kinematic factor
in pure spinor superspace \cite{ps_superspace}
times a function which is manifestly equal to their RNS and GS counterparts.
So the comparison of the results require the evaluation of the
pure spinor superspace integrals appearing in the kinematic factors.
For example, the supersymmetric kinematic factors in the massless four-point
amplitudes at one- and two-loop order were originally written
as \cite{nathan_multiloop}\cite{twoloop} 
\be
\label{kin_one}
K_{\rm one-loop}=
\ee
\[
 \int d^{16}\t(\e T^{-1})^{((\a\b\g))}_{[\rho_1{\ldots} \rho_{11}]}
\t^{\rho_1}{\ldots}\t^{\rho_{11}} 
(\g_{mnpqr})_{\b\g} \Big[ A_{1\a}(\t)
(W_2(\t)\g^{mnp}W^3(\t)){\cal F}^{qr}_4(\t) +{\rm cycl}(234)\Big]
\]
\be
\label{kin_two}
K_{\rm two-loop} =
\ee
\[
\int d^{16}\theta (\e T^{-1})^{((\a\b\g))}_{[\rho_1{\ldots} \rho_{11}]}
\t^{\rho_1}{\ldots}\t^{\rho_{11}}
(\gamma^{mnpqr})_{\alpha\beta}\gamma^{s}_{\gamma\delta}
\Big[
{\cal F}^1_{mn}(\theta){\cal F}^2_{pq}(\theta){\cal F}^3_{rs}(\theta)
W^{4\delta}(\theta)
\Delta(z_1,z_3)\Delta(z_2,z_4)\Big.
\]
\[
+\rm{perm}(1234)\Big]
\]
where  
$\Delta(y,z)= \epsilon^{CD}\omega_C(y)\omega_D(z)$,
$\omega_C$ are the two holomorphic one-forms defined in \cite{dhoker},
$A_{\a}^I(\t)$, $W^{I\a}(\t)$ and 
${\cal F}_{mn}^I(\t)$
are the super-Yang-Mills connection and the linearized
spinor and vector superfield-strengths for the
$I^{th}$ external state 
with momentum $k^m_I$ satisfying $k_I\cdot k_I=0$, 
and $(T^{-1})^{\a\b\g}_{\rho_1 ...\rho_{11}}$ is
a Lorentz-invariant tensor which is antisymmetric
in $[\rho_1 ...\rho_{11}]$ and symmetric and $\g$-matrix traceless
in $(\a\b\g)$. Up to an overall normalization constant,
\[
(T^{-1})^{\a\b\g}_{\rho_1 ...\rho_{11}} =
\e_{\rho_1 ...\rho_{16}}(\g^m)^{\k\rho_{12}}
(\g^n)^{\s\rho_{13}}
(\g^p)^{\tau\rho_{14}}
(\g_{mnp})^{\rho_{15}\rho_{16}} (\d^{(\a}_{\k}\d^\b_\s\d^{\g)}_{\tau}
-\frac{1}{40}\g_q^{(\a\b}\d^{\g)}_\k \g^q_{\s\tau}).
\]
To finally get the final result for these pure spinor amplitudes
one should 
use the $\t$-expansions of the super-Yang-Mills superfields listed 
in \eqref{sym_exp}, plug them back into the above expressions,
compute the traces of a multitude of gamma matrix arrays (some of them containing
as much as twenty gamma matrices) and finally perform the $d^{16}\t$ superspace
integration. Looking at the multitude of vector and spinor
indices of \eqref{kin_one} and \eqref{kin_two} one would
conclude that the manifest Lorentz covariance and supersymmetry of the pure spinor
formalism were
making those kinematic factors expressions look awkward and laborious to
be evaluated. 
Fortunately that is not the case, in fact quite the opposite is true.
With the observation that, up to an overall coefficient, \cite{twolooptwo}
\be
\label{smart}
\int d^{16}\t(\e T^{-1})^{((\a\b\g))}_{[\rho_1{\ldots} \rho_{11}]}
\t^{\rho_1}{\ldots}\t^{\rho_{11}}f_{\a\b\g}(\t) = \langle \l^{\a}\l^{\b}\l^{\g}
f_{\a\b\g}(\t)\rangle
\ee
those scary-looking kinematic factors of \eqref{kin_one} and \eqref{kin_two}
simply become\footnote{
The biggest problem with the brute-force approach of 
computing thousand of gamma matrix traces which follow from expressions
like \eqref{kin_one}
is that one misses various identities which become clear in their
pure spinor superspace representation 
of \eqref{one_loop} and \eqref{two_loop}. The identities
\eqref{id_one} and \eqref{id_two} are simple examples of what can be
accomplished.
Furthermore, as the usual tool to compute traces of gamma matrices at my disposal
at that time was Mathematica with the package GAMMA, which become inefficient at
this specific task, computations along those lines could take more than 24 hours of
run-time, which I considered unacceptable. With the method of
Appendix \ref{chap_t} those
computations don't take longer than 1 minute (with FORM \cite{FORM}\cite{tform} it is a matter of
a couple of seconds). And following Linus Torvalds' citation
in this thesis, ``performance is a big part of usability''. In hindsight, it was the
performance requirements which I set as a goal in the beginning of this enterprise which
allowed the quick verification of superspace identities, making further progress much faster
than otherwise it would be.}
\be
\label{one_loop}
K_{\rm 1} =
\langle (\l A^1)(\l \g^m W^2)(\l \g^n W^3){\cal F}^4_{mn}\rangle
+ \rm{cycl.(234)},
\ee
\be
\label{two_loop}
K_{\rm 2} =
\langle (\l \g^{mnpqr}\l){\cal F}^1_{mn}{\cal F}^2_{pq}{\cal F}^3_{rs}
(\l \g^s W^4)\rangle\Delta(1,3)\Delta(2,4)
+ \rm{perm.(1234)}.
\ee
The 
pure spinor correlator in the right-hand side of \eqref{smart}
was defined since day one in \cite{nathan_pure} and allows a tremendous
simplification in the computations, which in fact become trivial to
perform. As explained in \cite{nathan_pure}, the pure spinor expression
$\langle \l^{\a}\l^{\b}\l^{\g} f_{\a\b\g}(\t)\rangle$ is to be evaluated
by selecting the terms which contain five $\t$'s proportional to 
the (normalized) pure spinor measure,
\be
\label{ps_measure}
\langle (\l\g^m \t)(\l\g^n \t)(\l\g^p \t)(\t\g_{mnp} \t)\rangle = 1.
\ee
In Appendix \ref{chap_t} we show that the computation of pure spinor
expressions containing an arbitrarily complicated combination of
three $\l$'s and five $\t$'s is uniquely determined by symmetry alone.
Using the method described in the appendix, many pure spinor superspace
expressions were explicitly evaluated with not much effort.
For example, 
in \cite{mafra_one} and \cite{twolooptwo} the kinematic factors for
the bosonic components of
\eqref{one_loop} and \eqref{two_loop} were shown to reproduce the
$t_8$-tensorial structure appearing in the low energy effective action
for superstrings \eqref{F4}. That provided the proof that the pure
spinor formalism
reproduces the same results as
the RNS formalism \cite{dhoker} up to the two-loop level. And as the
pure spinor expressions for the kinematic factors are supersymmetric,
the computation of the fermionic terms pose no further difficulties
and were also evaluated \cite{stahn}\cite{mafra_tree}.
This situation is in deep contrast to the need of computing
each amplitude separately for all the external superpartners as is the case in the RNS
and GS formalisms (see some RNS fermionic computations in \cite{sen}\cite{lin}). 
Furthermore, as summarized below, the simple
nature of the pure spinor representation for the kinematic factors allowed the
explicit proof that the massless four-point amplitudes were all related
to one another at different orders in perturbation theory, namely at tree-,
one- and two-loops.

Firstly, the idea was to
obtain a pure spinor superspace expression for the massless four-point kinematic
factor at tree-level \cite{mafra_tree}
\be
\label{tree_level}
K_{\rm 0}  = \half k^m_1k_2^n \langle(\l A^1)(\l A^2)(\l A^3){\cal F}^4_{mn}\rangle
-(k^1\cdot k^3)\langle A^1_n (\l A^2)(\l A^3)(\l \g^n W^4)
\rangle +(1\leftrightarrow 2).
\ee
Then it 
was shown through manipulations in pure spinor superspace that \eqref{tree_level}
was proportional to the massless four-point kinematic factor at 
one-loop \eqref{one_loop}
\be
\label{id_one}
K_{\rm 0} = -\langle (\l A^1)(\l\g^m W^2)(\l\g^n W^3){\cal F}^4_{mn}\rangle
          = -{1\over 3} K_{\rm 1}.
\ee
After that, using a proof based on BRST-equivalence of some
pure spinor expressions, the two-loop kinematic factor \eqref{two_loop} 
was related to the tree-level factor as follows
\be
\label{id_two}
K_{\rm 2} = -32 K_0 \left[(u-t)\Delta(1,2)\Delta(3,4) 
+(s-t)\Delta(1,3)\Delta(2,4) + (s-u)\Delta(1,4)\Delta(2,3)
\right].
\ee
That was the first time ever that these kinematic factors were shown to be
related as a whole without having to compute every possible scattering of
bosonic and/or fermionic states\footnote{And here we note that the scattering computation
of fermionic states at two loops has never been done using the RNS formalism (and nothing
at all with the GS formalism).} individually, showing case by case 
their proportionality to each other. Of course once one obtains the same kinematic
factor for the computation of four-point gluon
(or graviton) scattering at different loop orders, supersymmetry can 
be used to argue that the kinematic factors for the
superpartners are also the same, as can be seen in \cite{gswII}:
\begin{quote}
\emph{
Having discovered this result for bosons, it becomes plausible that supersymmetry
ensures that the one-loop four-particle amplitudes involving fermions also have
the same kinematic factors as the tree diagrams. In fact, this must be the case,
because the various K factors given in $\mathcal{x}$7.4.2 can be related to one another
by supersymmetry transformations.}
\end{quote}
Nevertheless, it is worth having an explicit simple proof that the kinematic
factors \eqref{tree_level}, \eqref{one_loop} and \eqref{two_loop} satisfy
the identities \eqref{id_one} and \eqref{id_two}. Note that explicit
two-loop computations
involving fermionic external states have never been done before the 
pure spinor computations of \cite{mafra_tree}\cite{stahn}.
Furthermore, with identities like \eqref{id_one} and \eqref{id_two}
it is not even needed to compute the one- and two-loop
kinematic factors explicitly in components
anymore. That is truly a remarkable simplification compared to
the standard RNS and GS formalisms.

So this thesis emphasises the study of pure spinor superspace
expressions and their role in obtaining simple relations for
seemingly complicated amplitudes. It is structured as follows.

In chapter 2 we review the pure
spinor formalism and the prescriptions to compute scattering
amplitudes in the minimal and non-minimal versions. 

In chapter 3 the manifestly supersymmetric 
kinematic factors for massless massless four-point amplitudes
at tree-, one- and two-loop levels are studied and
explicitly evaluated in components. 
In section \ref{anomaly_sec} we also compute the gauge variation of the
massless six-point amplitude for open strings, which gives
rise to a pure spinor
superspace representation for the gauge anomaly kinematic factor 
\be
\label{anomaly}
K_{\rm anom} = \langle (\l\g^m W^1)(\l\g^n W^2)(\l\g^p W^3)(W^4 \g_{mnp} W^5)\rangle.
\ee
Furthermore, in section \ref{t8_e10} we evaluate the bosonic components of
the interesting pure spinor superspace expression
\be
\label{teight}
\langle (\l\g^r W^1)(\l\g^s W^2) (\l\g^t W^3)(\t\g^m\g^n\g_{rst} W^4) \rangle,
\ee
from which the $t_8$ and $\e_{10}$ tensors naturally emerge in a unified manner,
in the form $\eta^{mn}t_8^{m_1n_1{\ldots}m_4n_4}-{1\over 2}\e_{10}^{mnm_1n_1{\ldots}
m_4n_4}$. 

In section \ref{crazy_four_sub} we digress about the expression
\be
\label{crazy_one}
\langle (\l\g^m\t)(\l\g^n W)(\l\g^p W)(W\g_{mnp} W) \rangle,
\ee
which turns out to be proportional to the one-loop kinematic factor 
of \eqref{one_loop} and
consequently it is supersymmetric despite the explicit appearance of $\t$. 
We will show its supersymmetry by relating the bosonic 
components of \eqref{crazy_one} with the
left hand side of the following identity
\be
\label{expla}
\langle \Big[(D\g_{mnp} A)\Big](\l\g^m W)(\l\g^n W)(\l\g^p W)\rangle = 
- 8 \langle (\l A)(\l\g^m W)(\l\g^n W){\cal F}_{mn}\rangle,
\ee
in such a way as to finally prove that
\[
\langle (\l\g^m\t)(\l\g^n W)(\l\g^p W)(W\g_{mnp} W) \rangle =
8 \langle (\l A)(\l\g^m W)(\l\g^n W){\cal F}_{mn}\rangle.
\]

Finally, in section \ref{fracas} we consider an intriguing 
pure spinor superspace expression
\be
\label{crazy_five}
\langle (\l\g^m \t)(\l\g^n \g^{rs} W^5)(\l\g^p W^1)(W^3\g_{mnp}W^4){\cal F}^2_{rs}\rangle
\ee
whose \emph{bosonic} component expansion reproduces the massless five-point amplitude of
open strings. We will show that \eqref{crazy_five} is proportional to
\be
\label{explatwo}
\langle (D\g_{mnp} A^1)(\l\g^m \g^{rs} W^5)(\l\g^n W^3)(\l\g^p W^4){\cal F}^2_{rs}\rangle - (2\leftrightarrow 5)
\ee
which is one of the terms produced in the evaluation of
\be
\label{cincoNMPS}
\langle (\lb\g_{mnp} D)\Big[ (\l A^1)(\l\g^m \g^{rs} W^5)(\l\g^n W^3)
(\l\g^p W^4){\cal F}^2_{rs}\Big]\rangle - (2\leftrightarrow 5),
\ee
which appears in the massless five-point computation with the non-minimal
pure spinor formalism\footnote{This is work in progress with the 
collaboration of Christian Stahn. Note added: It is now completed, see \cite{5pt}.}.

Chapter \ref{conclu} contains some conclusions and possible directions for
further inquiry along the lines of the study presented in this thesis.

In Appendix \ref{chap_t} we describe an efficient method to compute 
pure spinor superspace expressions in terms of the polarizations
and momenta of the external particles. This is the method  which was used
in several papers to obtain the final component expression for various
kinematic factors.

A brief review of  ${\cal N}=1$ super-Yang-Mills theory in $D=10$ is
given in Appendix \ref{sym_ap}, together with the explicit $\t$-expansion
of the superfields used in this thesis.

And finally the famous $t_8$-tensor is written down explicitly 
in Appendix \ref{t8_ap}. This is done both in terms of 
explicit Kronecker deltas as well as in terms of its contraction with four
field-strengths $F_{mn}$. We also present its
$U(5)$-covariant form which can be deduced from the pure spinor 
expression \eqref{crazy_one}.

%\newpage

%*****************************
\chapter{The %Almighty 
Pure Spinor Formalism}
%*****************************

The pure spinor formalism is an efficient tool to compute superstring
scattering amplitudes in a covariant way, and this is the aspect
which we will emphasize in this thesis.

Being manifestly supersymmetric
and containing no worldsheet spinors, it does not require the summation over
the spin structures which makes the evaluation of higher-loop amplitudes
in the RNS formalism a difficult task. And as it can be covariantly 
quantized, one does not need to go to the light-cone gauge as in the
Green-Schwarz formulation, avoiding the problems when one has to do
so. We will now review the origins of the pure spinor formalism and
how it was constructed, establishing our notation along the way.
Then we will explain how amplitudes are to be computed
using Berkovits' formalism.

%**************************************************************
\section{Siegel's modification of the Green-Schwarz formalism}
%**************************************************************

The main difficulty one faces when trying to quantize the Green-Schwarz
action (written here in the conformal gauge)
\[
S =  \frac{1}{\pi}\int d^2z \left[ 
		\frac{1}{2}\p X^m \bar{\p}X_m 
          	- i \p X^m \theta_L\gamma_m \bar{\p}\theta_L
                - i \bar{\p} X^m \theta_R\gamma_m\p\theta_R \right.
\]		
\[ - \left. \frac{1}{2}(\theta_L\gamma^m \bar{\p}\theta_L)(\theta_L\gamma_m\p\theta_L 
          +  \theta_R\gamma_m \p\theta_R)
        -\frac{1}{2}(\theta_R\gamma^m \p\theta_R)(\theta_L\gamma_m \bar{\p}\theta_L 
	  + \theta_R\gamma_m \bar{\p}\theta_R) 
	  \right],
\]
is related to the complicated nature of the fermionic constraints
$d_{\alpha}$. To see this we compute the conjugate momentum to 
$\t^{\a}_L$, denoted by $p^L_{\a}$, to obtain
\[
p^L_{\alpha} % =  \pi \frac{\delta S}{\delta \p_0 \t^{\a}_L} \\
              =  \frac{i}{2}\(\gamma_m\t_L\)_{\a}\left[ 
                  \Pi^m +\frac{i}{2}\(\t_L\gamma^m\p_1\t_L\)  \right].
\]
As it depends on $\t_L^{\a}$, it defines a constraint 
$d^L_{\alpha} = p^L_{\alpha} - \frac{i}{2}\(\t_L\gamma^m\)_{\a}
\Pi_m + \frac{1}{4}\(\t_L\gamma^m\)_{\a}\(\t_L\gamma_m\p_1\t_L\)$
which satisfies the OPE
\be
\label{ope_dd_bad}
 d^L_{\alpha}(z) d^L_{\beta}(w) \rightarrow -i\frac{\gamma_{\alpha\beta}^m\Pi_m}{z-w}.
\ee
Due to the Virasoro constraint $\Pi_m\Pi^m=0$ 
the relation \eqref{ope_dd_bad} mixes first and second class types of 
constraints in such a way
that is very difficult to disentangle them covariantly. The standard way 
to deal with this
situation is to go to the light-cone gauge, where the two types of 
constraints can be
treated separately in \eqref{ope_dd_bad}.

In 1986 Warren Siegel \cite{siegel} proposed a new approach to 
deal with this problem. His idea was to treat the conjugate momenta for 
$\theta^{\alpha}$
as an independent variable, proposing the following action
for the left-moving variables\footnote{We will restrict our attention
to the left-moving variables only, as it is straightforward to add
the right-moving part.}
\be
\label{siegel}
S = \frac{1}{2\pi} \int d^2z \left[ \frac{1}{2}\p X^m \bar{\p}X_m 
    + p_{\alpha}\bar{\p} \theta^{\alpha} \right].
\ee
Together with \eqref{siegel} one should add an appropriate set of first-class 
constraints to reproduce
the superstring spectrum.
The Virasoro constraint  
$T = -\frac{1}{2}\Pi^m\Pi_m - d_{\alpha}\p \theta^{\alpha}$ and the kappa symmetry
generators of the GS formalism, given by 
$G^{\alpha} =\Pi^m(\gamma_m d)^{\alpha}$, where
\be
\label{pi}
\Pi^m = \p X^m + \frac{1}{2}(\t\gamma^m\p \t)
\ee
should certainly be elements of that set of constraints.
Furthermore, in his approach the variable $d_{\a}$ 
\[
d_{\alpha} = p_{\alpha} -\frac{1}{2}\left(\p X^m 
             +\frac{1}{4}(\t \g^m\p\t) \right)(\g_m\t)_{\a}
\]
was not supposed to be a constraint.

Even though there was a successful description of the superparticle using
Siegel's approach, the whole set of constraints was never found for the
superstring case. However, as we shall see below, Siegel's idea was used by Berkovits
in his proposal for the pure spinor formalism.

Note that the action \eqref{siegel} defines a CFT whose OPE's are given by
\be
 X^m(z,\bar{z})X^n(w,\bar{w}) \longrightarrow -\frac{\a'}{2}\eta^{mn}\ln|z-w|^2,
\quad
\label{pteta}
 p_{\a}(z)\t^{\b}(w) \longrightarrow \frac{\d^{\b}_{\a}}{z-w},
\ee
\be
\label{dd_ope}
 d_{\a}(z)d_{\b}(w) \longrightarrow -\frac{\a'}{2}\frac{\g^m_{\a\b}\Pi_m}{z-w},
\quad
 d_{\a}(z)\Pi^m(w) \longrightarrow \frac{\a'}{2}\frac{(\g^m \p\t)_{\a}}{z-w}.
\ee
Furthermore, if $V(y,\t)$ is a generic superfield then its OPE's with
$d_{\a}$ and $\Pi^m$ are computed as follows
\be
\label{super_ope}
 d_{\a}(z)V(y,\t) \longrightarrow \frac{\a'}{2}\frac{D_{\a}V(y,\t)}{z-y}, \quad
 \Pi^m(z)V(y,\t) \longrightarrow \frac{\p^mV(y,\t)}{z-y},
\ee
where the supersymmetric derivative $D_{\a}$ is given by
\be
\label{susy_deriv}
D_{\a} = \frac{\p}{\p\t^{\a}} +\frac{1}{2}(\g^m\t)_{\a}\p_m.
\ee

The energy momentum tensor for the action 
\eqref{siegel} is given by
\[
T(z) = -\frac{1}{2} \p X^m \p X_m - p_{\alpha}\p \theta^{\alpha}
\]
as can be easily checked by using the known results of the bosonic string 
and the $bc$ system
with $\l=1$, in the notation of \cite{polchinski_1}. Furthermore, the
central charge is $c=+10-32=-22$, where each pair of $p_{\a}$ and $\t^{\b}$
have $c=-3(2\lambda -1)^2 + 1 = -2$, for a total of $-32$. 

The non-vanishing of the central charge leads to problems when
quantizing the theory, so that was a major difficulty in Siegel's
approach to the GS formalism.

Furthermore in \cite{siegel} Siegel proposed that the supersymmetric 
integrated
massless vertex operator in his approach should be 
\be
\label{siegel_U}
U = \int dz(\p\t^{\a}A_{\a}(x,\t) +A_m(x,\t)\Pi^m + d_{\a}W^{\a}(x,\t) )
\ee
where the superfields appearing in \eqref{siegel_U} are the
SYM superfields which are reviewed in the appendix.
But there is a problem with this supposition if one wants it
to be equivalent to the RNS formalism, where the vertex operator
for a gluon is given by (see (7.3.25) in \cite{gswI})
\be
\label{int_rns}
U_{\rm gluon}^{\rm RNS} = \int dz (A_m \p X^m + \frac{1}{2}\psi^m \psi^n F_{mn}),
\ee
where the field-strength is $F_{mn} = \p_m A_n - \p_n A_m$.
To see this one uses the superfield
expansions of appendix \ref{sym_ap} to conclude that the gluon
vertex operator obtained from \eqref{siegel_U} is 
\[
U_{\rm gluon}^{\rm Siegel} = \int dz (A_m \p X^m - \frac{1}{4}(p\g^{mn}\t)F_{mn})
\]
Comparing both expressions we see that the operator which multiplies
$\frac{1}{2}F_{mn}$ is the Lorentz current for the fermionic
variables in each formalism.
To see this we use
Noether's method to define the variation of \eqref{siegel}
under the Lorentz transformation to be 
$\delta S = \frac{1}{2\pi} \int \frac{1}{2}\bar{\p}\varepsilon_{mn}\Sigma^{mn}$,
where
\[
\delta p_{\alpha}  = 
 \frac{1}{4}\ve_{mn}(\gamma^{mn})_{\a}^{\phantom{a}\beta}p_{\beta}, \qquad
\delta \theta^{\a}  = \frac{1}{4}\ve_{mn}(\gamma^{mn})^{\a}_{\phantom{a}\beta}\t^{\beta}.
\]
Therefore the variation of \eqref{siegel} is 
\begin{align*}
\frac{1}{2\pi}\int \delta(p_{\alpha}\bar{\p}\theta^{\alpha}) & =
  \frac{1}{2\pi}\int\left[
 \frac{1}{4}\ve_{mn}(\gamma^{mn})_{\a}^{\phantom{a}\beta}p_{\beta} \bar{\p}\theta^{\alpha}
+\frac{1}{4}p_{\alpha}\pbar (\ve_{mn}(\gamma^{mn}\t)^{\alpha})
   \right].\\
 & = \frac{1}{2\pi}\int\left[ 
+\frac{1}{4}\bar{\p}\ve_{mn}p_{\a}(\gamma^{mn})^{\a}_{\phantom{a}\beta}\t^{\beta}
                       \right],
\end{align*}
so that 
\be
\label{sigma_siegel}
\Sigma^{mn} = \frac{1}{2}(p\gamma^{mn}\theta)
\ee
is the Lorentz currents of the fermionic variables.
However the Lorentz currents of the fermionic variables in
Siegel's approach had a double pole coefficient of $+4$
instead of $+1$ as in the RNS formalism.
Using the OPE \eqref{pteta} we get
\[
\Sigma^{mn}(w)\Sigma^{pq}(z) 
         = \frac{1}{4}\frac{ p(\gamma^{mn}\gamma^{pq} - \gamma^{pq}\gamma^{mn})\theta}{w-z}
	    +\frac{1}{4}\left( \frac{\tr(\gamma^{mn}\gamma^{pq})}{(w-z)^2} \right),
\]
\be
\label{ope_siegel}
= \frac{\eta^{p[n}\Sigma^{m]q} 
   - \eta^{q[n}\Sigma^{m]p}}{w-z} 
                               + 4\frac{\eta^{m[q}\eta^{p]n}}{(w-z)^2}
\ee
where we used that $\gamma^{mn}\gamma^{pq} - \gamma^{pq}\gamma^{mn} = 2\eta^{np}\g^{mq} -
2\eta^{nq}\g^{mp}+ 2\eta^{mq}\g^{np} - 2\eta^{mp}\g^{nq}$ and ${\rm tr}(\gamma^{mn}\gamma^{pq}) =
-32 \d^{mn}_{pq}$. Recalling that in the RNS formalism the OPE of the Lorentz currents
for the fermionic variables $\Sigma_{\rm RNS} = \psi^m\psi^n$ satisfies
\be
\label{ope_rns}
\Sigma_{RNS}^{mn}(w)\Sigma_{RNS}^{pq}(z) \rightarrow 
\frac{\eta^{p[n}\Sigma_{RNS}^{m]q} - \eta^{q[n}\Sigma_{RNS}^{m]p}}{w-z} 
                               + \frac{\eta^{m[q}\eta^{p]n}}{(w-z)^2}
\ee
the different double pole coefficient in \eqref{ope_siegel} and \eqref{ope_rns}
would make the computations of scattering amplitudes using \eqref{int_rns} or \eqref{siegel_U}
not agree with each other.

% **************************************************************
\section{The elements which led to the pure spinor formalism}
%****************************************************************

The modification of Siegel's approach proposed by 
Berkovits in the year 2000 was based in the observation that there
existed a set of ghost variables with $c_g=+22$ and whose contribution
to the double pole of the Lorentz currents was $-3$. So
the problems described above would no longer exist if
that set of ghosts was added to 
Siegel's action \eqref{siegel}.
That discovery led to the creation of the pure spinor formalism.
Let's now take a look at some of its ingredients in such a way
as to motivate the solution found by Berkovits\footnote{The ``history''
presented here is merely a pedagogical attempt to show how pure spinors
naturally solve the issues which were present in Siegel's approach, namely
the non-vanishing central charge and the different double pole in the
Lorentz generator appearing in the integrated vertex operator. It should not
be interpreted as ``history'' {\it per se}, but as an exposition artifact. It
is interesting to note, however, the prior developments which happened 
with the superembedding approach pioneered by 
Sorokin {\it et. al.} \cite{embedding}\cite{Nemb}\cite{Temb}. For a 
review see \cite{sor_rev}. In 2002 there was a paper which
discussed the pure spinor formalism from the perspective of the 
superembedding approach \cite{mazz}.}.

%***************************************************
\subsection{Lorentz currents for the ghosts}
%***************************************************

When trying to construct the Lorentz currents for the fermionic
variables in the pure spinor formalism,
Berkovits suggested to modify the Lorentz currents \eqref{sigma_siegel}
by the addition of a contribution $N^{mn}$
coming from the ghosts,
\[
M^{mn} = \Sigma^{mn} + N^{mn}.
\]
The newly defined $M^{mn}$ would satisfy the same OPE \eqref{ope_rns} 
as in the RNS formalism if 
\begin{align}
\label{ope_NN}
N^{mn}(w)N^{pq}(z) &\rightarrow \frac{\eta^{p[n}N^{m]q} - \eta^{q[n}N^{m]p}}{w-z}
                               -3\frac{\eta^{m[q}\eta^{p]n}}{(w-z)^2},\\
\label{ope_SN}			       
\Sigma^{mn}(w)N^{pq}(z) &\rightarrow \text{regular},
\end{align}
as one can check as follows
\begin{align*}
M^{mn}(w)M^{pq}(z) &= (\Sigma^{mn}(w)+N^{mn}(w))(\Sigma^{pq}(z)+N^{pq}(z))\\
                   &\rightarrow \Sigma^{mn}(w)\Sigma^{pq}(z) + N^{mn}(w)N^{pq}(z)\\
		   &\rightarrow \frac{\eta^{p[n}M^{m]q} - \eta^{q[n}M^{m]p}}{w-z} 
                               +\frac{\eta^{m[q}\eta^{p]n}}{(w-z)^2}.
\end{align*}
At the same time those ghosts should have the right properties as to contribute
$c_g=+22$ to the central charge, otherwise the total central charge would be
non-vanishing.
Fortunately the right solution to both
problems was found when a proposal for the BRST charge was put forward.
As we will see, that 
provided the hint as to what was missing in the long quest for finding
a manifestly spacetime supersymmetric and covariant formalism: \textbf{pure spinors}.

%**************************************************
\subsection{The BRST operator}
%**************************************************

The next step in the line of reasoning which led to the pure spinor 
formalism is the proposal of the BRST operator\footnote{It is interesting
to note that pure spinors had already been used by Howe in \cite{howe_ps1}
(see also \cite{howe_ps2}) to obtain the on-shell constraint of
ten-dimensional super-Yang-Mills and supergravity (and also D=11 SG). 
One can check that equation (4) of \cite{howe_ps1} is essentially the BRST charge
\eqref{Q_brst} of the pure spinor formalism.}
\be
\label{Q_brst}
Q_{BRST} = \oint \lambda^{\alpha}(z)d_{\alpha}(z),
\ee
where $\l^{\a}$ are bosonic and  
$
d_{\a} = p_{\a} -\frac{1}{2}(\gamma^m\t)_{\a}\p X_m
-\frac{1}{8}(\gamma^m\t)_{\a}(\t\gamma_m \p\t).
$
However the BRST
charge \eqref{Q_brst} must satisfy the basic consistency condition $Q^2_{BRST}=0$
(see \cite{polchinski_1}\cite{kiritsis} for more details about the BRST quantization). 
Using \eqref{Q_brst} we obtain
\[
Q^2_{BRST} = \frac{1}{2}\{Q_{BRST}, Q_{BRST}\} = 
            -\frac{1}{2}\oint dz (\lambda\gamma^m\lambda)\Pi_m,	    
\]
therefore the bosonic fields $\l^{\a}$ must satisfy the constraints
\be
\label{espinor_puro}
(\lambda\gamma^m\lambda) = 0.
\ee
\begin{def.}[Pure Spinor]
A ten dimensional Weyl spinor 
$\lambda^{\alpha}$ is said to be a pure spinor if 
\eqref{espinor_puro} is satisfied for $m=0{\ldots} 9$.
\end{def.}

The formalism discovered by Berkovits is based on the properties of the pure spinor $\l^{\a}$, and
it is important to study what are the number of degrees of freedom which survives
the constraints \eqref{espinor_puro}. Naively one could think that those ten constraints
would imply that a ten dimensional pure spinor would have only $16-10=6$ degrees of
freedom, but that's not the case. To see this it is convenient to perform a Wick rotation
and break manifest $SO(10)$ Lorentz symmetry to its $U(5)$ subgroup.

A Weyl spinor of SO(10) decomposes under U(5) as follows
\[
16 \rightarrow (1_{\frac{5}{2}},10_{\frac{1}{2}}, \bar{5}_{-\frac{3}{2}}),
\]
where the subscript denotes the U(1)-charge. Using this decomposition for $\l^{\a}$
we can solve the constraints \eqref{espinor_puro} explicitly,
\begin{align}
\lambda_+ &= e^s, \\
\lambda_{ab} &= u_{ab},\\
\quad \lambda^a &= -\frac{1}{8}e^{-s}\epsilon^{abcde}u_{bc}u_{de}, \label{la_fant}
\end{align}
for any $s$ and antisymmetric $u_{ab}$.
To prove this\footnote{The proof is based on \cite{grassi_cov}, which the reader should
consult for more details.} we note that 
$\la \gamma^m_{\alpha\beta} \l^{\b}$ is obtained from the 32-dimensional expression
$\lambda^T (C \Gamma^m) \lambda$, where $C$ is the conjugation matrix satisfying
$C\Gamma^m = -\Gamma^{m,T}C$ which is given by
$C=\prod_{i=1}^5 (a_i - a^i)$.

Under the decomposition of $SO(10) \rightarrow U(5)$ the constraint \eqref{espinor_puro}
goes to two independent equations
\begin{align}
\label{satis0}
\bra{\lambda}Ca^i\ket{\lambda} &= 0 \hspace{0.7cm} i=1,2,3,4,5\\
\bra{\lambda}Ca_i\ket{\lambda} &= 0 \label{satis}.
\end{align}
In the above expressions the only non-vanishing terms are the ones proportional
to\footnote{To prove this one computes $\bra{0} C a^1
a^2a^3a^4a^5\ket{0} =1$ and notes that the expression is completely
antisymmetric in the exchange of its indices.}
$
\bra{0} C a^i a^j a^k a^l a^m \ket{0} = \epsilon^{ijklm}.
$ Therefore, using the known expansion of a Weyl spinor in terms of creation
operators
\[
\ket{\lambda} = \lambda_{+}\ket{0} +\frac{1}{2}\lambda_{ij}a^ja^i\ket{0} 
+\frac{1}{24}\lambda^i \epsilon_{ijklm}a^ma^la^ka^j\ket{0}
\]
equation \eqref{satis0} becomes
\be
\label{pure}
\bra{\lambda}Ca^p\ket{\lambda} = \lambda_{+}\bra{0}Ca^p\ket{\lambda}
  +\frac{1}{2}\lambda_{ij}\bra{0}a_i a_j Ca^p\ket{\lambda}
  +\frac{1}{24}\lambda^i\epsilon_{ijklm}\bra{0}a_j a_k a_l a_m Ca^p\ket{\lambda}.
\ee
But, 
\begin{align*}
\lambda_{+}\bra{0}Ca^p\ket{\lambda} & = \frac{1}{24}\lambda_{+}\lambda^i\epsilon_{ijklm}
	\bra{0}C a^p a^m a^l a^k a^j \ket{0}\\
		                    & = \frac{1}{24}\lambda_{+}\lambda^i\epsilon_{ijklm}
	\epsilon^{pmlkj}\\
	                            & = \lambda_{+}\lambda^p
\end{align*}
Analogously, by noting that $a_iC= -Ca^i$ and $a^iC=-Ca_i$ we obtain
\begin{align*}
\frac{1}{2}\lambda_{ij}\bra{0}a_i a_j Ca^p\ket{\lambda} &=  \frac{1}{4}\epsilon^{pijkl}
	\lambda_{ij}\lambda_{kl} \\
\frac{1}{24}\lambda^i\epsilon_{ijklm}\bra{0}a_j a_k a_l a_m Ca^p\ket{\lambda} & =
	\lambda_{+}\lambda^p.
\end{align*}
Plugging the above results into \eqref{pure} we arrive at
$
2\lambda_{+}\lambda^a + \frac{1}{4}\epsilon^{abcde}
	\lambda_{bc}\lambda_{de} = 0
$ 
which is easily solved by
\be
\lambda_{+}  \equiv e^s, \quad
\lambda_{ab}  \equiv u_{ab}, \quad
\label{solution}
\lambda^a  = -\frac{1}{8}e^{-s}\epsilon^{abcde}u_{bc}u_{de}.
\ee
One can also show that
\eqref{satis} is automatically satisfied by the above parametrisation, therefore
the eleven degrees of freedom of $e^s$ and $u_{ab}$ together with \eqref{solution}
correctly describe the ten-dimensional pure spinor $\l^{\a}$.

Let's now see how the pure spinor $\l^{\a}$
can be used to solve the issues present in the approach of Siegel to
the Green-Schwarz formulation.

\section{The pure spinor formalism}

To solve the pure spinor constraint it was convenient to break the manifest
$SO(10)$ symmetry to its subgroup $U(5)$, so the solution \eqref{solution} is
written in terms of $U(5)$ variables. Therefore using this solution one is
able to write down only the $U(5)$-covariant Lorentz currents
\[
N^{mn} \longrightarrow (n, n^b_a, n_{ab}, n^{ab}).
\]
We will be required to check whether
the $U(5)$ Lorentz currents constructed out of the variables 
$s(z)$, $u_{ab}(z)$ and their conjugate momenta $t(z)$ and $v^{ab}(z)$
satisfy the required condition \eqref{ope_NN}. To do this we will first
need to know how the OPE \eqref{ope_NN} decomposes under $SO(10)\rightarrow U(5)$.
This can be summarized by the following statement, which we will prove
in the appendix.
\begin{teo}
\label{teo_lorentz}
If the SO(10)-covariant OPE of the Lorentz currents
$N^{mn}$ is given by
\begin{equation}
\label{ope}
N^{kl}(y)N^{mn}(z) \rightarrow \frac{\delta^{m[l}N^{k]n}(z) -
\delta^{n[l}N^{k]m}(z)}{y-z} - 3\frac{\delta^{kn}\delta^{lm} -
\delta^{km}\delta^{ln}}{(y-z)^2},
\end{equation}
then the U(5)-covariant currents $(n,n_a^b,n_{ab},n^{ab})$ satisfy the
following OPE's:
\begin{align}
\label{ope1}
n_{ab}(y)n^{cd}(z) & \rightarrow 
              \frac{-\delta^{[c}_{[a}n^{d]}_{b]}(z) - 
	      \frac{2}{\sqrt{5}}\delta^{[c}_a\delta^{d]}_b n(z) }{y-z} 
	      + 3\frac{\delta^{[c}_a\delta^{d]}_b}{(y-z)^2} \\
\label{ope2}	      
n_a^b(y)n_c^d(z) &\rightarrow \frac{\delta^b_cn_a^d(z) -\delta_a^dn_c^d(z)}{y-z}
	      -3 \frac{\delta_a^d\delta_c^b -\frac{1}{5}\delta^b_a\delta_c^d}{(y-z)^2}\\
\label{ope3}	      
n(y)n(z) &\rightarrow  - \frac{3}{(y-z)^2}, \\
\label{ope4}
n(y)n_{ab}(z) &\rightarrow  -\frac{2}{\sqrt{5}}\frac{n_{ab}}{y-z}\\
\label{ope5}
n(y)n^{ab}(z) &\rightarrow +\frac{2}{\sqrt{5}}\frac{n^{ab}}{y-z} \\
\label{ope6}
n(y)n^a_b(z) &\rightarrow \text{regular}
\end{align}
\end{teo}

Furthermore there is one more (consistency) condition to be obeyed when constructing
those $U(5)$ Lorentz currents. The pure spinor 
$\la$ must obviously transform as a spinor under the action of $N^{mn}$,
\[
\delta \la = \frac{1}{2}\left[ \oint dz
\e_{mn}M^{mn},\la \right] = \frac{1}{4}\e_{mn}(\g^{mn}\l)^{\a}
\]
As the OPE of $\la$ with $\Sigma^{mn}$ have no poles we conclude that 
the pure spinor must satisfy,
\be
\label{lambda_N}
N^{mn}(y)\l^{\a}(z) \rightarrow 
\frac{1}{2}\frac{(\g^{mn})^{\a}_{\phantom{m}\b}\l^{\b}(z)}{(y-z)}.
\ee
By the same reasoning, the OPE \eqref{lambda_N} must also be broken to U(5) if
we want the check whether the U(5) Lorentz currents to be described below 
satisfy it. That is
\begin{teo}
\label{teo_lambda}
If the OPE in SO(10)-covariant language is given by
\be
\label{ope_Nl}
N^{mn}(y)\lambda^{\alpha}(z) \rightarrow 
\frac{1}{2}\frac{(\g^{mn})^{\a}_{\phantom{m}\b}\l^{\b}(z)}{(y-z)},
\ee
then the OPE's between 
$(n,n^a_b,n_{ab},n^{ab})$ and
$(\lambda_+,\lambda_{cd},\lambda^c)$ are given by
\begin{align}
n(y)\lambda_+(z) &\rightarrow -\frac{\sqrt{5}}{2}\frac{\lambda_+(z)}{y-z} \label{ope_3}\\
n(y)\lambda_{cd}(z) &\rightarrow -\frac{1}{2\sqrt{5}}\frac{\lambda_{cd}(z)}{y-z} \label{segu}\\
n(y)\lambda^c(z) &\rightarrow \frac{3}{2\sqrt{5}}\frac{\lambda^c(z)}{y-z} \label{sex}\\
n^a_b(y)\lambda_+(z) &\rightarrow \text{regular} \label{ter} \\
n^a_b(y)\lambda_{cd}(z) &\rightarrow \frac{ \delta^a_d\lambda_{cb} 
                                     - \delta^a_c\lambda_{db} }{(y-z)}
                                     -\frac{2}{5}\frac{\delta^a_b\lambda_{cd}}{(y-z)} \label{quar}\\
n^a_b(y)\lambda^c(z) &\rightarrow \frac{1}{5}\delta^a_b\lambda^c -\delta^c_b\lambda^a \\
n_{ab}(y)\lambda_+(z) &\rightarrow \frac{\lambda_{ab}(z)}{y-z} \label{pri} \\
n_{ab}(y)\lambda_{cd}(z) &\rightarrow \epsilon_{abcde}\lambda^e \label{ope_dificil} \\%pag 19
%\end{align}
%\begin{align}
n_{ab}(y)\lambda^c(z) &\rightarrow \text{regular} \label{6ind}\\
n^{ab}(y)\lambda_+(z) &\rightarrow \text{regular} \label{reg} \\
n^{ab}(y)\lambda_{cd}(z) &\rightarrow 
                         -\frac{\delta^{[a}_c\delta^{b]}_d\lambda_+(z)}{y-z}\label{quin}\\
n^{ab}(y)\lambda^c(z) &\rightarrow -\frac{1}{2}\epsilon^{abcde}\lambda_{de} 
\label{quinf}
\end{align}
\end{teo}

Will it be possible to find an action for 
$s(z)$, $u_{ab}(z)$, $t(z)$ and $v^{ab}(z)$ and explicitly construct the Lorentz
currents $(n,n_a^b,n_{ab},n^{ab})$ out of those variables in such a way as to
reproduce all the above OPE's? If it was impossible to do
this then the pure spinor formalism would have never been born.
In the following paragraphs we will see the solution found by Berkovits.

%***************************************************
\subsection{The action for the ghosts}
%***************************************************

The action for the ghosts appearing in the pure spinor constraint is given by
\be
\label{acao_fantasma}
S_{\lambda} = \frac{1}{2\pi}\int d^2z \left( 
	-\p t\pbar s +\frac{1}{2}v^{ab}\bar{\p}u_{ab}
\right)
\ee
where $t(z)$ and $v^{ab}(z)$ are the conjugate momenta for $s(z)$ and $u_{ab}(z)$.
Furthermore $s(z)$ and $t(z)$ chiral bosons, so that we must impose their
equations of motions by hand $\bar{\p}s=\bar{\p}t =0$.
The OPE's are given by
\begin{align}
\label{a1}
t(y)s(z) &\rightarrow \ln{(y-z)} \\
\label{a2}
v^{ab}(y)u_{cd}(z) & \rightarrow 2\frac{\d^{ab}_{cd}}{y-z} =
\frac{\delta^{[a}_c\delta^{b]}_d}{y-z}.
\end{align}

One of the most important results which allowed the birth of the pure 
spinor formalism is given by the following theorem
\begin{teo}
\label{misterio}
If the U(5)-symmetric Lorentz currents are built out of the ghosts as
follows
\begin{align}
n &= -\frac{1}{\sqrt{5}}\left( \frac{1}{4}u_{ab}v^{ab}+\frac{5}{2}\p t 
- \frac{5}{2}\p s \right) \label{n_lorentz} \\
n^a_b & = u_{bc}v^{ac} - \frac{1}{5}\delta^a_b u_{cd}v^{cd} \\
n^{ab} & = -\text{e}^s v^{ab} \label{n^ab_lorentz}\\
n_{ab} & = \text{e}^{-s}\left(
           	2\p u_{ab} - u_{ab}\p t - 2u_{ab}\p s + u_{ac}u_{bd}v^{cd} 
                - \frac{1}{2}u_{ab}u_{cd}v^{cd} 
	   \right) \label{teo_ghost}
\end{align}
then their OPE's among themselves and with 
$\lambda_+$, $\lambda_{ab}$ e $\lambda^a$ correctly reproduce
the relations \eqref{ope1}-\eqref{ope6} and \eqref{ope_3}-\eqref{quinf}, if
$s(z)$, $t(z)$, $v^{ab}(z)$ e $u_{ab}(z)$ satisfy the OPE's \eqref{a1} and \eqref{a2}. 
\end{teo}
\begin{proof}
We will explicitly check a few of those OPE's as the others can be shown along
similar lines.
For example, one can easily check \eqref{segu} as follows,
\begin{align*}
n(y)\lambda_{cd}(z) &= -\frac{1}{\sqrt{5}}\(\frac{1}{4}u_{ab}v^{ab}+\frac{5}{2}\p t
-\frac{5}{2}\p s  \)u_{cd}(z) \\
 & \rightarrow -\frac{1}{4\sqrt{5}}u_{ab}(y)\(\frac{\delta^{[a}_c\delta^{b]}_d}{y-z}\) =
 -\frac{1}{2\sqrt{5}}\frac{\lambda_{cd}(z)}{y-z}.
\end{align*}
Similarly, \eqref{reg} is easily seen to be true because $s(z)$ has no
poles with itself nor with $v^{ab}(y)$,
\[
n^{ab}(y)\lambda_+(z) =-(e^sv^{ab})e^s \rightarrow \text{regular}.
\]
Using \eqref{a2} we check the validity of \eqref{quin},
\[
n^{ab}(y)\lambda_{cd}(z) = -(e^sv^{ab})u_{cd} 
  \rightarrow -\frac{\delta^{[a}_c\delta^{b]}_d}{y-z}\lambda_+.
\]
The OPE \eqref{ope4} requires a bit more work but it also comes
out right. Using \eqref{n_lorentz} we get
\[
n(y)n(z)  \rightarrow\frac{1}{80}\( u_{ab}v^{cd}:v^{ab}u_{cd}:
+u_{ab}:v^{ab}u_{cd}:v^{cd}
+:u_{ab}v^{cd}:v^{ab}u_{cd}\)
 -\frac{5}{4}\(:\p t\p s: + :\p s \p t:\),
\]
and one can check that the simple pole terms cancel while for the double
pole we get
\[
  \rightarrow -\frac{1}{80}
 \frac{\delta^{[a}_c\delta^{b]}_d\delta^{[c}_a\delta^{d]}_b}{(y-z)^2}
 -\frac{10}{4}\frac{1}{(y-z)^2} 
       \rightarrow -\frac{3}{(y-z)^2},	 
\]
so it correctly reproduces \eqref{ope4}.
Finally we check \eqref{ope_dificil},
\begin{align*}
n_{ab}(y)\lambda_{cd}(z) &\rightarrow e^{-s}\(
u_{ae}u_{bf}:v^{ef}u_{cd}: -\frac{1}{2}u_{ab}u_{ef}:v^{ef}u_{cd}: \)\\
 & \rightarrow 
 e^{-s}\(u_{ae}u_{bf}-\frac{1}{2}u_{ab}u_{ef}\)\frac{\delta^{[e}_c\delta^{f]}_d}{y-z}\\
 & \rightarrow e^{-s}\(u_{ac}u_{bd} -u_{ad}u_{bc} -u_{ab}u_{cd} \) =
 \epsilon_{abcde}\lambda^e,
\end{align*}
where in the last line we used \eqref{la_fant}.
The proof for all the other cases is analogous and will be omitted.
\end{proof}

We will show in the following that the central charge for the ghost action
\eqref{acao_fantasma}
is $+22$, which is indeed the required value for it to annihilate the
total central charge when added to Siegel's action.

The energy momentum tensor for the ghosts can be found using
Noether's procedure, with the following definition for the 
variation of the action
\[
\delta S_{\lambda} = \frac{1}{2\pi}\int d^2z [\bar{\p}\varepsilon T_{\lambda}(z)
                     + \p\bar{\varepsilon} \bar{T}_{\lambda}(\bar{z}) ],
\]
under the conformal transformations of 
\begin{align}
\label{t1}
\delta v^{ab} & = \p\varepsilon  v^{ab} +\varepsilon \p v^{ab} +\bar{\varepsilon}\bar{\p}v^{ab}  \\
\label{t2}
\delta u_{ab} & =  \varepsilon \p u_{ab} +\bar{\varepsilon}\bar{\p}u_{ab}\\
\label{t3}
\delta \p s & =  \p\varepsilon \p s +\varepsilon \p^2 s 
                 +\p\bar{\varepsilon} \bar{\p} s +\bar{\p}\bar{\varepsilon}\p s  \\
\delta \bar{\p}t & =  \varepsilon \p\bar{\p}t + \bar{\p}\varepsilon\p t 
                      +\bar{\p}\bar{\varepsilon}\bar{\p}t 
                      +\bar{\varepsilon}\bar{\p}\bar{\p}t.
\label{t4}		      
\end{align}
Doing this we obtain
\[
T_{\lambda}(z) = \frac{1}{2}v^{ab}\p u_{ab} +\p t \p s + \p^2 s.
\]
For example,
\begin{align*}
\delta (\bar{\p}t\p s) &= \left(  
	\varepsilon \p\bar{\p}t + \bar{\p}\varepsilon\p t 
                      +\bar{\p}\bar{\varepsilon}\bar{\p}t 
                      +\bar{\varepsilon}\bar{\p}\bar{\p}t
\right)\p s + \bar{\p}t\left(
	\p\varepsilon \p s +\varepsilon \p^2 s 
                 +\p\bar{\varepsilon} \bar{\p} s +\bar{\p}\bar{\varepsilon}\p s
\right)\\
                     & = \p(\varepsilon\bar{\p}t\p s) + \bar{\p}(\bar{\varepsilon}\bar{\p}t\p s)
		         + \bar{\p}\varepsilon \p t\p s + \p \bar{\varepsilon}\bar{\p}s\bar{\p}t,
\end{align*}
so up to a surface term, $T(z)=\p t\p s$ is the contribution from the variables $s,t$.
The contribution from the variables $v^{ab}$ e $u_{ab}$ can be easily obtained
by noticing that it is a $\beta\gamma$ system with $\lambda=1$, if the following
identification is made $\beta\rightarrow -1/2v^{ab}$ and $\gamma \rightarrow u_{ab}$.
As the energy momentum tensor for $\b\g$ system is given by \cite{polchinski_2}
$T(z)=\p\beta\gamma-\lambda\p(\beta\gamma)=\frac{1}{2}v^{ab}\p u_{ab}$, it follows 
that
\[
T(z)  = \frac{1}{2}v^{ab}u_{ab} + \p t \p s.
\]
To justify the addition of the term $\p^2 s$ in $T(z)$ we compute the OPE
of $T(y)$ with the Lorentz current $n(z)$ from Theorem \ref{misterio}. We get
\[
T(y)n(z) \rightarrow \frac{\sqrt{5}}{(y-z)^3} + \frac{n(z)}{(y-z)^2} + \frac{\p n(z)}{(y-z)}
\]
where the triple pole comes from 
\[
\left(\frac{1}{2}v^{ab}\p u_{ab}\right)\left(\frac{1}{4\sqrt{5}}u_{cd}v^{cd}\right) \rightarrow
\frac{1}{8\sqrt{5}}\frac{\delta^{[a}_c \delta^{b]}_d\delta^{[c}_a \delta^{d]}_b}{(y-z)^3}
    = \frac{\sqrt{5}}{(y-z)^3}.
\]
Therefore the Lorentz current would fail to be a primary field, but that
can be fixed by the addition of $\p^2 s$, because
\[
\p^2 s(y) n(z) \rightarrow \frac{\sqrt{5}}{2}:\p^2 s(y) \p t(z): = -\frac{\sqrt{5}}{(y-z)^3}.
\]
So we have shown that the energy momentum tensor for the ghost variables is given by
\be
\label{t_lambda}
T_{\lambda}(z) = \frac{1}{2}v^{ab}\p u_{ab} +\p t \p s + \p^2 s.
\ee
The central charge can be easily computed by considering the fourth order
pole in $T_{\lambda}(y)T_{\lambda}(z)$. There are two contributions 
\begin{align*}
\frac{1}{4} :v^{ab}(y)\p u_{cd}(z)::\p u_{ab}(z)v^{cd}(y):  & =
    	\frac{1}{4}\frac{\delta^{[a}_c \delta^{b]}_d\delta^{[c}_a \delta^{d]}_b}{(y-z)^4} 
	 =  \frac{10}{(y-z)^4},\\
\intertext{and}
:\p t(y)\p s(z):: \p s(z) \p t(y): & = \frac{1}{(y-z)^4},
\end{align*}
whose sum imply that $c_g= +22$.
Therefore, as there are no poles between the ghosts and matter variables,
the total central charge of the energy momentum tensor in
the pure spinor formalism 
\be
\label{tensor_total}
T(z) = -\frac{1}{2} \p X^m \p X_m - p_{\alpha}\p \theta^{\alpha} + 
       \frac{1}{2}v^{ab}\p u_{ab} +\p t \p s + \p^2 s,
\ee
is zero.

The conclusion from the previous discussion is that the  addition
of the pure spinor ghost action of \eqref{acao_fantasma} to the 
Siegel action \eqref{siegel} makes the central charge of the theory
to vanish and implies that the Lorentz currents have the same
OPE as in the RNS formalism. So the pure spinor formalism action
for the left-moving fields 
is given by
\be
\label{acao_nathan}
S= \frac{1}{2\pi}\int dz \big[ \frac{1}{2}\p X^m \bar{\p}X_m +p_{\a}\bar{\p}\t^{\a} -
\p t\pbar s +\frac{1}{2}v^{ab}\pbar u_{ab} \big].
\ee
The variables in the pure spinor formalism have the following 
supersymmetry transformations
\be
\label{susy1}
\delta X^m  = \frac{1}{2} \(\varepsilon \gamma^m \theta\),\quad
\delta \theta^{\a}  = \varepsilon^{\a}, \quad \delta (\text{ghosts})  = 0,
\ee
\be
\label{susy3}
\delta p_{\beta}  = -\frac{1}{2}\varepsilon^{\a} \gamma^m_{\a\beta}\p X_m
                   + \frac{1}{8}\varepsilon^{\a}\t^{\gamma}\p
		   \t^{\d}\g^m_{\b\d}\gamma_{m\, \g\a}\quad
\ee
and one can check that they are generated by
\[
Q_{\a} =  \oint ( p_{\a} +\frac{1}{2}\gamma^m_{\a\beta}\t^{\beta}\p X_m 
+ \frac{1}{24}\gamma^m_{\a\beta}\gamma_{m\,\gamma\delta}\t^{\beta}\t^{\gamma}\p
\t^{\delta}),
\]
which satisfy the supersymmetry algebra
\[
\{Q_{\a},Q_{\beta}\} = \gamma^m_{\a\beta}\oint \p X_m.
\]
For example, the variation of \eqref{acao_nathan} under \eqref{susy1}-\eqref{susy3}
can be checked to be
\be
\label{var_ia}
\delta S = \int dz \Big[\frac{1}{4}(\ve\g^m\p\t)\pbar X_m
+ \frac{1}{4}(\ve\g^m\pbar\t)\p X_m -\frac{1}{2}(\ve\g^m\pbar\t)\p X_m 
-\frac{1}{8}(\bar{\p}\t\gamma^m\p \t)(\varepsilon\gamma_m\t) \Big].
\ee
Integrating the first term by parts we get
$-1/4\int (\ve\g^m\pbar\p\t)X_m$, which can be integrated by
parts again to result in
$+1/4\int (\ve\g^m\pbar\t)\p X_m $. Therefore the sum of the first two terms
of \eqref{var_ia} cancels the third. So the supersymmetry variation of
the pure spinor action \eqref{acao_nathan} will be zero if
$\int (\bar{\p}\t\gamma^m\p \t)(\varepsilon\gamma_m\t)$ vanishes.
To see that this we integrate it by parts to obtain
\begin{align*}
\int (\bar{\p}\t\gamma^m\p \t)(\varepsilon\gamma_m\t) & =
-\int (\t\gamma^m\p \t)(\varepsilon\gamma_m\bar{\p}\t)
-\int (\t\gamma^m\bar{\p}\p \t)(\varepsilon\gamma_m\t)\\
& = -\int (\t\gamma^m\p \t)(\varepsilon\gamma_m\bar{\p}\t)
+\int (\p\t\gamma^m\bar{\p} \t)(\varepsilon\gamma_m\t)
+\int (\t\gamma^m\bar{\p} \t)(\varepsilon\gamma_m\p\t)\\
& = - \int\t^{\a}\p\t^{\beta}\varepsilon^{\gamma}\bar{\p}\t^{\sigma}
\(\gamma^m_{\a\beta}\gamma_{m\,\gamma\sigma} - \gamma^m_{\beta\sigma}\gamma_{m\,\gamma\a}
+ \gamma^m_{\a\sigma}\gamma_{m\,\gamma\beta} \)\\
& = + 2\int \t^{\a}\p\t^{\beta}\varepsilon^{\gamma}\bar{\p}\t^{\sigma}
    \(\gamma^m_{\beta\sigma}\gamma_{m\,\gamma\a}\) \\
    & = -2\int \(\bar{\p}\t\gamma^m\p\t\)\(\varepsilon\gamma_m\t\),
\end{align*}
where we used $\g^m_{\a(\b}(\g_m)_{\g\d)} = 0$. We therefore
conclude that $\int (\bar{\p}\t\gamma^m\p \t)(\varepsilon\gamma_m\t) =0$,
which finishes the proof that \eqref{acao_nathan} is
supersymmetric.

%******************************************************
% *** CORRENTE FANTASMA *******************************
%******************************************************
We can define the ghost number of any state $\Psi(y)$ by
\[
[ \oint dz J(z), \Psi(y) ] = n_g \Psi(y),
\]
where the ghost current $J(z)$ is given by \cite{nathan_ICTP}
\be
\label{corr_ghost}
J(z) = \frac{1}{2}u_{ab}v^{ab} + \p t + 3\p s.
\ee
One can check that the ghost current defined above satisfies the
following OPE's \cite{nathan_multiloop}\cite{nekrasov_charac},
\begin{align}
J(y)\la(z) &\rightarrow \frac{\la}{y-z} \label{u5_denovo}\\
J(y)J(z) & \rightarrow -\frac{4}{(y-z)^2}\\
J(y)T(z) & \rightarrow -\frac{8}{(y-z)^3}+\frac{J(z)}{(y-z)^2} \label{anom}\\
J(y)N^{mn}(z) &\rightarrow \text{regular} \label{N_escalar}\\
\label{ano}
T(y)J(z) & \rightarrow \frac{8}{(y-z)^3} + \frac{J(z)}{(y-z)^2} + \frac{\p J(z)}{y-z}.
\end{align}
For example, to show that
\eqref{u5_denovo} is true we must compute the OPE's of 
$J(y)$ with the U(5) components of 
$\la$ to check that the results are compatible. So
\[
J(y)\lambda_+(z) \rightarrow \quad :\p t(y)e^s(z): =
                 \rightarrow \frac{\lambda_+}{y-z}
\]		
\[
J(y)\lambda_{cd} \rightarrow \frac{1}{2}u_{ab}(y):v^{ab}(y)u_{cd}(z):
                  \quad =\frac{1}{2}u_{ab}\frac{\delta^{[a}_c\delta^{b]}_d}{y-z}
                 =  \frac{\lambda_{cd}}{y-z}
\]
The triple pole of \eqref{anom}, for example, comes from the following
contractions
\begin{align*}
\frac{1}{4}:u_{cd}(y)v^{ab}(z)::v^{cd}(y)\p u_{ab}(z): & \rightarrow 
   -\frac{1}{4}\frac{\delta^{[a}_c\delta^{b]}_d\delta^{[c}_a\delta^{d]}_b}{(y-z)^3} 
   = -\frac{10}{(y-z)^3} \\
:\p t(y)\p^2 s(z): &\rightarrow \frac{2}{(y-z)^3}
\end{align*}
whose sum results in the coefficient $-8$.
The proof for the other OPE's is similar and therefore
will be omitted.

From 
\eqref{u5_denovo} we can see that the ghost number of the pure spinor
$\la$ is +1. Moreover from \eqref{N_escalar} we see that $J(z)$
is a Lorentz scalar (as it should be) and from \eqref{ano}
that there is an anomaly of +8
in the ghost current, which has conformal weight $h=1$.

% *************************************************************
\section{Massless vertex operators}
%**************************************************************

The physical states in the pure spinor formalism 
are defined to be in the cohomology of the BRST
operator 
\[
Q = \frac{1}{2\pi i}\oint \l^{\a}d_{\a}
\]
which satisfy $Q^2=0$ due to the pure spinor condition \eqref{espinor_puro}
and the OPE \eqref{dd_ope}. Therefore we can define the 
unintegrated and integrated massless vertex operators for the super-Yang-Mills 
states
as follows
\be
\label{V}
V = \l^{\a}A_{\a}(x,\t)
\ee
\be
\label{integrado}
U(z) = \p\t^{\a}A_{\a}(x,\t) + A_m(x,\t)\Pi^m + d_{\a}W^{\a}(x,\t) 
+ \half N_{mn}{\cal F}^{mn}(x,\t),
\ee
where the superfields $A_{\a}$, $A_m$, $W^{\a}$ and ${\cal F}_{mn}$
describe the super-Yang-Mills theory in D=10, which is briefly reviewed
in appendix \ref{sym_ap}.

In the RNS formalism the unintegrated vertex operator satisfies $QU = \p V$,
as one can check by recalling that $U = \{\oint b, V\}$ and $T = \{Q,b\}$.
The proof then follows from the Jacobi identity
\be
\label{qu}
QU = [Q,\{\oint b, V\}] = - [ V, \{Q, \oint b\}] - [\oint b, \{V,Q\}]
= \p V
\ee
because the cohomology condition requires $\{V, Q\}=0$ and the conformal
weight zero of $V$ implies $[\oint T, V] = \p V$.

In the pure spinor formalism the integrated vertex \eqref{integrado} also 
satisfies \eqref{qu}.
To see this we use the OPE's
\eqref{pteta}, \eqref{dd_ope} and \eqref{super_ope} and the equations of
motion for the SYM superfields listed in Appendix \ref{sym_ap} to
get 
\[
Q(\p\t^{\a} A_{\a}) = (\p \l^{\a})A_{\a} - \p\t^{\a}\l^{\b}D_{\b}A_{\a}
\]
\[
Q(\Pi^m A_m) = (\l\g^m\p \t) A_m + \Pi^m\l^{\a}(D_{\a} A_m)
\]
\[
Q(d_{\a}W^{\a}) = - (\l\g^m W)\Pi_m - d_{\b}\l^{\a}D_{\a}W^{\b}
\]
\[
Q(\frac{1}{2}N_{mn}{\cal F}^{mn}) = \frac{1}{4}(\g^{mn}\l)^{\a}d_{\a}{\cal F}_{mn}
+ \frac{1}{2} N_{mn}\l^{\a}D_{\a}{\cal F}^{mn}
\]
Therefore
\[
QU = (\p\l^{\a})A_{\a} - \p\t^{\b}\l^{\a}(D_{\a}A_{\b} - \g^m_{\a\b}A_m)
+ \l^{\a} \Pi^m(D_{\a}A_m - (\g_m W)_{\a})
\]
\be
\label{qu2}
-\l^{\a}d_{\b}(D_{\a}W^{\b}+\frac{1}{4}(\g^{mn})_{\a}^{\phantom{m}\b}{\cal F}_{mn})
+N_{mn}(\l\g^n \p^m W).
\ee
Using the equations of motion listed in Appendix \ref{sym_ap} we get
\be
\label{qu3}
QU = (\p\l^{\a})A_{\a} + \l^{\a}\p\t^{\b}D_{\b}A_{\a} + \l^{\a}\Pi^m\p_m A_{\a}
\ee
where the last term in \eqref{qu2} vanished by the pure spinor condition
$(\l\g^n)_{\a}(\l\g_n)_{\b} = 0$ and the equation of motion $\g^m_{\a\b}\p_mW^{\b} = 0$,
\[
N_{mn}(\l\g^n \p^m W) = \frac{1}{2}(w\g^m\g^n\l)(\l\g^n \p^m W) - (w\l)(\l\g^m\p_m W) = 0.
\]
Using the definitions \eqref{pi} and \eqref{susy_deriv} one easily checks that
\eqref{qu3} becomes
\[
QU = (\p\l^{\a})A_{\a} + \l^{\a}( \p\t^{\b}\p_{\b} A_{\a} + \p X^m \p_m A_{\a}) 
\]
\[
= (\p\l^{\a})A_{\a} + \l^{\a} \p A_{\a} = \p (\l A) = \p V,
\]
as we wanted to show.

The unintegrated vertex operator satisfies $QV = 0$ if the superfield $A_\a$
is on-shell, \emph{i.e.,} if equation \eqref{symeq} is obeyed,
\[
QV = \oint \l^\a(z) d_\a(z) \l^\b(w)A_{\b}(x,\t) = \l^\a\l^\b D_\a A_\b = 0, 
\]
where we used that $\l^\a\l^\b = (1/3840)(\l\g^{mnpqr}\l)\g_{mnpqr}^{\a\b}$
for pure spinors $\l^\a$.

%****************************************************
\section{Tree-level prescription}
%****************************************************

The prescription to compute N-point superstring amplitudes at tree-level 
is given by
\be
\label{tree_presc}
{\cal A} = \langle V^1 V^2 V^3 \int U^4 {\ldots} \int U^N
\rangle
\ee
where the angle brackets is defined in such a way as to be non-vanishing
only when there are three pure spinor $\l$'s and five $\t$'s in a
combination proportional to 
\be
\label{measure}
\langle (\l\g^m \t)(\l\g^n \t)(\l\g^p \t)(\t\g_{mnp}\t)
\rangle = 1.
\ee
One can check that the measure \eqref{measure} is in the cohomology
of the pure spinor BRST operator \eqref{Q_brst}. It is BRST-closed due to
the pure spinor constraint \eqref{espinor_puro}. And it is not
BRST-trivial because there is no Lorentz scalar built out of
two $\l$'s and six $\t$'s. To check this one uses the theory of
group representations as follows. 

The representation of two pure spinors $\l^{\a}$ is
given by $[0,0,0,0,2]$ while six antisymmetric thetas are represented
by $[0,1,0,2,0]+[2,0,1,0,0]$. Therefore\footnote{I acknowledge the use of 
LiE in doing these computations \cite{LiE}.}
\[
[0,0,0,0,2]\otimes \Big[ [0,1,0,2,0]+[2,0,1,0,0]\Big] =
1[0,0,0,1,1] +1[0,0,0,2,2] +2[0,0,1,2,0] + {\ldots} 
\]
has no scalar component. 

The pure spinor measure \eqref{measure} together with BRST-closedness
of the vertex operators imply that the amplitude prescription is
supersymmetric. To see this one notes that the only possibility
of getting a non-vanishing result after a supersymmetry transformation
$\d \t^\a = \e^\a$ is if the amplitude of \eqref{tree_presc} contains
the term
\be
\label{notinv}
{\cal A} = \langle (\l\g^m \t)(\l\g^n \t)(\l\g^p \t)(\t\g_{mnp}\t)
(\t^\a \Phi_\a + {\ldots} )
\rangle
\ee
for some $\Phi_\a$. If that were true then the supersymmetry variation $\d_S {\cal A}$
would be
\be
\label{susyt}
\d_S {\cal A} = \int dz_4 \dotsi \int dz_N \e^\a \Phi_\a.
\ee
But note that the result of the amplitude calculation of \eqref{tree_presc}
\[
\int dz_4 \dotsi \int dz_N \l^\a\l^\b\l^\g f_{\a\b\g}(\t)
\]
must satisfy the BRST-closedness property of
\be
\label{closed}
\int dz_4 \dotsi \int dz_N \l^\a\l^\b\l^\g \l^\d D_\d f_{\a\b\g}(\t) = 0.
\ee
Plugging \eqref{notinv} into \eqref{closed} we conclude that
\[
\int dz_4 \dotsi \int dz_N \l^\a\l^\b\l^\g \l^\d \Phi_\d = 0,
\]
which is only possible if $\Phi_\d$ is a total derivative, implying
that the supersymmetry variation of \eqref{susyt} 
vanishes, $\d_S {\cal A} = 0$.

%******************************************************
\section{Multiloop prescription}
\label{multi_min}
%******************************************************

The prescription to compute multiloop amplitudes in the minimal
pure spinor formalism was spelled out in \cite{nathan_multiloop}, which
we now briefly review.

The multiloop prescription in the pure spinor formalism was made possible by the construction of
the analogous operators of the picture changing operators in the RNS formalism, which
can be understood as being necessary to absorb the zero-modes of the various
variables. As it is well-known \cite{verlinde}, the zero-modes of bosonic
variables require the introduction of delta functions
which depend on the variable which has the zero mode. The fermionic zero modes
require the insertion of as much fermionic variables as is the number of
zero modes, otherwise the Berezin integration will produce a vanishing
result.

So the analysis of zero modes will play a crucial r\^ole in the multiloop
prescription. But for our purposes in this thesis it will be sufficient to
know that a conformal weight one variable $\Phi_1$ has $g$ zero-modes in
a genus $g$ Riemann surface, while a conformal weight zero variable $\Phi_0$
always has one zero mode in every genus.

In the pure spinor formalism the zero modes of $\l^{\a}$, $N^{mn}$ and $J$
will require insertions of delta functions involving these variables. They
are given as follows
\be
\label{operators}
Y_C = C_{\a}\t^{\a}\d(C_{\b}\l^{\b}), \quad
Z_B = \frac{1}{2} B_{mn} (\l\g^{mn}d) \d(B^{pq}N_{pq}),\quad Z_J=
(\l^\a d_\a) \d(J),
\ee
where $C_{\a}$ and $B_{mn}$ are constant tensors. They will be responsible
for killing the eleven zero modes of $\l^{\a}$ and $11g$ zero modes of
$w_{\a}$. Therefore after eliminating the conformal weight one variables
through their OPE's one will be left with an expression containing 
only the zero modes of all the variables which are part of the pure spinor
formalism. Those zero modes will be absorbed by the insertions of the
operators \eqref{operators}, but one will need explicit measures to
integrate what is left.

The measure for integration over the eleven $\l$ zero-modes is 
given by
\[
[d\l] \l^{\a}\l^{\b}\l^{\g} = \e_{\rho_1{\ldots} \rho_{11}\k_1{\ldots} \k_5}
T^{((\a\b\g))[\k_1\k_2\k_3\k_4\k_5]} d\l^{\rho_1}{\ldots}
d\l^{\rho_{11}}
\]
while for the $w_{\a}$ zero modes it reads
\[
(d^{11}N)^{[[m_1 n_1][m_2 n_2]...[m_{10}n_{10}]]} =
[dN]
\]
\[
\Big[ (\l\g^{m_1 n_1 m_2 m_3 m_4}\l)
(\l\g^{m_5 n_5 n_2 m_6 m_7}\l)
(\l\g^{m_8 n_8 n_3 n_6 m_9}\l)
(\l\g^{m_{10} n_{10} n_4 n_7 n_9}\l) + {\rm permutations}\Big]
\]
where 
\[
(d^{11}N)^{[[m_1 n_1][m_2 n_2]...[m_{10}n_{10}]]}\equiv
dN^{[m_1 n_1]} \wedge
dN^{[m_2 n_2]} \wedge ... \wedge dN^{[m_{10} n_{10}]} \wedge dJ.
\]
and 
\[
T^{((\a\b\g))[\k_1\k_2\k_3\k_4\k_5]} = (\g_m)^{\k[\k_1}(\g_n)^{\s |\k_2}
(\g_p)^{\tau\k_3}(\g^{mnp})^{\k_4\k_5]}
(\d^{(\a}_{\k}\d^\b_\s\d^{\g)}_{\tau}
-\frac{1}{40}\g_q^{(\a\b}\d^{\g)}_\k \g^q_{\s\tau}).
\]

To compute multiloop amplitudes over a g-genus Riemann surface
one needs to have a measure
for the integration over the moduli space of Riemann surfaces.
The standard way to achieve this is through the insertion of
$3g-3$ factors containing the b-ghost and the Beltrami differential, 
which is a conformal weight $(-1,1)$ differential defined by
\[
\mu_{z}^{\phantom{m}{\bar z}} = g^{z{\bar z}}\frac{\p g_{zz}}{\p \tau}.
\]
That insertion has the property of being a density for the moduli
integration, because the Beltrami differential transforms as
\[
\mu_{z}^{\phantom{m}{\bar z}} = {\tilde \mu}_{z}^{\phantom{m}{\bar z}}
\frac{\p{\tilde \tau}}{\p \tau}.
\]
Explicitly the b-ghost insertion reads
\[
\langle b \cdot \mu \rangle = \int d^2z b_{zz}
\mu^{z}_{\phantom{m}{\bar z}}
\]
However the b-ghost must satisfy the property of $\{Q,b(z)\}= T(z)$ 
because $\langle b \cdot \mu\rangle$ must be BRST-invariant after
the integration over moduli space. But
in the pure spinor formalism there is no such object, because there
is no gauge invariant operator with ghost number  $-1$
(with respect to $J=(\l w)$).

The idea to overcome this difficulty was to construct an operator
$b(u,z)$ such that 
\[
\{Q, b_B(u,z)\} = T(u)Z_B(z)
\]
because whenever one needs to insert the $3g-3$ 
b-ghosts in the scattering
amplitude prescription one also needs to insert $10g$ of
$Z_B$ and $1g$ of $Z_J$ to deal with the zero modes of $w_{\a}$. Then
the idea was to borrow $3g-3$ $Z_B$'s into the factor containing 
the measure for the moduli space. Therefore the insertion of
$\langle b_B\cdot \mu \rangle$ in the pure spinor amplitude
prescription will respect its BRST-closedness property up to
a total derivative in moduli space.

The multiloop amplitude prescription for genus higher than one 
is given by
\[
A = \int d^2\tau_1 {\ldots} d^2\tau_{3g-3} \langle
\prod_{P=1}^{3g-3}\int d^2 u_P \mu_P(u_P)
{\tilde b_{B_P}}(u_P,z_P) 
\]
\[
\prod_{P=3g-2}^{10g} Z_{B_P}(z_P) \prod_{R=1}^{g} Z_J(v_R)
\prod_{I=1}^{11} Y_{C_I}(y_I)~|^2 ~\prod_{T=1}^N \int d^2 t_T U_T(t_T)
\rangle,
\]
where the $b_B$-ghost is a complicated operator whose expression
can be looked in \cite{nathan_multiloop} (see also detailed computations
in \cite{oda_Y}\cite{oda_b_min}).
For the genus one surface the prescription is given by
\[
A_{\rm one-loop}
\int d^2\tau \langle ~|~
\int d^2 u \mu(u)
\widetilde b_{B_1}(u,z_1) 
\]
\[
\prod_{P=2}^{10} Z_{B_P}(z_P) Z_J(v)
\prod_{I=1}^{11} Y_{C_I}(y_I)~|^2 ~V_1(t_1)
\prod_{T=2}^N \int d^2 t_T U_T(t_T)~\rangle,
\]
where due to translational invariance of the torus one 
can fix the position of one unintegrated vertex operator $V_1$.

The $\langle \quad \rangle$ brackets means the integration
over the zero modes of the various variables using the
measures described above together with the Berezin 
integrals over $\int d^{16}\t$
and $\int d^{16} d$.

%************************************************
%************************************************
\section{The non-minimal pure spinor formalism}
%************************************************
%************************************************

In the year 2005 a modification of the pure spinor formalism
was proposed in \cite{NMPS} which features the addition of the
left-moving non-minimal variables
$(r_{\a}, s^{\b})$ and $(\lb_{\a}, {\bar w}^{\a})$.
The action is given by 
\be
\label{nmps_action}
S_{\rm NMPS} = \frac{1}{2\pi}\int d^2 z (\half \p x^m \bar\p x_m +p_\a \bar\p\t^\a
- w_\a\bar\p \l^\a - \bar w^\a \bar\p \bar\l_\a + s^\a \bar\p r_\a )
\ee
where ${\bar w}^\a$ and $s^\a$ are the conformal weight one conjugate momenta of
the bosonic pure spinor $\lb_{\a}$ and the fermionic spinor 
$r_{\a}$ which satisfies
\[
(\lb \g^m r) = 0.
\]
Their OPE's are given by
\[
\lb_\a(z) {\bar w}^\b(y) \rightarrow \frac{\d^\b_\a}{z-y}, \quad
s^\a(z) r_\b(w) \rightarrow \frac{\d_\b^\a}{z-w}.
\]
Analogously to the minimal pure spinor formalism variables where
\[
N_{mn}= \frac{1}{2} (w\g_{mn}\l),
\quad J_\l=w_\a\l^\a,
\quad T_\l = w_\a\p\l^\a,
\]
the new variables also have their associated Lorentz and ghost currents,
\[
\bar N_{mn}= \half (\bar w\g_{mn}\lb - s\g_{mn} r),\quad
\bar J_{\bar \l} =\bar w^\a \bar\l_\a -s^\a r_\a,\quad 
T_\lb = {\bar w}^\a\p\lb_\a - s^\a \p r_\a,
\]
Furthermore one also defines 
\[
S_{mn} = \half (s \g_{mn} \lb),\quad  S=  s^\a \lb_\a, \quad J_r = (r s).
\]
and the total ghost current to be\footnote{There is a typo in equations (3.14)
and (3.15) of \cite{NMPS}, where $\l^\a r_\a$ was written as $\lb_\a r^\a$.}
\be
\label{nmps_J}
J = w_\a\l^\a - s^\a r_\a - \frac{2}{(\l\lb)}\big[(\lb\p\l) + (r \p\t)\big]
+ \frac{2}{(\l\lb)^2}(\l r)(\lb \p\t),
\ee
which is BRST equivalent to 
\[
J_b = J_\l - {\bar J}_\lb + J_r = w_\a \l^\a - {\bar w}^\a\lb_\a.
\]
The non-minimal BRST operator is defined by
\be
\label{nmps_BRST}
Q = \int dz (\l^\a d_\a + {\bar w}^{\a}r_\a).
\ee
Using the Kugo-Ojima (KO) quartet mechanism \cite{kugo}\cite{rybkin}
one can show that the cohomology of the non-minimal BRST operator 
\eqref{nmps_BRST} doesn't depend on the ``quartet'' of 
non-minimal variables $(r_\a, s^\a), (\lb_\a, \wb^\a)$.
That will allow us to choose a gauge were the external vertex
operators are independent of the non-minimal variables, so that
the same vertices as in the minimal pure spinor formalism can be
used.

Furthermore, due to the existence of the 
pure spinor field $\lb_\a$ it is possible to
construct a b-ghost satisfying $\{Q, b(z)\} = T(z)$, where (see also the discussion using
the Y-formalism \cite{oda_b})
\[
b = 
s^\a\p\lb_\a  + \frac{1}{4(\lb\l)}\Big[(2\Pi^m (\lb\g_m d)-  N_{mn}(\lb\g^{mn}\p\t)
- J_\l (\lb \p\t) -  (\lb\p^2\t)\Big]
\]
\be
\label{nmps_bghost}
+ {{(\lb\g^{mnp} r)(d\g_{mnp} d +24 N_{mn}\Pi_p)}\over{192(\lb\l)^2}} 
-
{{(r\g_{mnp} r)(\lb\g^m d)N^{np}}\over{16(\lb\l)^3}} +
{{(r\g_{mnp} r)(\lb\g^{pqr} r) N^{mn} N_{qr}}\over{128(\lb\l)^4}} 
\ee
and the total energy momentum tensor is given by
\be
\label{nmps_EM}
T(z) = -\half \p x^m \p x_m - (p \p\t) + (w \p\l)
+ ({\bar w} \p\lb) - (s \p r).
\ee

Now the key aspect of this non-minimal construction follows
from the observation that the operators 
$T(z), G^+(z)= 2j_{\rm BRST}, G^-(z) = b$ and $J(z)$ satisfy the twisted
${\hat c} = 3$ $N= 2$ algebra
\[
T(z)T(w) \rightarrow \frac{c/2}{(z-w)^4} + \frac{2T}{(z-w)^2} + \frac{\p T}{(z-w)}
\]
\[
T(z)G^\pm \rightarrow \frac{3}{2}\frac{G^\pm}{(z-w)^2}+ \frac{\p G^\pm}{(z-w)}
\]
\[
G^+(z)G^-(w) \rightarrow \frac{2c/3}{(z-w)^3} + \frac{2J}{(z-w)^2} + \frac{T}{(z-w)}
\]
\[
T(z)J(w) \rightarrow \frac{{\hat c}}{(z-w)^3} + \frac{J}{(z-w)^2} + \frac{\p J}{(z-w)}
\]
\[
J(z)G^\pm(w) \rightarrow \pm \frac{G^\pm}{(z-w)}
\]
\[
J(z)J(w) \rightarrow \frac{c/3}{(z-w)^2}.
\]
In particular we note that the anomaly of ${\hat c} = +3$ in the ghost current
of \eqref{nmps_J} is the same as the anomaly of $J= -bc$ in bosonic string theory.
The anomaly of $+3$ in the ghost current implies the non-conservation
of $3g-3$ units of charge in a genus g Riemann surface, via the Riemann-Roch
theorem. That is the same
as the number of moduli parameters of the surface. It is this equality
that allows one to use topological string methods in the computation 
of superstring scattering amplitudes in 
the non-minimal pure spinor formalism (see for example \cite{BCOV}).

%************************************************
\section{The scattering amplitude prescription}
%************************************************

We will now briefly review how scattering amplitudes are to be
computed using the non-minimal pure spinor formalism.

\subsection{Tree-level prescription}

N-point tree-level scattering amplitudes are computed by a correlation
function with three unintegrated vertices \eqref{V} and $N-3$ integrated
vertices \eqref{integrado},
\be
\label{nmps_tree}
{\cal A} = \langle {\cal N}
V^1 V^2 V^3 \int U^4 \dotsi \int U^N
\rangle.
\ee
The computation of \eqref{nmps_tree} proceeds as usual in a CFT.
First one integrates out the conformal weight one variables through
their OPE's to get an expression containing only zero modes for 
$\l$'s and $\t$'s,
\[
{\cal A} = \int [d\l][d\lb][dr]d^{16}\t {\cal N} \l^\a\l^\b\l^\g f_{\a\b
\g}(\t).
\]
The measures $[d\l], [d\lb]$ and $[dr]$ are given by
\be
\label{dl}
[d\l] \l^{\a}\l^{\b}\l^{\g} = \e_{\rho_1{\ldots} \rho_{11}\k_1{\ldots} \k_5}
T^{((\a\b\g))[\k_1\k_2\k_3\k_4\k_5]} d\l^{\rho_1}{\ldots}
d\l^{\rho_{11}}
\ee
\be
\label{dlb}
[d\lb]\lb_\a\lb_b\lb_\g = \e^{\a_1{\ldots} \a_{11}\k_1{\ldots} \k_5}
T_{((\a\b\g))[\k_1\k_2\k_3\k_4\k_5]}\,
d\lb_{\a_1}\dotsi d\lb_{\a_{11}}
\ee
\be
\label{dr}
[dr] = \e_{\a_1{\ldots} \a_{11}\k_1{\ldots} \k_5}
T^{((\a\b\g))[\k_1\k_2\k_3\k_4\k_5]}\lb_\a\lb_\b\lb_\g\,
\p^{\a_1}_r \dotsi \p^{\a_{11}}_r
\ee

This is almost the same recipe as in the minimal
formalism, the difference is the insertion of a regularization factor 
${\cal N}$, where
\[
{\cal N} = \exp(\{Q,\chi\}) = {\rm e}^{-(\l\lb) - (r\t)} \quad \text{for}
\quad \chi = -(\lb\t).
\]
The purpose of the regularization factor is due to the fact that 
the integration over $\l$ and $\lb$ may diverge because they
are non-compact. However, as ${\cal N} = 1 + Q\Omega$ the integral
will be independent of the choice for the regularization.

Using the measures \eqref{dl} -- \eqref{dr} one can show that
\[
{\cal A} = \int [d\l][d\lb][dr]d^{16}\t {\cal N} \l^\a\l^\b\l^\g f_{\a\b
\g}(\t) = \langle \l^\a\l^\b\l^\g f_{\a\b\g}(\t) \rangle
\]
and therefore the non-minimal prescription for tree-level amplitudes
is equivalent to the minimal pure spinor formalism.

%***********************************************
\subsection{Multiloop prescription}
\label{multiloop_NMPS}
%***********************************************

The prescription to compute $g-$loop amplitudes is given by
\be
\label{multi_presc}
{\cal A} = \int d^{3g-3}\tau \langle {\cal N}(y)
\prod_{i=1}^{3g-3}(\int dw_i \mu_i(w_j) b(w_j)) \prod_{j=1}^N 
\int dz_j U(z_j) \rangle
\ee
where $U(z)$ is the same integrated vertex operator
of \eqref{integrado} and the b-ghost is given by
\eqref{nmps_bghost}. After the integration of non-zero
modes appearing in the correlator \eqref{multi_presc} one
is left with the problem of how to integrate over the
$g-$zero modes of the conformal weight one variables
\[
N_{mn}(z),\hspace{.2cm} {\bar N}_{mn}(z),\hspace{.2cm} J_\l(z),\hspace{.2cm}
J_\lb(z),\hspace{.2cm}d_\a(z),\hspace{.2cm} S_{mn}(z)\hspace{.15cm}\text{and}
\hspace{.15cm} S(z).
\]
and also the zero modes of the conformal weight zero variables
$\l^\a, \lb_\a$ and $r_\a$. In general, a conformal weight +1
field $\Phi_1$ is written in a genus g Riemann surface as follows
\[
\Phi_1(z) = {\hat\Phi_1}(z) + \sum_{I=1}^g \Phi_1^I w_I(z)
\]
where $w_I(z)$ are the holomorphic one-forms and ${\hat\Phi_1}(z)$ has no
zero-mode. They satisfy
\[
\int_{a_I} w_J = \d_{IJ}, \quad \int_{a_I} dz {\hat\Phi_1}(z) = 0\quad\,
\forall I=1,{\ldots},g.
\]
Therefore one can show that, for example
\[
w^I_\a = \int_{a_I} dz w(z)_\a,
\]
and this notation will be used in the following discussion.
The integration over the zero modes
of the pure spinor fields and of $r_\a$ is performed with the
measures \eqref{dl} -- \eqref{dr} described above, while the
other zero modes are integrated with the measures defined by
\begin{align}
\label{dw}
[dw] &= (\l\g^m)_{\k_1}(\l\g^n)_{\k_2}(\l\g^p)_{\k_3}(\g_{mnp})_{\k_4\k_5}
\e^{\k_1{\ldots} \k_5\rho_1{\ldots} \rho_{11}} dw_{\rho_1} \dotsi dw_{\rho_{11}},\\
\label{dwb}
[d\wb] & = 
(\lb\g^m)^{\k_1}(\lb\g^n)^{\k_2}(\lb\g^p)^{\k_3} (\g_{mnp})^{\k_4\k_5}\,
\e_{\k_1{\ldots} \k_5\a_1{\ldots} \a_{11}} d\wb^{\a_1}\dotsi d\wb^{\a_{11}}\\
\label{ds}
[ds] &= (\l\lb)^{-3}(\l\g^m)_{\k_1}(\l\g^n)_{\k_2}(\l\g^p)_{\k_3}(\g_{mnp})_{\k_4\k_5}
\e^{\k_1{\ldots} \k_5\rho_1{\ldots} \rho_{11}} \p^s_{\rho_1} \dotsi \p^s_{\rho_{11}},
\end{align}
Note that the measure \eqref{dw}
is gauge invariant under $\d w_{\a} = (\l\g^m)_{\a} \Omega_m$ because
\[
(d\l\g^q)_{[\d_1}(\l\g^m)_{\k_1}(\l\g^n)_{\k_2}(\l\g^p)_{\k_3}
(\g_{mnp})_{\k_4\k_5]} =0,
\]
which comes from the fact that there is no vector representation in 
the decomposition of
$\l^4\t^6$ (here the $\t^6$ factor is to emulate the antisymmetry 
over the spinor indices).
To define the regularization factor we use $\chi = - (\lb\t)- (w^Is^I)$
to obtain
\be
\label{reg_loop}
{\cal N}(y) = \exp\big[-(\l\lb)-(r\t) - (w^I\wb^I) + (s^Id^I)\big].
\ee
Note that here we are using a different (non gauge invariant) $\chi$ from what 
was originally defined
in \cite{NMPS}. However the non gauge invariance of \eqref{reg_loop} should
not affect the amplitudes because ${\cal N} - 1$ continues to be BRST-trivial
even if it is not gauge invariant.

Therefore the evaluation of \eqref{multi_presc} will give rise to an expression
of the form
\[
{\cal A} = \int [d\l][d\lb][dr]\prod_{I=1}^g
[dw^I][d\wb^I][ds^I](d^{16}d^I)d^{16}\t 
\,{\cal N} f(\t).
\]
From the measures \eqref{dl} -- \eqref{dr} and \eqref{dw} -- \eqref{ds}
one can deduce the following behaviour as $(\l\lb) \rightarrow 0$
\be
\label{div}
\int [d\l][d\lb][dr] \prod_{I=1}^g
[dw^I][d\wb^I][ds^I](d^{16}d^I)d^{16}\t 
\,{\cal N} \rightarrow \l^{8+3g} \lb^{11},
\ee
therefore $f(\l,\lb,r,\t)$ must diverge slower than $\l^{-8-3g}\lb^{-11}$
as $(\l\lb)\rightarrow 0$ so that \eqref{div} is guaranteed not to
diverge. Since each b-ghost diverges as $\l^{-4}\lb^{-3}$ the maximum
number of loops in which this regularization can be safely used is
$g=2$, where $f$ could diverge as $\l^{-14}\lb^{-11}$ but whose
$3g-3$ b-ghosts makes it diverge at most like $\l^{-12}\lb^{-9}$.
There exists a regularization prescription which can in principle
be used to go beyond $g=2$, but so far no concrete computations
were ever done with it \cite{nekrasov_reg}.

We will see in the next chapter that in fact this multiloop prescription
was successfully used to compute massless four-point amplitudes up to
two-loops \cite{nmps_two}.

Due to the fact that the external vertices don't depend on the non-minimal
variables and that the $r_{\a}$'s appearing in the b-ghost can be 
substituted by $D_\a$, we can easily guess the result of the integrations over
$[dw], [d\wb]$ and $[ds]$. That will enable us to easily obtain the kinematic
factors at one-loop, for example.

Note that at one-loop there are eleven zero-modes of $s^\a$,
which can only\footnote{The term $s^\a\p\lb_\a$ of the b-ghost 
does not contribute because there is no $\wb^\a$ in the external vertices
to kill the non zero-modes of $\p\lb_\a$.}
come from the term $(s d)$ in the regularizator ${\cal N}$.
Therefore the remaining five $d_\a$ zero modes must come from the
b-ghost and the external vertices. Therefore by ghost number conservation
we obtain
\[
\int d^{16}d[dw][d\wb][ds] \exp\big[-(w\wb) + (sd) - (\l\lb) - (r\t)\big]
d_{\k_1}\dotsi d_{\k_5} f^{\k_1{\ldots} \k_5}(r_\a,\t) =
\]
\[
= (\l^3)_{[\k_1\k_2\k_3\k_4 \k_5]} f^{\k_1{\ldots} \k_5}(D_\a,\t)
\]
where $(\l^3)_{[\k_1\k_2\k_3\k_4 \k_5]}$ is some tensor with five antisymmetric
free indices containing three pure spinors. The unique such tensor is given
by
\be
\label{inter}
(\l\g_m)_{\k_1}(\l\g_n)_{\k_2}(\l\g_p)_{\k_3}(\g^{mnp})_{\k_4\k_5}.
\ee
Thus we can see that the effect of evaluating the pure spinor
measures is to substitute five $d_\a$'s from the b-ghost and the
external vertices by \eqref{inter}. Explicitly,
\be
\label{truque}
d_{\k_1}d_{\k_2}d_{\k_3}d_{\k_4}d_{\k_5} \rightarrow 
(\l\g_m)_{\k_1}(\l\g_n)_{\k_2}(\l\g_p)_{\k_3}(\g^{mnp})_{\k_4\k_5}.
\ee
It is interesting to note that the right hand side of \eqref{truque} 
already is completely antisymmetric in
$[\k_1{\ldots} \k_5]$ because of the pure spinor condition. To see this
one notices that the only
non-obvious antisymmetry to check 
is over the exchange of the indices $\k_1$ and $\k_4$,
for example. However, as $(\l\g^p)_\a(\l\g_p)_\b=0$ we can write 
$(\g^{mnp})_{\k_4\k_5}=\g^m_{\k_4\s}(\g^n\g^p)_{\phantom{m}\k_5}^{\s}$
and use the gamma matrix identity $\eta_{rs}\g^r_{\a(\b}\g^s_{\g\d)} = 0$
to obtain
\[
\l^\a(\g_m)_{\a\k_1}(\l\g_n)_{\k_2}
(\l\g_p)_{\k_3}\g^m_{\k_4\s}(\g^n\g^p)_{\phantom{m}\k_5}^{\s}
= - (\l\g^m\g^n\g^p)_{\k_5}(\l\g_n)_{\k_2}
(\l\g_p)_{\k_3}(\g_m)_{\k_4\k_5} 
\]
\[ 
- (\l\g_m)_{\k_4}(\l\g_n)_{\k_2}(\l\g_p)_{\k_3}(\g^{mnp})_{\k_1\k_5}.
\]
The proof follows from the vanishing of the first term of the right
hand side due to the pure spinor condition.

%*****************************************************
\chapter{Computing Pure Spinor Scattering Amplitudes}
%*****************************************************

\section{Massless three-point amplitude at tree-level }

The tree-level scattering amplitudes with up to four fermions 
were shown to be equivalent to RNS in \cite{nathan_brenno}.
As a brief illustration we will compute the scattering of three massless
particles at tree-level. This is the simplest example possible because
the prescription \eqref{tree_presc} implies that there are no 
integrated vertices and therefore there is no need to compute 
OPE's. Only the zero modes contribute to the amplitude and their
contribution is completely determined by the measure \eqref{measure}.

The amplitude to compute is given by
\be
\label{tree_3}
\mathcal{A} = \vev{V_1(z_1)V_2(z_2)V_3(z_3)} + (2\leftrightarrow 3),
\ee
where $V = \l^{\a} A_{\a}(\t){\rm e}^{ik\cdot X}$ and the theta expansion
of $A_{\a}(\t)$ is given in Appendix \ref{sym_ap}. The sum over the permutation
of labels 2 and 3 has to be done because a general 
M\"obius transformation does not change the cyclic ordering of the vertex
operators\footnote{That is because a non-cyclic transformation
always has a fixed point. For example, it is impossible to map 
$y_1y_2y_3$ into $y_1y_3y_2$
because the fixed point $y_1$ implies that the M\"obius transformation is
the identity,
\[
y_1 = \frac{1y_1+0}{0+1}.
\]}, so both orderings must be summed over.

The contribution from the exponential is proportional to a constant
because the particles are massless. Therefore $k^2_j = 0$ implies
$k_i\cdot k_j=0$ due to momentum conservation. The non trivial
part of the computation comes from the $\t$ zero modes. 

To compute the scattering of three gluons we use the $A_\a(\t)$ expansion
of appendix \ref{sym_ap} to get three
different possibilities to obtain five thetas, given by
\begin{center}
\begin{tabular}{|c|c|c|} \hline
$A^1_{\a}(\t)$ & $A^2_{\a}(\t)$ & $A^3_{\a}(\t)$ \\ \hline
1 & 1 & 3  \\ \hline
1 & 3 & 1  \\ \hline
3 & 1 & 1  \\ \hline
\end{tabular}
\end{center}
Explicitly we get, for one of the permutations of \eqref{tree_3},
\be
\label{3_pt}
\mathcal{A}_{\rm BBB} = -\frac{1}{64}\(
k^3_m e^1_{r}e^2_{s} e^3_{n} -
k^2_m e^1_{r}e^2_{n} e^3_{s} +
k^1_m e^1_{n}e^2_{r} e^3_{s}
\)
\vev{\(\lambda\g^r\t\)\(\lambda\g^s\t\)\(\lambda \g_p \t\)\(\t\gamma^{pmn}\t\)}.
\ee
As we will see in Appendix \ref{chap_t}, the above correlator is given by
\[
\langle (\l\g^r\t)(\l\g^s\t)(\l \g_p \t)(\t\g^{pmn}\t)\rangle 
= \frac{1}{120}\d^{rsp}_{pmn} = \frac{1}{45}\d^{rs}_{mn}.
\]
Then \eqref{3_pt} evaluates to
\be
\label{anti23}
\mathcal{A}_{\rm BBB} = - \frac{1}{2880}\Big[
  (e_1\cdot e_2)(e_3\cdot k_2)
+ (e_1\cdot e_3)(e_2\cdot k_1)
+ (e_2\cdot e_3)(e_1\cdot k_3)
\Big]
\ee
where we used momentum conservation and 
$e_i\cdot k_i = 0$. Note that \eqref{anti23} is antisymmetric in
$(2\leftrightarrow 3)$ and therefore the whole amplitude vanishes for photons, 
whereas
for gluons it is non-vanishing due to the Chan-Patton factors.
Up to an overall constant, this is the same
result as in the RNS formalism (see
equation
(7.4.30) of \cite{gswI}).

As \eqref{tree_3} is supersymmetric, the contribution of
fermionic states is as easy to compute as the bosonic case
considered above. For example, the $B^1F^2F^3$ scattering amplitude
is given by the following theta distribution
\begin{center}
\begin{tabular}{|c|c|c|} \hline
$A^1_{\a}(\t)$ & $A^2_{\a}(\t)$ & $A^3_{\a}(\t)$ \\ \hline
1 & 2 & 2  \\ \hline
\end{tabular}
\end{center}
which is computed to be
\[
{\cal A}_{\rm BFF} = -\frac{1}{288}e^1_{n_1}(\chi^2\g^r\chi^3)
\langle (\l\g^n1\t)(\l\g^m\t)(\l\g^p \t)(\t\g_{mrp}\t)\rangle
= \frac{1}{2880} e^1_m(\chi^2\g^m \chi^3),
\]
which again is non-vanishing after summing $(2\leftrightarrow 3)$ 
only for a non-abelian group.

%******************************************************
\section{Massless four-point amplitude at tree-level}
%******************************************************

It has been known for over eight years now that tree-level
amplitudes computed with the pure spinor formalism are equivalent
to their RNS counterparts \cite{nathan_brenno}. Nevertheless, apart
from the trivial massless tree-point amplitude, no other tree-level 
amplitude has been explicitly computed with the pure spinor formalism. 
When ones attention is directed towards pure spinor superspace
expressions for kinematic factors, the natural amplitude to
study is the scattering of four massless strings. In \cite{mafra_tree}
this task has been completed and the following pure spinor superspace expression
for the kinematic factor was obtained,
\[
K_0 = 2\langle (\p_m A_n) (\l A)\p^m(\l A)(\l \g^n W) \rangle
- \langle (\l A)\p^m(\l A)\p^n(\l A){\cal F}_{mn}\rangle.
\]
With this superspace representation for the kinematic
factor one can show through pure spinor manipulations that this
is in fact proportional to the kinematic factor for this same
amplitude, but at the one-loop level. That this could be shown
in a few pages is a remarkable display of the usefulness in
having kinematic factors written in pure spinor superspace.
We will now review the computation of \cite{mafra_tree}.

Following the tree-level prescription of \eqref{tree_presc}, the amplitude
to compute is 
\be
\label{amp}
{\cal A} = \langle V^1(z_1,{\bar z_1})V^2(z_2,{\bar z_2})V^3(z_3,{\bar z_3})
\int_{C}d^2 z_4 U(z_4,{\bar z}_4) 
\rangle.
\ee
The closed string vertices are given by the holomorphic square of
the open string vertices,
$V(z,{\bar z}) = {\rm e}^{ik\cdot X}\l^{\a}{\bar \l}^{\b}A_{\a}(\t)
{\bar A}_{\b}(\t)$ and $U(z,{\bar z}) = {\rm e}^{ik\cdot X}U(z){\bar U}({\bar z})$,
where the integrated vertex operator is given by \eqref{integrado}.

In the computation of \eqref{amp} we note that standard SL(2,C) 
invariance allows us to fix $z_1=0, z_2=1$ and $z_3=\infty$,
so the expectation value for the exponentials simplify, 
\[
\langle\prod_{i=1}^4 :{\rm e}^{ik^i\cdot X(z_i,{\bar z}_i)}:\rangle
= |z_4|^{-\half \a' t}|1-z_4|^{-\half \a' u} \equiv M(z_4,{\bar z}_4).
\]
Now we remove the conformal weight one operators of the integrated vertex \eqref{integrado}
in \eqref{amp} by
using their OPE's. The first term of \eqref{integrado} does not contribute because there is
no $p_{\a}$'s in the unintegrated vertices, 
while the second gives
\be
\label{triv}
\langle
A^4_m\Pi^m(z_4) \prod_{j=1}^4 :{\rm e}^{ik^j\cdot X(z_j,{\bar z}_j)}:
\rangle = \sum_{j=1}^3 {\a'\over 2}{ik^m_j \over z_j-z_4}
\langle 
(\l A^1)(\l A^2)(\l A^3)A^4_m
\rangle M(z_4,{\bar z}_4).
\ee
Using the standard OPE's
\be
\label{opes}
N^{mn}(z_4)\l^{\a}(z_j) =
{\a'\over 4}{(\l\g^{mn})^{\a} \over z_j-z_4}, \quad
d_{\a}(z_4)V(z_j) = -{\a'\over 2}{D_{\a}V \over z_j - z_4},
\ee
we obtain the following OPE identity:
$$
\langle
(\l A^1)(\l A^2)(\l A^3)\left(
d_{\a}(z_4)W_4^{\a} +\half N^{mn}(z_4){\cal F}_{mn}^4 \right)
\rangle = 
$$
\be
\label{pilita}
= {\a' \over 2(z_1-z_4)}\langle
A^1_m (\l A^2)(\l A^3)(\l \g^m W^4) 
\rangle - (1\leftrightarrow 2) + (1\leftrightarrow 3).
\ee
To show this, one uses \eqref{opes} to get
$$
\langle
(\l A^1)(z_1)(\l A^2)(z_2)(\l A^3)(z_3)d_{\a}(z_4)W_4^{\a}
\rangle = 
$$
$$
{\a'\over 2(z_1-z_4)}\langle
D_{\a}(\l A^1)(\l A^2)(\l A^3) W_4^{\a}
\rangle - (1\leftrightarrow 2) + (1\leftrightarrow 3).
$$
Concentrating for simplicity on the first term, the use of the
super-Yang-Mills identity 
$D_{\a}(\l A) = -(\l D)A_{\a} + (\l\g^m)_{\a}A_m$
allows the numerator to be rewritten as
\be
\label{tmp}
\langle
D_{\a}(\l A^1)(\l A^2)(\l A^3) W_4^{\a}
\rangle = - \langle
(\l D A^1_{\a})(\l A^2)(\l A^3) W_4^{\a}
\rangle + \langle
A^1_m(\l A^2)(\l A^3) (\l\g^m W^4)
\rangle.
\ee
As BRST-exact terms decouple, the first term in the right hand side of \eqref{tmp}
becomes
$$
-{\a'\over 2(z_1-z_4)} \langle
(\l D A^1_{\a})(\l A^2)(\l A^3) W_4^{\a}
\rangle = 
-{\a'\over 2(z_1-z_4)} \langle
A^1_{\a}(\l A^2)(\l A^3) (\l D)W_4^{\a}
\rangle 
$$
$$
= -{\a'\over 8(z_1-z_4)}\langle
(\l \g^{mn} A^1)(\l A^2)(\l A^3){\cal F}_{mn}^4
\rangle.
$$
However, this term is exactly canceled by the $(z_1-z_4)^{-1}$ contribution from
the OPE
$$
\half\langle
(\l A^1)(\l A^2)(\l A^3)
(N^{mn}{\cal F}^4_{mn})
\rangle = {\a'\over 8(z_1-z_4)} \langle
(\l \g^{mn} A^1)(\l A^2)(\l A^3){\cal F}_{mn}^4
\rangle +{\ldots} ,
$$
which finishes the proof of \eqref{pilita}.

With the results \eqref{triv} and \eqref{pilita}, the correlation in 
the amplitude \eqref{amp} reduces to 
$$
{\cal A} = \left({\a'\over 2} \right)^2 \int_{C}d^2z_4
\left(
{F_{12}\over z_4} + {F_{21}\over 1-z_4}
\right)\left(
{{\bar F}_{12}\over\zb_4} + {{\bar F}_{21}\over 1-\zb_4}
\right)|z_4|^{-\half \a' t}|1-z_4|^{-\half \a' u},
$$
where $F_{12} = ik_1^m\langle (\l A^1)(\l A^2)(\l A^3) A^4_m
\rangle + \langle A^1_m (\l A^2)(\l A^3)(\l\g^m W^4)
\rangle$ and $F_{21}$ is obtained by exchanging $1\leftrightarrow 2$.
The integral can be evaluated using the following formula 
$$
\int_{C} d^2z z^N(1-z)^M \zb^{\bar N}(1-\zb)^{\bar M} =
2\pi {\Gamma(1+N)\Gamma(1+M)  \over \Gamma(2+N+M) }
{\Gamma(-1-{\bar N}-{\bar M})  \over \Gamma(-{\bar N})\Gamma(-{\bar M}) }.
$$
After a few manipulations one finally gets
$$
{\cal A} = -2\pi ({\a'\over 2})^4 K_0{\bar K}_0
{\Gamma(-{\displaystyle \a' t\over 4})\Gamma(-{\displaystyle \a' u\over 4})
\Gamma(-{\displaystyle \a' s\over 4}) 
\over \Gamma(1+{\displaystyle \a' t\over 4})\Gamma(1+{\displaystyle \a' u\over 4})
\Gamma(1+{\displaystyle \a' s\over 4})},
$$
where $K_0 = \half(u F_{12} + t F_{21})$ is given by
\be
\label{kin}
K_0  = 
\langle  (\p_m A^1_n) (\l A^2)\p^m(\l A^3)(\l \g^n W^4)\rangle
- \half\langle \p^m(\l A^1)\p^n(\l A^2)(\l A^3){\cal F}^4_{mn} \rangle
+ (1\leftrightarrow 2),
\ee
which is the sought-for kinematic factor in pure spinor superspace.

Note that $K_0$ is BRST-closed because
\[
Q F_{12} = -\frac{t}{2}\langle (\l A^1)(\l A^2)(\l A^3)(\l A^4)\rangle,
\quad
Q F_{21} = +\frac{u}{2}\langle (\l A^1)(\l A^2)(\l A^3)(\l A^4)\rangle.
\]

When trying to relate the above tree-level kinematic factor with
its one-loop cousin, it is convenient to rewrite \eqref{kin} without
explicit labels,
\be
\label{this}
K_0 = 2\langle (\p_m A_n) (\l A)\p^m(\l A)(\l \g^n W) \rangle
- \langle (\l A)\p^m(\l A)\p^n(\l A){\cal F}_{mn}\rangle.
\ee

We postpone the explicit evaluation in components 
of \eqref{kin} to section \ref{complete_kin}. Before
that we will show how \eqref{kin} relates to amplitudes
at higher-loop orders.

%***************************************************
\section{Massless four-point amplitude at one-loop}
%***************************************************

We can compute the massless four-point amplitude at the one-loop order
with the two different pure spinor formalism prescriptions described
in sections \ref{multi_min} an \ref{multiloop_NMPS}.
It will be shown that they are equivalent up to a constant factor. Note
that it has been recently formally proved that these two prescriptions are 
equivalent \cite{joost}, and the results presented here can 
be regarded as an example of that.

If one is not interested in the overall coefficient, it also happens that
the kinematic factor is readily obtained by a zero-mode saturation argument,
avoiding the long procedure of functionally integrating using the measures
$[dr]$, $[ds]$, $[d\l]$ etc. This is the route taken in the 
papers \cite{nathan_multiloop}\cite{nmps_two} and
which will be described here.

%*******************************************************
\subsection{Minimal pure spinor computation}
\label{one_min}
%*******************************************************

Using the minimal pure spinor prescription the open superstring 
amplitude is given by
\[
{\cal A} = \int d\tau \langle
\int d w \mu(w)
\widetilde b_{B_1}(w,z_1) 
\prod_{P=2}^{10} Z_{B_P}(z_P) Z_J(v)
\prod_{I=1}^{11} Y_{C_I}(y_I) ~V_1(t_1)
\prod_{T=2}^N \int d^2 t_T U_T(t_T)~\rangle,
\]
and the kinematic factor in pure spinor superspace is obtained 
by considering how the sixteen zero modes for $d_{\a}$ can be saturated.

From \eqref{operators} we see that the nine $Z_B$ and one $Z_J$ will 
provide ten $d_\a$ zero modes. Since there is no term in the b-ghost which
contains three or five $d_\a$'s, the amplitude will be non-vanishing if the b-ghost
contributes with four $d$'s and 
the three integrated vertices provide
two $d$'s through the term $(d W)(d W)$. Furthermore,
as there is a delta function derivative of $N^{mn}$ coming from the 
b-ghost, the amplitude will be non-vanishing if one of the external vertices
provide an explicit $N^{mn}$, so that the analogous delta function property of
$\int dx \delta'(x) x = -1$ can be used. 

Looking at the integrated vertex
\eqref{integrado} we see that the term containing $N^{mn}$ has the superfield
${\cal F}_{mn}$, so 
we have shown the kinematic factor to be composed out of the following
superfields 
\be
\label{sf_oneloop}
(\l)^2(\l A)W^2{\cal F},
\ee 
where we already used the fact (which can be shown by integrating the
measures)
pure spinor
superspace expressions contain three pure spinors $\la$.
We are now required
to check how many different Lorentz invariant contractions can be constructed
out of these fields in \eqref{sf_oneloop}. If there is a unique contraction 
then we can shortcut the
functional integration procedure and immediately write down the answer in
pure spinor superspace. It is a happy fortuitous fact that this is the case here,
in deep contrast to the massless five-point amplitude of section \ref{crazy_5}.

Fortunately, it is easy to show there is a unique Lorentz-invariant
way to contract the indices in \eqref{sf_oneloop}. To show this, first choose
a Lorentz frame in which the only non-zero component of $\l^\a$
is in the $\l^+$ direction. This choice preserves a $U(1)\times SU(5)$
subgroup of $SO(10)$, under which a Weyl spinor $U^\a$ and an anti-Weyl
spinor $V_\a$ decompose as
\[
U^\a \longrightarrow \(U^+_{5\over 2}, 
U_{{1\over 2}[ab]}, U^a_{-{3\over 2}}\) ,\quad
V_{\a} \longrightarrow \(V_{-{5\over 2}+}, V^{[ab]}_{-{1\over 2}}, 
V_{+{3\over 2}a}\),
\]
where the subscript denotes the $U(1)$ charge.
So the unique way to cancel the $+15/2$ U(1)-charge
of the three $\la$'s is when the superfields contribution is
\[
K = \langle (\l^+)^3 A_+ W^a W^b F_{ab} \rangle
\]
which can be written in covariant SO(10) language as 
\be
\label{fimloop}
K = \langle (\l A)(\l\g^m W)(\l\g^n W){\cal F}_{mn} \rangle,
\ee
which is the final pure spinor superspace expression for this 
important amplitude. Now let's analyse it a bit.

%********************************************************
\subsubsection{Gauge invariance of the kinematic factor}

The appearance
of the explicit superfield $A_{\a}$ in the kinematic factor of \eqref{fimloop} 
might spoil the gauge invariance
of the amplitude, as it transforms as \eqref{gaugetr}
\be
\label{notgauge}
\d A_{\a} = D_{\a} \Omega.
\ee
However it is easy to check that using the properties of pure spinor
superspace and the equations of motion of the SYM superfields \eqref{eqW} and \eqref{eqF}, 
the kinematic
factor \eqref{fimloop} is indeed gauge invariant. This is because the
variation \eqref{notgauge} implies that the gauge transformation of the
unintegrated vertex operator
is BRST-exact $\delta (\l A) = \l^{\a} D_{\a} \Omega = Q_{\rm BRST} \Omega$,
which allows the BRST-charge to be ``integrated by parts'' using the property
that pure spinor superspace expressions of BRST-exact terms are zero. So the
gauge variation of \eqref{fimloop} is given by
\[
\delta K = \langle Q(\Omega)(\l\g^m W)(\l\g^n W){\cal F}_{mn} \rangle =
-\langle \Omega Q[ (\l\g^m W)(\l\g^n W){\cal F}_{mn} ] \rangle = 0
\]
where we used
\[
Q(\l\g^m W) = -\frac{1}{4}(\l\g^m \g^{rs} \l){\cal F}_{rs} = 0
\]
and
\[
(\l\g^m W)(\l\g^n W)Q{\cal F}_{mn} = 2(\l\g^m W)(\l\g^n W)\p_{[m}(\l\g_{n]}W) = 0,
\]
which can be shown using the equations of motion and the defining pure spinor
property of $(\l\g^m\l) = 0$. We have then shown that the massless four-point
amplitude at one-loop level is indeed gauge invariant.

%***********************************************************
\subsubsection{An equivalent pure spinor superspace expression}
%***********************************************************

If we use the SYM identity \eqref{cst}
\[
QA_m = (\l\g^m W) + \p^m (\l A)
\]
and the vanishing of BRST-exact terms in pure spinor superspace, we 
can rewrite \eqref{fimloop} in a manifestly gauge invariant way.
To see this substitute $(\l A)$ in $\langle (\l A)(\l\g^m W)(\l\g^n W){\cal F}_{mn}\rangle$ by
\be
\label{cool}
(\l A) = \frac{1}{(H\cdot k)} H^p QA_p - \frac{1}{(H\cdot k)}H_p (\l\g^p W)
\ee
where $H^p$ is an arbitrary vector such that $(H\cdot k) \neq 0$.
Note that the first term in the right hand side of \eqref{cool} will not 
contribute because 
\[
\langle (Q A_p) (\l\g^m W)(\l\g^n W){\cal F}_{mn}\rangle
\]
is BRST-exact. So we get
\be
\label{4equiv}
\langle (\l A^1)(\l\g^m W^2)(\l\g^n W^3){\cal F}^4_{mn}\rangle =
- \frac{1}{(H\cdot k^1)}H_p \langle (\l\g^p W^1)(\l\g^m W^2)(\l\g^n W^3){\cal F}^4_{mn}\rangle.
\ee

We can easily use the method of appendix \ref{chap_t} to evaluate the
purely bosonic part of the right hand side of \eqref{4equiv}. We obtain
the following table for the distribution of thetas,
\begin{center}
\begin{tabular}{|c|c|c|c|} \hline
$W_1^{\a}(\t)$ & $W_2^{\a}(\t)$ & $W_3^{\a}(\t)$ & $F^4_{mn}(\t)$ \\ \hline
1 & 1 & 1 & 2  \\ \hline
1 & 1 & 3 & 0  \\ \hline
1 & 3 & 1 & 0 \\ \hline
3 & 1 & 1 & 0 \\ \hline
\end{tabular}
\end{center}
which, using the superfields of appendix \ref{sym_ap}, expands to
\[
\frac{1}{256(H\cdot k^1)}F^1_{m_1n_1}F^2_{m_2n_2}F^3_{m_3n_3}F^4_{m_4n_4}
\Big[ 
\]
\[
+ \langle(\l\g^p\g^{m_1n_1}\t)(\l\g^m\g^{m_2n_2}\t)
(\l\g^{k_4}\g^{m_3n_3}\t)(\t\g_{[m}\g^{m_4n_4}\t)\rangle k^4_{n]}
\]
\[
+ \frac{1}{3}\langle(\l\g^p\g^{m_1n_1}\t)(\l\g^{[m_4|}\g^{m_2n_2}\t)
(\l\g^{|n_4]}\g^{k_3a}\t)(\t\g_{a}\g^{m_3n_3}\t)\rangle k^3_{k_3}
\]
\[
+ \frac{1}{3}\langle(\l\g^p\g^{m_1n_1}\t)(\l\g^{[m_4|}\g^{k_2a}\t)
(\l\g^{|n_4]}\g^{m_3n_3}\t)(\t\g_{a}\g^{m_2n_2}\t)\rangle k^2_{k_2}
\]
\[
+ \frac{1}{3}\langle(\l\g^p\g^{k_1a}\t)(\l\g^{[m_4|}\g^{m_2n_2}\t)
(\l\g^{|n_4]}\g^{m_3n_3}\t)(\t\g_{a}\g^{m_1n_1}\t)\rangle k^1_{k_1}\Big].
\]
After a long but straightforward calculation we obtain,
\[
= -\frac{1}{5760}\Big[ + \frac{1}{2}t^2 (e^1\cdot e^3)(e^2\cdot e^4)
        + \half tu (e^1\cdot e4)(e^2\cdot e^3)
        + \half tu (e^1\cdot e3)(e^2\cdot e^4)
\]
\[
        - \half tu (e^1\cdot e2)(e^3\cdot e^4)
	+ \half u^2 (e^1\cdot e^4)(e^2\cdot e^3)
\]
\[
       + t(k^4 \cdot e^2)(k^4 \cdot e^3)(e^1\cdot e^4)
       - t(k^4 \cdot e^1)(k^4 \cdot e^3)(e^2\cdot e^4)
       - t(k^3 \cdot e^4)(k^4 \cdot e^2)(e^1\cdot e^3)
\]
\[
       - t(k^3 \cdot e^1)(k^4 \cdot e^3)(e^2\cdot e^4)
       + t(k^3 \cdot e^1)(k^4 \cdot e^2)(e^3\cdot e^4)
       - t(k^2 \cdot e^4)(k^4 \cdot e^3)(e^1\cdot e^2)
\]
\[
       + t(k^2 \cdot e^4)(k^3 \cdot e^2)(e^1\cdot e^3)
       - t(k^2 \cdot e^4)(k^3 \cdot e^1)(e^2\cdot e^3)
       + t(k^2 \cdot e^3)(k^4 \cdot e^2)(e^1\cdot e^4)
\]
\[
       - t(k^2 \cdot e^3)(k^4 \cdot e^1)(e^2\cdot e^4)
       - t(k^2 \cdot e^3)(k^2 \cdot e^4)(e^1\cdot e^2)
       - u(k^3 \cdot e^4)(k^4 \cdot e^1)(e^2\cdot e^3)
\]
\[
       - u(k^3 \cdot e^2)(k^4 \cdot e^3)(e^1\cdot e^4)
       + u(k^3 \cdot e^2)(k^4 \cdot e^1)(e^3\cdot e^4)
       + u(k^3 \cdot e^2)(k^3 \cdot e^4)(e^1\cdot e^3)
\]
\[
       - u(k^3 \cdot e^1)(k^3 \cdot e^4)(e^2\cdot e^3)
       + u(k^2 \cdot e^4)(k^3 \cdot e^2)(e^1\cdot e^3)
       - u(k^2 \cdot e^4)(k^3 \cdot e^1)(e^2\cdot e^3)
\]
\[
       + u(k^2 \cdot e^3)(k^4 \cdot e^2)(e^1\cdot e^4)
       - u(k^2 \cdot e^3)(k^4 \cdot e^1)(e^2\cdot e^4)
       - u(k^2 \cdot e^3)(k^3 \cdot e^4)(e^1\cdot e^2)
\]
\[
       - u(k^2 \cdot e^3)(k^2 \cdot e^4)(e^1\cdot e^2)\Big],
\]
where we used the Mandelstam variables and momentum conservation
as $s = -t-u$.
Therefore the answer does not depend on $H_p$ and we will
see in \eqref{final}
that it matches with the computation of the left hand
side of \eqref{4equiv}, by considering the  
identity \eqref{mapocha}.

%***************************************************
\subsection{Non-minimal pure spinor computation}
%***************************************************

In this section we will compute the same one-loop amplitude
of section \ref{one_min}, but now
using the NMPS prescription of \ref{multiloop_NMPS}.

At the genus one surface the variables $s^\a$ and $d_\a$ have
eleven and sixteen zero-modes, respectively. 
Using the one-loop prescription of the non-minimal formalism,
the only place which can provide the 11 zero modes of $s^\a$ is
the regulator ${\cal N}$ of \eqref{reg_loop}. But they
are multiplied by the eleven $d_\a$ zero modes, and so the 
remaining five zero modes of
$d_\a$  must come either from the vertex operators or from
the single $b$ ghost. 

Since the three integrated vertex operators
can provide at most three $d_\a$ zero modes through the terms
$(W^\a d_\a)$, the single $b$ ghost of \eqref{nmps_bghost}
must provide two $d_\a$ zero modes
through the term
\be
\label{twod}
{(\lb\g^{mnp} r)(d\g_{mnp}d) \over 192(\l\lb)^2}.
\ee
After integrating over the zero modes of the dimension one fields
$(w_\a, \bar w^\a, d_\a, s^\a)$ using the measure factors described
in \ref{multiloop_NMPS}, one is left with an expression proportional
to 
\be
\label{loopone}
\int d^{16}\t \int [d\l][d\lb][dr] (\l\lb)^{-2} (\l)^4 
(\lb\g^{mnp} r) A W W W \exp (-\l\lb - r\t)
\ee
\be
\label{looptwo}
=\int d^{16}\t
\int [d\l][d\lb][dr] \exp (-\l\lb - r\t)(\l\lb)^{-2} (\l)^4 
(\lb\g^{mnp} D) A W W W 
\ee
where
$D_\a= {\p\over{\p\t^\a}} + (\g^m\t)_\a \p_m$ is the usual
superspace derivative and the index contractions on
\be
\label{indexc}(\l)^4 (\lb\g^{mnp} D) A W W W 
\ee
have not been worked out. 
Note that \eqref{looptwo} is obtained from
\eqref{loopone}
by writing $ r_\a \exp(-r\t)= 
{\p\over{\p\t^\a}} \exp(-r\t)$, integrating by parts with respect
to $\t$, and using conservation of momentum
to ignore total derivatives with respect to $x$. 
Furthermore, the factor of $(\l)^4$ in \eqref{loopone}
comes from the $\l$ in the
unintegrated vertex operator,
the $11$ factors of $\l$ and $\lb$ which multiply the zero modes of
$d_\a$ and $s_\a$ in ${\cal N}$, the factor of $(\l)^{-8}(\lb)^{-8}$
in the measure factor of $w_\a$ and $\bar w^{\a}$, and the factor
of $(\lb)^{-3}$ in the measure factor of $s^\a$.

Fortunately, it is easy to show there is a unique Lorentz-invariant
way to contract the indices in \eqref{indexc}.
Since $(\l^+)^4$ carries $+10$ $U(1)$ charge, 
$(\lb\g^{mnp} D) A W W W$ must carry $-10$ $U(1)$ charge
which is only possible if
$(\lb\g^{mnp} D)$ carries $-3$ charge, $A_\a$ carries $-{5\over 2}$
charge, and each $W^\a$ carries $-{3\over 2}$ charge.
Contracting the $SU(5)$ indices, one finds that the unique 
$U(1)\times SU(5)$ invariant contraction of the indices is
\be
\label{exone}(\l^+)^4(\lb\g_{abc} D) A_+ W^a W^b W^c.
\ee
Returning to covariant notation, one can easily see that
\eqref{indexc} must be proportional to the Lorentz-invariant expression
\be
\label{exonea}
(\lb\g_{mnp} D) (\l A) (\l \g^m W)(\l \g^n W)(\l \g^p W),
\ee
which reduces to \eqref{exone}
in the frame where $\l^+$ is the only non-zero component of $\l^\a$.

Note that this same conclusion can be obtained covariantly from the
prescription if we use the trick given by \eqref{truque}. That is 
because from the zero mode counting the kinematic factor is 
proportional to 
\[
\langle (\lb\g^{rst}D)(d\g_{rst}d)(\l A)(d W)(d W)(d W) \rangle,
\]
which, upon use of \eqref{truque}, immediately becomes \eqref{exonea}.

However, to express the kinematic factor as an integral over pure spinor
superspace as in \eqref{fimloop}, it is convenient to have an
expression in which all $\lb_\a$'s
appear in the combination $(\l^\a \lb_\a)$.
If all $\lb$'s appear in this combination one can use that, up to a constant,
\be
\label{trick}
\int d^{16} \t\int [d\l][d\lb][dr] \exp (-\l\lb - r\t)(\l\lb)^{-n} 
\l^\a \l^\b \l^\g f_{\a\b\g} 
= 
\langle \l^\a \l^\b \l^\g f_{\a\b\g} \rangle.
\ee

To convert \eqref{exonea} to this form, it is convenient to return to
the frame in which
$\l^+$ is the only non-zero component of $\l^\a$ and write \eqref{exone} as
\be
\label{extwo}(\l^+)^4 \e_{abcde}
(\lb^{[de]} D_+ - \lb_+ D^{[de]}) A_+ W^a W^b W^c.
\ee
Using the superspace equations of motion for $A_\a$ and $W^\a$, it
is easy to show that 
\be
\label{ides}
D_+ A_+ = D_+ W^a=0, \quad
D^{[de]} A_+ + D_+ A^{[de]} =0, \quad \e_{abcde} D^{[ab]} W^c = \cF_{de}.
\ee
So \eqref{extwo} is proportional to two terms which are
\be
\label{exthree}
(\l^+)^4 \lb_+  \e_{abcde} (D_+ A^{[de]}) W^a W^b W^c
\quad{\rm and}\quad
(\l^+)^4 \lb_+  A_+ W^a W^b  \cF_{ab}.
\ee

The second term in \eqref{exthree}, when written in covariant language, is proportional
to
\be
\label{exfiveb}(\l\lb) (\l A)(\l \g^m W)(\l \g^n W) \cF_{mn},
\ee
which produces the desired pure spinor superspace integral 
of \eqref{fimloop}.
And the first term in \eqref{exthree} can be written as
\be
\label{exfour}
(\l\lb) \big[(\l D)(\l\g^{mn}A)\big] (\l\g^p W)(W\g_{mnp} W),
\ee
which produces the pure spinor superspace integral
\be
\label{exfive}
\langle
\big[(\l D)(\l\g^{mn}A)\big] (\l\g^p W)(W\g_{mnp} W)\rangle.
\ee

But since BRST-trivial operators decouple, 
$$\langle (\l D) 
\big[(\l\g^{mn}A)(\l\g^p W)(W\g_{mnp} W)\big]\rangle =0,$$
which implies that \eqref{exfive} is equal to
\[
\langle
(\l \g^{mn} A) (\l D)\big[(\l\g^p W)(W\g_{mnp} W)\big]\rangle = 
-\frac{1}{2}(\l \g^{mn} A)(\l \g^p W)(\l\g^{rs}\g_{mnp} W){\cal F}_{rs}
\]
\[
= -24 \langle (\l A)(\l\g^r W)(\l \g^s W){\cal F}_{rs}
\]
where we used the equation of motion for the superfield $W^{\a}$ and
several gamma matrix identities in the last step.
So we finally have shown that the non-minimal computation
of the kinematic factor is proportional to the minimal kinematic
factor of \eqref{fimloop}, which was indeed to be expected
due to the formal proof presented in \cite{joost}.

%*******************************************************
\subsection{Covariant proof of equivalence}
%*******************************************************

In the last section we used the U(5) decomposition 
of the superfields to show the
equivalence between the kinematic factor \eqref{fimloop}, 
obtained with the minimal pure spinor formalism, and the non-minimal expression of
\eqref{exonea}. The proof was rather straightforward but it doesn't teach us
how to deal with expressions containing four $\l^{\a}$'s and one $\lb_{\a}$ in general, 
where it is not always the case that going to the U(5) frame where $\l_{[ab]}=\l^a=0$
leads to manageable expressions. 

So it is worth devoting some time in trying to find a covariant (and general) method
to deal with kinematic factors containing $\l^4\lb$ which are obtained in the
non-minimal pure spinor formalism. These expressions are of the following form
\be
\label{fourl}
\int d^{16}\t \int [d\l][d\lb][dr] {\rm e}^{-(\l\lb)-(r\t)}
(\l\lb)^{-n}
\l^{\a}\l^{\b}\l^{\g}\l^{\d}\lb_{\e}
f_{\a\b\g\d}^{\e}(\t)
\ee
and are obtained after integration over the non-minimal measures $[ds]$,$[dw]$ and $[d{\bar w}]$.
The reason to integrate over these particular variables is because the end result contains the
same integrations to perform as the tree-level amplitude prescription.
This is interesting because for tree-level amplitudes there is the notion of 
pure spinor superspace, where one uses the fact that the integrations over 
$[dr]$,$[d\l]$,$[d\lb]$ and $d^{16}\t$
select
the terms proportional to $\langle (\l\g^m\t)(\l\g^n\t)(\l\g^p\t)(\t\g_{mnp}\t)\rangle$.

What we now require is a rule analogous to \eqref{trick} for expressions of the
type \eqref{fourl} or, in other words, we
need to have a covariant prescription to integrate\footnote{Note that
one could in principle use the explicit form for the measures
$[dr]$,$[d\l]$,$[d\lb]$ to integrate the pure spinor variables and arrive
at the final answer. But in doing so one looses the elegance and the 
simplifying features of expressions written in pure spinor superspace,
so we avoid that route here.}
over four $\l$'s  and one $\lb$,
\be
\label{quat}
\langle \l^{\a}\l^{\b}\l^{\g}\l^{\d}\lb_{\e} f_{\a\b\g\d}^{\e}\rangle_{(4,1)}.
\ee

To derive such a rule one needs to
write down a tensor
$T^{\a\b\g\d}_{\e}$ which is symmetric and gamma-matrix traceless
with respect to the four Weyl indices,
\[
\frac{33}{2}T^{\a\b\g\d}_{\e} = \d^{\a}_{\e}T^{\b\g\d} + \d^{\b}_{\e}T^{\a\g\d}
+ \d^{\g}_{\e}T^{\a\b\d} + \d^{\d}_{\e}T^{\a\b\g}
-\frac{1}{12}\Big[ \g^m_{\e\k}\g_m^{\a\b} T^{\g\d\k} +
\g^m_{\e\k}\g_m^{\a\d} T^{\b\g\k} \Big.
\]
\be
\label{Tq}
\Big.
+ \g^m_{\e\k}\g_m^{\a\g} T^{\b\d\k}
+ \g^m_{\e\k}\g_m^{\b\g} T^{\a\d\k} + \g^m_{\e\k}\g_m^{\b\d} T^{\a\g\k}
+ \g^m_{\e\k}\g_m^{\g\d} T^{\a\b\k} \Big].
\ee
This is valid because there is only one scalar built out of four pure spinors $\l^{\a}$, one
pure spinor $\lb_{\a}$ and five unconstrained $\t^{\a}$'s. Using the theory of group
representations one can show the following to be true \cite{LiE}
\[
\lb\l^4 = [0,0,0,1,0]\otimes [0,0,0,0,4] 
=  1X[0,0,0,0,3] +1X[0,0,0,1,4] +1X[0,1,0,0,3]
\]
\[
\t^5 = 1X[0,0,0,3,0] +1X[1,1,0,1,0]
\]
so that
\[
\lb\l^4\t^5 = 1X[0,0,0,0,0] + 2X[0,0,0,0,4] + 5X[0,0,0,1,1] + 1X[0,0,0,1,5] + {\ldots} 
\]
where the ${\ldots} $ are higher rank representations. If $\l^{\a}$ were not a pure spinor
then there would be three different scalars in the above decomposition.

Using \eqref{Tq} we can translate pure spinor superspace expressions of the
type $(4,1)$ into a sum of familiar $(3,0)$ expressions
\be
\label{T41}
\frac{33}{2(\l\lb)}\langle \l^{\a}\l^{\b}\l^{\g}\l^{\d}\lb_{\e} f_{\a\b\g\d}^{\e}\rangle_{(4,1)} =
\ee
\[
= 
\langle \l^{\b}\l^{\g}\l^{\d} f_{\a\b\g\d}^{\a}\rangle + 
\langle \l^{\a}\l^{\g}\l^{\d} f_{\a\b\g\d}^{\b}\rangle +
\langle \l^{\a}\l^{\b}\l^{\d} f_{\a\b\g\d}^{\g}\rangle +
\langle \l^{\a}\l^{\b}\l^{\g} f_{\a\b\g\d}^{\d}\rangle
\]
\[
-\frac{1}{12}\Big[ 
\langle (\l\g^m)_{\e}\g_m^{\a\b}\l^{\g}\l^{\d}f^{\e}_{\a\b\g\d}\rangle +
\langle (\l\g^m)_{\e}\g_m^{\a\d}\l^{\b}\l^{\g}f^{\e}_{\a\b\g\d}\rangle +
\langle (\l\g^m)_{\e}\g_m^{\a\g}\l^{\b}\l^{\d}f^{\e}_{\a\b\g\d}\rangle
\Big.
\]
\be
\label{tricktwo}
\Big.
\langle (\l\g^m)_{\e}\g_m^{\b\g}\l^{\a}\l^{\d}f^{\e}_{\a\b\g\d}\rangle +
\langle (\l\g^m)_{\e}\g_m^{\b\d}\l^{\a}\l^{\g}f^{\e}_{\a\b\g\d}\rangle +
\langle (\l\g^m)_{\e}\g_m^{\g\d}\l^{\a}\l^{\b}f^{\e}_{\a\b\g\d}\rangle
\Big]
\ee

We will now proceed to show that using \eqref{tricktwo} one can derive that
\be
\label{toshow}
\langle (\lb\g_{mnp} D)[ (\l A) (\l \g^m W)(\l \g^n W)(\l \g^p W)]\rangle_{(4,1)}
= 40 (\l\lb)\langle (\l A) (\l \g^m W)(\l \g^n W){\cal F}_{mn}\rangle,
\ee
which is the covariant equivalence proof we are looking for.
Acting with the derivative over the superfields in \eqref{toshow} one obtains
\[
\langle (\lb\g_{mnp} D)[ (\l A) (\l \g^m W)(\l \g^n W)(\l \g^p W)]\rangle =
\]
\be
\label{obt}
\langle \big[(\lb\g_{mnp} D)(\l A)\big](\l\g^m W)(\l\g^n W)(\l\g^p W)\rangle
-
\langle (\l A)(\lb \g_{mnp} D)[(\l\g^m W)(\l\g^n W)(\l\g^p W)] \rangle.
\ee
However 
using pure spinor properties and
the equations of motion of super-Yang-Mills theory we find that the second
term in the right hand side of \eqref{obt} is given by
\be
\label{obvious}
- \langle (\l A)(\lb \g_{mnp} D)(\l\g^m W)(\l\g^n W)(\l\g^p W)\rangle = 
36 (\l\lb) \langle (\l A)(\l\g^m W)(\l\g^n W){\cal F}_{mn} \rangle,
\ee
so that \eqref{exonea} is manifestly equivalent to \eqref{fimloop} when
the derivative acts over the $W$'s. To finish the equivalence proof 
we need to evaluate the expression
\be
\label{need}
\langle \big[(\lb\g_{mnp} D)(\l A)\big](\l\g^m W)(\l\g^n W)(\l\g^p W)
\rangle.
\ee
Using \eqref{tricktwo} and paying attention to the normalization
one converts \eqref{need}
into a pure spinor superspace expression with three $\l$'s,
\[
\frac{33}{2(\l\lb)}\langle \big[(D \g_{mnp} \lb)(\l A)\big](\l\g^m W)(\l\g^n W)(\l\g^p W)
\rangle =
\]
\[
\langle [(D \g_{mnp} A)](\l\g^m W)(\l\g^n W)(\l\g^p W) \rangle
+ 24\langle [D_{\a}(\l A)](\g_{np} W)^{\a}(\l\g^n W)(\l\g^p W) \rangle
\]
\[
+ \frac{1}{4}\langle [(\l\g_q \g_{mnp}D)A_{\a}](\g^q \g^m W)^{\a}(\l\g^n W)(\l\g^p W)\rangle
\]
\be
\label{threel}
+ \frac{1}{4}\langle [(\l\g_q \g_{mnp}D)(\l A)](W\g^m \g^q \g^n W)(\l \g^p W)\rangle,
\ee
where in the above it has to be understood that $D_{\a}$ acts only over 
the superfield $A_{\b}$. Note that the last two terms of \eqref{threel} come from the
gamma matrix terms of \eqref{tricktwo}.

Using 
\be
\label{ymone}
D_{\a}A_{\b}+D_{\b}A_{\a} = \g^q_{\a\b}A_q
\ee
one can show that
\[
\langle D_{\a}(\l A)(\g_{np} W)^{\a}(\l\g^n W)(\l\g^p W) \rangle
= 
\]
\[
-\langle (\l D)A_{\a}(\g_{np} W)^{\a}(\l\g^n W)(\l\g^p W) \rangle
+ \langle A_m (\l\g^m \g^{np} W)(\l\g^n W)(\l\g^p W) \rangle
\]
where the last term is zero due to the pure spinor condition and the BRST charge
in the first one can be integrated by parts to give
\[
-\langle (\l D)A_{\a}(\g_{np} W)^{\a}(\l\g^n W)(\l\g^p W) \rangle =
+\frac{1}{4} \langle (\l \g^{rs} \g_{np} A)(\l\g^n W)(\l\g^p W){\cal F}_{rs} \rangle 
\]
\[
= -2 \langle (\l A)(\l\g^m W)(\l \g^n W){\cal F}_{mn}\rangle.
\]
Doing similar manipulations we also get
\[
\langle (\l \g_q \g_{mnp})^{\b} (D_{\b} A_{\a})(\g^q \g^m W)^{\a}(\l\g^n W)(\l\g^p W)\rangle
=  4 \langle \big[(\l D)A_{\a}\big](\g^q \g^m W)^{\a}(\l\g^m W)(\l\g^q W)
\rangle
\]
\[
= \langle (\l\g^{rs}\g_{mq} A)(\l\g^m W)(\l\g^q W){\cal F}_{rs}\rangle
= - 8 \langle (\l A)(\l\g^m W)(\l \g^n W){\cal F}_{mn}\rangle.
\]
and
\[
\langle (\l\g_q \g_{mnp})^{\a}[D_{\a}(\l A)](W\g^m \g^q \g^n W)(\l \g^p W)\rangle
= - \langle Q\big[(\l\g_q \g_{mnp} A)\big](W\g^m \g^q \g^n W)(\l \g^p W)\rangle
\]
\be
\label{zer}
+ \langle (\l\g_{mnpqr} \l)A_r (W\g^{mnq} W)(\l\g^p W) \rangle.
\ee
The last line of \eqref{zer} is zero due to the pure spinor condition, so \eqref{zer}
becomes
\[
\langle (\l\g_q \g_{mnp})^{\a}D_{\a}(\l A)(W\g^m \g^q \g^n W)(\l \g^p W)\rangle
= -\frac{1}{2}\langle (\l\g_q\g_{mnp} A)(\l\g^{rs}\g^m\g^q\g^n W)(\l\g^p W){\cal F}_{rs}
\rangle
\]
\be
\label{equa}
= - 32 \langle (\l A)(\l\g^m W)(\l \g^n W){\cal F}_{mn} \rangle,
\ee
where we integrated the BRST-charge by parts and went through a long list
of gamma matrix manipulations.

Finally, the first term in the right hand side of \eqref{threel} can be rewritten using
the gamma matrix identity of $\eta_{mn}\g^m_{\a(\b}\g^n_{\g\d)}=0$,
\[
\langle [(D \g_{mnp} A)](\l\g^m W)(\l\g^n W)(\l\g^p W) \rangle =
\]
\[
= \langle (\g_m\g_n W)^{\sigma}(\g_p\g^n \l)^{\rho}D_{\s}A_{\rho}(\l\g^m W)(\l\g^p W)\rangle
+ \langle(\g_m\g_n \l)^{\sigma}(\g_p\g^n W)^{\rho}D_{\s}A_{\rho}(\l\g^m W)(\l\g^p W)\rangle
\]
\[
= + 2\langle(W\g_n\g_m)^{\a}[D_{\a}(\l A)](\l\g^m W)(\l\g^n W)\rangle
+ 2\langle(\g^p\g^n W)^{\s}[(\l D)A_{\s}](\l\g^n W)(\l\g^p W)\rangle
\]
Using \eqref{ymone} in the first term we get,
\[
= - 4 \langle (W\g_n\g_m)^{\a}[D_{\a}(\l A)](\l\g^m W)(\l\g^n W)\rangle
= - \langle(\l\g^{rs}\g_n\g_m A)(\l\g^m W)(\l\g^n W){\cal F}_{rs}\rangle
\]
and so
\be
\label{lega}
\langle [(D \g_{mnp} A)](\l\g^m W)(\l\g^n W)(\l\g^p W) \rangle =
- 8 \langle (\l A)(\l\g^m W)(\l\g^n W){\cal F}_{mn}\rangle
\ee

Using all the above identities \eqref{T41} implies
\[
\langle \big[(\lb \g_{mnp} D)(\l A)\big](\l\g^m W)(\l\g^n W)(\l\g^p W)
\rangle = 4(\l\lb) \langle (\l A)(\l\g^m W)(\l\g^n W){\cal F}_{mn},
\]
and therefore, from \eqref{obt} and \eqref{obvious} one finally gets
\be
\label{shown}
\langle (\lb\g_{mnp} D)[ (\l A) (\l \g^m W)(\l \g^n W)(\l \g^p W)]\rangle_{(4,1)}
= 40 (\l\lb)\langle (\l A) (\l \g^m W)(\l \g^n W){\cal F}_{mn}\rangle,
\ee
which is what we wanted to show.
We have thus just obtained a covariant proof of equivalence 
between the minimal and non-minimal superspace expressions of
the the massless
four-point kinematic factor.

We also observe that the terms containing gamma matrices in \eqref{tricktwo}
-- they are responsible for the traceless property of $T^{\a\b\g\d}_{\e}$ --
covariantly generate the term \eqref{exfive} whose existence was deduced
through U(5) arguments in \cite{nmps_two}. This can be checked by noticing that
\eqref{zer} can also be written as,
\be
\label{equad}
\langle (\l\g_q \g_{mnp}D)(\l A)(W\g^m \g^q \g^n W)(\l \g^p W)\rangle =
- 4 \langle\big[(\l\g_{mn}D)(\l A)\big](W\g^{mnp}W)(\l\g_p W) \rangle,
\ee
by using the pure spinor property of $(\l\g_p)_{\a}(\l\g^p)_{\b}=0$ to
get rid of the $\g_p$ inside of $(\l\g_q \g_{mnp}D)(\l A)$.

%*********************************************************
\subsection{Yet another covariant proof of equivalence}
%*********************************************************

There is yet another covariant proof of the equivalence between 
\eqref{fimloop} and \eqref{exonea}, which is perhaps more elegant
than the proof presented in the previous section.

From \eqref{obt} and \eqref{obvious} we know that
\[
\langle (\lb\g_{mnp} D)[ (\l A) (\l \g^m W)(\l \g^n W)(\l \g^p W)]\rangle = 
36(\l\lb)\langle (\l A) (\l \g^m W)(\l \g^n W){\cal F}_{mn}\rangle
\]
\be
\label{needtwo}
+ \langle \big[(\lb\g_{mnp} D)(\l A)\big](\l\g^m W)(\l\g^n W)(\l\g^p W)
\rangle.
\ee
Using $\eta_{mn}\g^m_{\a(\b}\g^n_{\g\d)}=0$ and that the factor of
$\big[(\lb\g_{mnp} D)(\l A)\big]$ can be substituted by
$\big[(\lb\g_m\g_n\g_p D)(\l A)\big]$ we arrive at the following
identity for the second term of \eqref{needtwo}
\[
\langle \big[(\lb\g_{mnp} D)(\l A)\big](\l\g^m W)(\l\g^n W)(\l\g^p W)
\rangle 
= \langle (\lb\g_m\g_n W)\Big[(\l\g^n\g_p D)(\l A)\Big](\l\g^m W)(\l\g^p W)\rangle
\]
\be
\label{legal}
+ \langle (\lb\g_m\g_n \l)(W \g^n\g_p)^{\s} \Big[D_{\s}(\l A)\Big](\l\g^m W)(\l\g^p W)\rangle.
\ee
Using $\g^n\g_p = - \g_p\g^n + 2\d^n_p$ and the equation of motion $Q(\l A)= 0$ the first
term of \eqref{legal} vanishes, while the second can be rewritten as
\[
\langle (\lb\g_m\g_n \l)(W \g^n\g_p)^{\s} \Big[D_{\s}(\l A)\Big](\l\g^m W)(\l\g^p W)\rangle=
\]
%\[
%= 2(\l\lb)\langle (W \g_m\g_p)^{\s} \Big[D_{\s}(\l A)\Big](\l\g^m W)(\l\g^p W)\rangle 
%\]
\[
= -2(\l\lb)\langle (W \g_m\g_p)^{\s}\Big[(\l D)A_{\s}\Big](\l\g^m W)(\l\g^p W)\rangle
\]
\[
= + 4(\l\lb)\langle (\l A)(\l\g^m W)(\l\g^n W){\cal F}_{mn}\rangle,
\]
where we used \eqref{ymone} and integrated the BRST-charge by parts. So we have just
shown that \eqref{needtwo} is equal to
\[
\langle (\lb\g_{mnp} D)\Big[(\l A)(\l\g^m W)(\l\g^n W)(\l\g^p W)\Big]\rangle=
40(\l\lb)\langle (\l A)(\l\g^m W)(\l\g^n W){\cal F}_{mn}\rangle,
\]
which finishes the proof. The result is obviously the same as \eqref{shown}.

%*******************************************************
\subsection{On the pure spinor expression \eqref{crazy_one}}
\label{crazy_four_sub}
%*******************************************************

The pure spinor superspace expression \eqref{crazy_one} 
\be
\label{crazy_oneA}
K_c = \langle (\l\g^m\t)(\l\g^n W^1)(\l\g^p W^2)(W^3\g_{mnp} W^4) \rangle
\ee
is an interesting one. 
The complete explanation of its origins remains unknown\footnote{It was found when
trying to discover a prescription which uses only integrated vertices to compute the
massless four-point amplitude at one-loop and can be roughly seen as the substitution
rule of $\int b V \rightarrow (d^0 \theta) \int U$ inside the kinematic factor.} to this
date, but it may turn out to be useful as a hint to further studies.

First of all one should notice that \eqref{crazy_oneA} is manifestly gauge invariant
but appears to be non-supersymmetric. Secondly, it is BRST-closed and finally it has the 
dimensions of a $F^4$ term. What does the component expression of 
\eqref{crazy_oneA} look like?

Using the SYM superfield expansions of \eqref{sym_exp} one easily obtains the 
bosonic contribution
\[
K_c = -\frac{1}{256}\langle (\l\g^m \t)(\l\g^n \g^{m_1n_1}\t)(\l\g^p \g^{m_2n_2}\t)
(\t\g^{m_3n_3}\g_{mnp}\g^{m_4n_4}\t)\rangle
F^1_{m_1n_1}{\ldots}F^4_{m_4n_4}
\]
\be
\label{mist}
= A t_8^{m_1n_1{\ldots} m_4n_4}F^1_{m_1n_1}{\ldots}F^4_{m_4n_4}
\ee
whose proportionality with the four-point kinematic factor at one-loop seems surprising 
at first, and clearly deserves some kind of explanation.

Although \eqref{crazy_oneA} seems to be non-supersymmetric due to the explicit $\t$, one
can show that its supersymmetry variation is a total derivative,
\be
\label{susyvar}
\d_{susy} K_c = \langle (\l\g^m \e)(\l\g^n W^1)(\l\g^p W^2)(W^3\g_{mnp} W^4) \rangle
\ee
because $\d_{susy}K_c$ is proportional to the anomaly kinematic factor, which is known to
be a total derivative\footnote{To make it clearer, just suppose that one of the $W$'s 
in \eqref{anomaly}
is a constant spinor $\e$.}.

One can understand the appearance of the $t_8$ tensor in \eqref{mist} by noticing that
in the bosonic $\t$-expansion of $(D\g^{qrs}A)$, due to the antisymmetry of $\g^{qrs}$ 
in its spinor indices,
has no components with zero thetas,
\[
(D\g^{mnp}A) =  \frac{1}{4}(\t\g^{mnp}\g^{tu}\t)F_{tu}
+ \frac{1}{4}\p_t a_u (\t\g^t \g^{mnp} \g^u \t) + {\ldots} 
\]
\be
\label{twotheta}
=  -(\t\g^{mnp} W)
+ \frac{1}{4}\p_t a_u (\t\g^t \g^{mnp} \g^u \t).
\ee
where the substitution in the second line is valid up to fermionic terms. 

Now if one considers the bosonic computation of
\be
\label{tempL}
L = \langle \big[(D\g_{mnp} A)\big](\l\g^m W)(\l\g^n W)(\l\g^p W)\rangle
\ee
there is only one contribution coming from the superfield $(D\g_{mnp} A)$, namely the 
two-thetas term of
\eqref{twotheta}. So the pure spinor expression \eqref{tempL} becomes
\[
L = - \langle (\t\g_{mnp} W)(\l\g^m W)(\l\g^n W)(\l\g^p W)\rangle
+ \frac{1}{4}\p_t a_u \langle (\t\g^t \g^{mnp} \g^u \t)(\l\g^m W)(\l\g^n W)(\l\g^p W)\rangle.
\]
It is a straightforward exercise to show that the bosonic part of second term vanishes 
identically, 
leading us to
\[
L = - \langle (\t\g_{mnp} W)(\l\g^m W)(\l\g^n W)(\l\g^p W)\rangle.
\]
On one hand, using that $(\l\g^m)_{\d_1}(\l\g^n)_{\d_2}(\l\g^p)_{\d_3}
(\g_{mnp})_{\d_4\d_5}$ is completely antisymmetric in its spinor indices
one obtains, 
\be
\label{finalc}
L = - \langle  (\l\g^m\t)(\l\g^n W)(\l\g^p W)(W\g_{mnp} W) \rangle.
\ee
on the other hand, from \eqref{lega} and \eqref{tempL} one is led to conclude
\be
\label{pil}
\langle  (\l\g^m\t)(\l\g^n W)(\l\g^p W)(W\g_{mnp} W) \rangle
= 8 \langle (\l A)(\l\g^m W)(\l\g^n W){\cal F}_{mn} \rangle.
\ee
As it was shown in \eqref{susyvar}, the left hand side of \eqref{pil} is supersymmetric. So
it turns out that \eqref{pil} is valid for all bosonic and fermionic components! 
Consequently not only 
\eqref{mist} is expected to happen at the level of pure spinor superspace but it is
a completely equivalent way of computing the massless four-point kinematic factor at one-loop.
Note that somehow we are exchanging manifest supersymmetry by manifest gauge invariance 
when computing either the left or right hand side of \eqref{pil}. This may be a hint on how
to develop an amplitude prescription at one-loop which uses only integrated vertices, 
indicating
that maybe manifest supersymmetry will be spoiled.

It is also interesting to note that the analogous generalization for the massless 
five-point amplitude,
\[
K_5 = \langle (\l\g^m \t)(\l\g^n \g^{rs} W^5)(\l\g^p W^1)
(W^3\g_{mnp}W^4){\cal F}^2_{rs}\rangle
- (2\leftrightarrow 5)
\]
leading to
\[
K_5 = \langle (D\g_{mnp}A^1)(\l\g^m \g^{rs} W^5)(\l\g^n W^3)(\l\g^p W^4){\cal F}^2_{rs}\rangle
- (2\leftrightarrow 5)
\]
also gives rise to the correct bosonic component expression (as we shall compute in 
section \ref{fracas})\footnote{The computation of \ref{fracas} was a work in progress
when the thesis was finished in 2008. It has been completed now \cite{5pt}.}.
Whether they are the correct full kinematic factor is presently under investigation
with the collaboration of Christian Stahn.

%**************************************************
\section{Massless four-point amplitude at two-loops}
%**************************************************

In this section we will present the massless four-point amplitude at two-loops and
as in the other amplitudes considered in this thesis, we will focus our attention in
the form of the kinematic factor in pure spinor superspace. 

There are two different ways of computing this particular amplitude, using the
minimal and the non-minimal pure spinor formalism. As it turns out, the minimal
version is more efficient in getting the final expression for the kinematic factor.
With the non-minimal version what one obtains is a kinematic factor which contains
also the pure spinor $\lb$'s, complicating things a bit. But in the end one can
prove that both versions are equivalent to
\be
\label{tloop}
K_{\rm 2} =
\langle (\l \g^{mnpqr}\l){\cal F}^1_{mn}{\cal F}^2_{pq}{\cal F}^3_{rs}
(\l \g^s W^4)\rangle\Delta(1,3)\Delta(2,4)
+ \rm{perm.(1234)}.
\ee
This is again a manifestation of the fact that both prescriptions were
shown to be equivalent in \cite{joost}.

%********************************
\subsection{Minimal computation}
%********************************

Following the prescription given in section \ref{multi_min} the 
massless four-point closed\footnote{The closed string amplitude
is computed as the holomorphic square of the open string. We use this
fact implicitly all the time in this thesis, by considering only
one Riemann surface at each genus. However we always present the
results in terms of open superstring kinematic factors $K$, from which
the closed string kinematic factor can be obtained by the holomorphic
square $K\bar K$. In terms of the effective action, the Riemann tensor
is obtained by identifying $R_{mnpq}$ with $F_{mn}{\bar F}_{pq}$,
or $R_{mnpq} = k_{[m}h_{n][q}k_{q]}$, where we identified 
$e_n{\tilde e}_q \rightarrow h_{nq}$.}
superstring amplitude at genus two is given by
\[
A = \int d^2\tau_1 d^2\tau_2 d^2\tau_3 \langle |
\prod_{P=1}^{3}\int d^2 u_P \mu_P(u_P)
{\tilde b_{B_P}}(u_P,z_P) 
\]
\[
\prod_{P=4}^{20g} Z_{B_P}(z_P) \prod_{R=1}^{2} Z_J(v_R)
\prod_{I=1}^{11} Y_{C_I}(y_I)~|^2 ~\prod_{i=1}^4 \int d^2 z_i U^i(z_i)
\rangle,
\]
The explicit computation can be done by considering the 32 zero modes
of $d_{\a}$ and the 16 zero modes of $\t^\a$. 

Seventeen zero modes of
$d_\a$ come from the $Z_B$ operators, while another two come from 
the $Z_J$'s. Therefore the remaining thirteen zero modes must
come from the three b-ghosts and the external vertices, and there is
only one possibility. This is because the b-ghost has terms with
zero, one, two or four $d_\a$'s.
So suppose the three b-ghosts contribute ten $d$'s and the
external vertices the remaining three through $(dW)^3$. In
this case the amplitude vanishes because there is a factor of $\d'(B^{mn}N_{mn})$
in the term with four $d$'s in the b-ghost, 
\[
b_B\Big|_{4 d} = B_{mn}B^{qr}(d\g^{mnp}d)(d\g_{pqr}d)\d'(B^{st}N_{st})
\]
which makes the amplitude
vanish because there is no $N_{mn}$ coming from the external vertices
to cancel the derivative of the delta function\footnote{The derivative of
the delta function satisfies $x\d'(x) = - \d(x)$.}.

Going through all the possibilities one obtains a vanishing result for
all but one case, when the three b-ghosts provide twelve zero modes
of $d_\a$ and three $\d'(B^{mn}N_{mn})$ while the four external vertices
contribute with the required $(dW)(N\cdot F)^3$ term to give a non-vanishing
result.

The integrations over the pure spinor measures have the effect of
selecting only the components such that $\langle\l^3 {\cal F}{\cal F}{\cal F}W \rangle$
is proportional to the pure spinor measure \eqref{measure}, and one can easily
check that there is only one scalar which can be built out of this superfields,
namely
\be
\label{twol}
\langle 
(\l\g^{mnpqr}\l)(\l\g^sW^4){\cal F}^1_{mn}{\cal F}^2_{pq}{\cal F}^3_{rs}.
\rangle
\ee
The computation of the moduli space part is summarized as follows. Each
b-ghost has conformal weight $+2$ and no poles over the Riemann surface
of genus two and their product has zeros when their position coincide.
These conditions uniquely determine their Riemannian contribution
to be given by the product of three quadri-holomorphic 1-forms
$\Delta(u_1,u_2)\Delta(u_2,u_3)\Delta(u_3,u_1)$. Analogously, each external
vertex has conformal weight $+1$ and therefore their contribution is
given by some linear combination of the holomorphic 1-forms, 
$h^{IJKL}w_I(z_1)w_J(z_2)w_K(z_3)w_L(z_4)$. The only linear combination
compatible with the symmetries of \eqref{twol} is given by
the following permutations
\be
\label{twol2}
K_2 = 
\langle 
(\l\g^{mnpqr}\l)(\l\g^sW^4){\cal F}^1_{mn}{\cal F}^2_{pq}{\cal F}^3_{rs}
\Delta(z_1,z_3)\Delta(z_2,z_4)
\rangle  + \text{perm}(1234).
\ee
For example, due to the symmetry of $(\l\g^{mnpqr}\l){\cal F}_{mn}^1{\cal F}_{pq}^2$
under $(1 \leftrightarrow 2)$ there could be no factor of $\Delta(z_1,z_2)$ in the
above combination.

And from the theory of Riemann surfaces one can show that 
\[
\int d^2\tau_1 d^2\tau_2 d^2\tau_3| \prod_{j=1}^3 \int d^2 u_j
\mu(u_j) \Delta(u_1,u_2)\Delta(u_2,u_3)\Delta(u_3,u_1)|^2 = 
\int d^2\Omega_{11}d^2\Omega_{12}d^2\Omega_{22},
\]
where $\Omega_{IJ}$ is the $2\times 2$ period matrix of the genus two Riemann
surface.

Finally, the amplitude is given by
\[
{\cal A}_2 \propto {\rm e}^{2\phi} {\tilde K}_2{\bar {\tilde K}}_2 \int_{{\cal M}_2}
{|d^3\Omega|^2 \over ({\rm det}\,{\rm Im}\Omega)^5} F_2(\Omega, {\cal Y}),
\]
where $K_2= {\tilde K}_2{\cal Y}$ and $F_2(\Omega, {\cal Y})$, apart 
from the factor of $|{\cal Y}|^2$,
\[
{\cal Y}(s,t,u) = \left[(u-t)\Delta(1,2)\Delta(3,4) 
+(s-t)\Delta(1,3)\Delta(2,4) + (s-u)\Delta(1,4)\Delta(2,3)
\right]
\]
comes from the standard integration over zero modes
of $X^m$ and is given by
\[
F_2(\Omega, {\cal Y}) = \int |{\cal Y}|^2 \prod_{i<j}G(z_i,z_j)^{k_i\cdot k_j},
\]
where $G(z_i,z_j)$ is the scalar Green function.

%**************************************
\subsection{Non-Minimal computation}
%**************************************

The computation as dictated by the non-minimal pure spinor formalism
gives rise to the same moduli space part described in the last subsection, 
with a final integration over the period matrix of the Riemann surface.
The only difference comes from the computation of the kinematic factor,
due to the different expression for the b-ghost and the additional
functional integrations over the non-minimal variables.

If we restrict our attention to the kinematic factor, then the
computation can be easily performed by zero mode counting, and goes as
follows \cite{NMPS}.

First note that there is only one place where the 22 zero modes of
$s^\a$ can come from, the regulator ${\cal N}$ of \eqref{reg_loop}.
It must provide the whole 22 $s^\a$ zero modes, which come
multiplied by 22 zero modes of $d_\a$. 
So the remaining 10 $d_\a$ zero modes must come from the
four integrated vertex operators and the three $b_\a$ ghosts.
This is only possible if each integrated vertex operators provides
a $d_\a$ zero mode through the term $(W^\a d_\a)$ and each
$b$ ghost provides two $d_\a$ zero modes through the term of 
\eqref{twod}. 

After integrating over the zero modes of the conformal weight $+1$ fields
$(w^I_\a, \bar w^{I\a}, d^I_\a, s^{I\a})$ using the measure factors described
in \cite{NMPS}, one is left with an expression proportional
to 
\be
\label{sloopone}
\int d^{16}\t \int [d\l][d\lb][dr] (\l\lb)^{-6} (\l)^6 
(\lb\g^{mnp} r)^3 W W W W \exp (-\l\lb - r\t)
\ee
\be
\label{slooptwo}
=\int d^{16}\t
\int [d\l][d\lb][dr] \exp (-\l\lb - r\t)(\l\lb)^{-6} (\l)^6 
(\lb\g^{mnp} D)^3 W W W W 
\ee
where the index contractions on
\be
\label{indices}
(\l)^6 (\lb\g^{mnp} D)^3 W W W W 
\ee
will be found in the discussion below.
Note that the factor of $(\l)^6$ in \eqref{sloopone} comes from 
the $11g$ factors of $\l$ and $\lb$ which multiply the zero modes of
$d^I_\a$ and $s^I_\a$ in ${\cal N}$, the factor of $(\l)^{-8g}(\lb)^{-8g}$
in the measure factor of $w^I_\a$ and $\bar w^{I\a}$, and the factor
of $(\lb)^{-3g}$ in the measure factor of $s^{I\a}$.

As in the one-loop four-point amplitude, there is fortunately a unique
way of contracting the indices of \eqref{indices} in a Lorentz-invariant manner.
Choosing the Lorentz frame where $\l^+$ is the only non-zero component
of $\l^\a$, one finds that $(\l^+)^6$ contributes $+15$ $U(1)$ charge
so that each $(\lb\g^{mnp} D)$ must contribute $-3$ charge
and 
each $W$ must contribute $-{3\over 2}$ charge. Since the $-3$ component of
$(\lb \g^{mnp} D)$ is $(\lb^{[ab]} D_+ - \lb_+ D^{[ab]})$,
and since $D_+$ annihilates the $-{3\over 2}$ component of $W^\a$,
the only contribution to \eqref{indices} comes from a term of the form
\be
\label{gexone}
(\l^+)^6 (\lb_+)^3  (D^{[ab]} D^{[cd]} D^{[ef]})
 (W^g W^h W^j W^k)
\ee
where the ten $SU(5)$ indices are contracted with two $\e_{abcde}$'s.

The term of \eqref{gexone} produces three types of terms depending on how
the three $D$'s act on the four $W$'s. 
If all three $D$'s act on the same $W$,
one gets a term proportional to
$(\l^+)^6 (\lb_+)^3  W W W \p \cF $, which by $U(1)\times SU(5)$ invariance
must have the form
\be
\label{gexfour}
(\l^+)^6 (\lb_+)^3 W^a W^b W^c \p_a \cF_{bc}.
\ee
And if two $D$'s act on the same $W$, one gets a term proportional to
$(\l^+)^6 (\lb_+)^3  \cF W W \p W $, which by $U(1)\times SU(5)$ invariance
must have the form
\be
\label{gexthr}
(\l^+)^6 (\lb_+)^3 \cF_{bc} W^a W^b \p_a W^c.
\ee
Finally, if each $D$
acts on a different $W$, one obtains a term that is proportional to
$(\l^+)^6 (\lb_+)^3 W \cF \cF \cF$, which by $U(1)\times SU(5)$ invariance
must have the form
\be
\label{gextwo}
(\l^+)^6 (\lb_+)^3 \cF_{ab} \cF_{cd} \cF_{ef} W^f \e^{abcde}.
\ee

The first term in \eqref{gexfour} vanishes by Bianchi identities. And
the second term in \eqref{gexthr} is proportional to the first term after
integrating by parts with respect to $\p_a$ and using the equation
of motion $\p_a W^a=0$. So the only contribution
to the kinematic factor comes from the third term of \eqref{gextwo}, which can
be written in Lorentz-covariant notation as
\be
\label{geee}
(\l\lb)^3 (\l \g^{mnpqr} \l)\cF_{mn} \cF_{pq} \cF_{rs} (\l\g^s W).
\ee
So the non-minimal computation of the two-loop kinematic factor
agrees with the minimal computation of \eqref{twol2}.

%*********************************************************
\section{Relating massless four-point kinematic factors}
\label{ucla}
%*********************************************************

In this section the usefulness of having pure spinor superspace
expressions for the kinematic factors will reveal itself through 
a general proof that the scattering amplitudes of four massless strings
have the same kinematic factors at tree-, one- and two-loops, up
to Mandelstam invariants \cite{mafra_tree}.

%**************************************************************************
\subsection{Tree-level and one-loop}
%**************************************************************************

In the following sections we will make heavy use of 
the superfield equations of motion in the 
formulation of ten-dimensional 
Super-Yang-Mills theory in superspace (see review in Appendix A), namely
\be
\label{SYM}
Q{\cal F}_{mn} = 2\p_{[m} (\l\g_{n]} W), \quad
Q W^{\a} = {1\over 4}(\l\g^{mn})^{\a}{\cal F}_{mn},\quad  
QA_m = (\l\g_m W) + \p_m(\l A),
\ee
where $Q=\oint \l^{\a}d_{\a}$ is the pure spinor BRST operator.
With these
relations in hand we will show that \eqref{id_one} holds true.
To prove this we note that 
\[
\langle (\l A)\p^m(\l A)(QA^n)F_{mn}
\rangle = -\langle (\l A)\p^m(\l A)A^n (QF_{mn})\rangle,
\]
which upon use of
\eqref{SYM} and momentum conservation becomes
$$
\langle (\l A)\p^m(\l A)(QA^n)F_{mn}
\rangle =  \langle  (\l A)\p^m(\l A) \p_m A_n  (\l \g^n W) \rangle
$$
\be
\label{part}
- \langle  \p_n (\l A)\p_m(\l A) A^n  (\l \g^m W) \rangle
- \langle   (\l A)\p_n\p_m(\l A) A^n  (\l \g^m W) \rangle.
\ee
The second term can be rewritten like
$$
\langle  \p_n (\l A)\p_m(\l A) A^n  (\l \g^m W) \rangle =
- \langle  (\l A)(\l \g^m W)\left[ 
A^n \p_m\p_n(\l A) + \p^n(\l A)\p_m A_n 
\right]\rangle
$$
as can be shown by integrating $\p^m$ by parts and 
using the equation of motion for $W^{\a}$. So,
$$
\langle (\l A)\p^m(\l A)(QA^n)F_{mn} \rangle =
\langle (\l A)\p^m(\l A)(\l\g^n W)F_{mn} \rangle 
- 2 \langle   (\l A)\p_n\p_m(\l A) A^n  (\l \g^m W) \rangle
$$
which implies that
$
\langle (\l A)\p^m(\l A)\p^n(\l A)F_{mn} \rangle =
-2 \langle   (\l A)\p_n\p_m(\l A) A^n  (\l \g^m W) \rangle$, or equivalently,
\be
\label{impli}
\langle (\l A)\p^m(\l A)\p^n(\l A)F_{mn} \rangle = 
-2 \langle   (\l A)\p_n(QA_m) A^n  (\l \g^m W) \rangle.
\ee
Using $[Q,\p^n]=0$ and the decoupling of BRST-trivial operators, equation \eqref{impli}
becomes
$$
\langle (\l A)\p^m(\l A)\p^n(\l A)F_{mn} \rangle =
2\langle (\l A)(\p_n A_m) (QA^n)  (\l \g^m W) \rangle
$$
\be
\label{segue}
= \langle (\l A)(\l\g^m W)(\l \g^n W)F_{mn} \rangle +
2\langle (\p_n A_m)(\l A)\p^n(\l A) (\l\g^m W) \rangle.
\ee
Plugging \eqref{segue} in the tree-level kinematic factor \eqref{this}
we finally obtain
\be
\label{mapocha}
K_0 = - \langle (\l A)(\l\g^m W)(\l \g^n W)F_{mn} \rangle = -{1\over 3}K_1,
\ee
which finishes the proof and explicitly relates the tree-level
and one-loop kinematic factors.

%*****************************************************
%*****************************************************
\subsection{One- and two-loop}
%*****************************************************
%*****************************************************

To obtain a relation between the one- and two-loop kinematic
factors we first need to show that
$\langle (\l A^1)(\l \g^m W^2)(\l \g^n W^3){\cal F}^4_{mn}
\rangle$ is completely symmetric in the labels $(1234)$. This can be
done by noting that
\be
\label{brst}
\langle (\l \g^{mnpqr} \l)(\l A^1)(W^2 \g_{pqr} W^3){\cal F}^4_{mn}
\rangle =
4 \langle (\l A^1)Q\left[(W^2 \g_{pqr} W^3)\right](\l \g^{pqr} W^4)
\rangle.
\ee
Together with the identities 
\[
(\l\g^{mn}\g^{pqr}W^2)(\l\g_{pqr}W^4) 
= -48 (\l\g^{[m}W^2)(\l\g^{n]}W^4)
\]
\[
(\l \g^{mnpqr} \l)(W^2 \g_{pqr} W^3) = 
-96 (\l \g^{[m} W^2)(\l \g^{n]} W^3),
\]
equation \eqref{brst} implies that
$$
\langle (\l A^1)(\l \g^m W^4)(\l \g^n W^2){\cal F}^3_{mn}\rangle +
\langle (\l A^1)(\l \g^m W^3)(\l \g^n W^4){\cal F}^2_{mn}\rangle =
$$
\be\label{cont}
= 2 \langle  (\l A^1)(\l \g^m W^2)(\l \g^n W^3){\cal F}^4_{mn}
\rangle.
\ee
From \eqref{cont} it follows that,
\be\label{dem}
K_{\rm 1-loop} = 3 \langle (\l A^1)(\l \g^m W^2)(\l \g^n W^3){\cal F}^4_{mn}\rangle.
\ee
Furthermore, the independence of which vertex operator
we choose to be non-integrated  implies total
symmetry
of $\langle (\l A^1)(\l \g^m W^2)(\l \g^n W^3){\cal F}^4_{mn}
\rangle$ in the labels $(1234)$.

Now we can relate the one- and two-loop kinematic factors by noting that
$$
(\l \g^{mnpqr}\l){\cal F}^1_{mn}{\cal F}^2_{pq}{\cal F}^3_{rs}
(\l \g^s W^4) =
-4 Q\left[
(\l \g^r\g^{mn} W^2)(\l \g^s W^4){\cal F}^1_{mn}{\cal F}^3_{rs}
\right] 
$$
\be\label{note}
- 8 ik^1_m (\l\g_n W^1)(\l\g^r\g^{mn}W^2)(\l\g^s W^4){\cal F}^3_{rs},
\ee
where the pure spinor constraint $(\l\g^m\l)=0$ and the
identity $\eta_{mn}\g^m_{\a(\b}\g^n_{\g\delta)}=0$ must be used to show the
vanishing of terms containing factors of $(\l\g^m)_{\a}(\l\g_m)_{\b}$. 
Furthermore,
as BRST-exact terms  decouple from 
pure spinor correlations $\langle{\ldots}\rangle$,
equation \eqref{note} implies
\be\label{inte}
\langle
(\l \g^{mnpqr}\l){\cal F}^1_{mn}{\cal F}^2_{pq}{\cal F}^3_{rs}
(\l \g^s W^4)\rangle =
+16 ik^1_m \langle 
(\l\g^r W^1)(\l\g^m W^2)(\l\g^s W^4){\cal F}^3_{rs},
\rangle,
\ee
where we have used $k^1_m(\l\g_n W^1)(\l\g^r\g^{mn}W^2) =
-2 k^1_m(\l\g^r W^1)(\l\g^m W^2)$, which is valid when
the equation of motion $k^1_m(\g^m W^1)_{\a}=0$ is satisfied.

Using $(\l\g_m W^2) = QA^2_m - ik^2_m(\l A^2)$ and  
$\langle (\l\g^r W^1)Q(A_2^m)(\l\g^s W^4)
{\cal F}^3_{rs}\rangle = 0$ we arrive at the following pure
spinor superspace identity
\be\label{mardelpi}
\langle
(\l \g^{mnpqr}\l){\cal F}^1_{mn}{\cal F}^2_{pq}{\cal F}^3_{rs}
(\l \g^s W^4)\rangle =
- 16 (k^1\cdot k^2) \langle 
(\l A^2)(\l\g^r W^1)(\l\g^s W^4){\cal F}^3_{rs}
\rangle
\ee
Multiplying \eqref{mardelpi} by $\Delta(1,3)\Delta(2,4)$ and
summing over permutations leads to the following identity,
\be\label{nt}
K_2 = {32 \over 3}K_1 \left[(u-t)\Delta(1,2)\Delta(3,4) 
+(s-t)\Delta(1,3)\Delta(2,4) + (s-u)\Delta(1,4)\Delta(2,3)
\right],
\ee
where we used \eqref{dem} and the standard Mandelstam variables
$s=-2(k^1\cdot k^2)$, $t=-2(k^1\cdot k^4)$,  $u=-2(k^2\cdot k^4)$.

In view of the results from the next section, \eqref{nt} not only provides 
a simple proof of two-loop 
equivalence with the (bosonic) RNS result of \cite{dhoker} but it also
automatically implies the knowledge of the full amplitude, including
fermionic external states (which has not been computed in the RNS yet).

%********************************************************************
\section{The complete supersymmetric kinematic factors at tree-level,
one- and two-loops in components}
\label{complete_kin}
%********************************************************************

In this section the usefulness and the simplifying power of the pure 
spinor formalism become manifest. Using either the RNS or GS formalism,
the computation of the kinematic factor for all possible external state
combination allowed by supersymmetry for the massless four-point amplitude
at tree-level, one-loop and two-loop order would be a daunting task. In
fact, after more than two decades of effort, the kinematic factors
for fermionic states have never been explicitly computed at two-loops.
Using the pure spinor formalism derivation of the identities in section
\ref{ucla} this task becomes almost trivial, as shown below.

To obtain the complete supersymmetric kinematic factors for
the massless four-point amplitudes at tree-level, one- and two-loops
all one needs to do is to use the method of appendix \ref{chap_t}
to evaluate the pure spinor superspace expression of
\be
\label{kint}
K_0 = \half k_1^mk_2^n \langle (\l A^1)(\l A^2)(\l A^3){\cal F}^4_{mn}\rangle
- (k^1\cdot k^3)\langle A^1_m (\l A^2)(\l A^3)(\l\g^m W^4)\rangle + (1\leftrightarrow 2).
\ee
The first term doesn't contribute in the computation of $K_0(f_1f_2f_3f_4)\equiv K_0^{\rm 4F}$,
while the second leads to\footnote{I acknowledge the use of GAMMA \cite{ulf} and specially
of FORM \cite{FORM}\cite{tform} in these computations.}
$$
K_0^{\rm 4F} = -{1\over 9}(k^1\cdot k^3)\langle (\l\g^a \t)(\l\g^b \t)(\l\g^c \chi^4)
(\chi^3\g^b\t)(\t\g^c\chi^1)(\chi^2\g^a \t)\rangle +(1\leftrightarrow 2),
$$
$$
 = {1\over 5760}\left[(\chi^1\g^m\chi^2)(\chi^3\g_m\chi^4)\left[
(k^2\cdot k^3)-(k^1\cdot k^3)
\right]
-{1\over 12}(k^3\cdot k^4)(\chi^1\g^{mnp}\chi^2)(\chi^3\g_{mnp}\chi^4)
\right].
$$
Using the following Fierz identity 
$$
(\chi^1\g^{mnp}\chi^2)(\chi^3\g_{mnp}\chi^4)= 24(\chi^1\g^m\chi^3)(\chi^2\g_m\chi^4)
-12(\chi^1\g^m\chi^2)(\chi^3\g_m\chi^4),
$$ 
the four-fermion kinematic factor can be conveniently rewritten as
\be
\label{qf}
K_0^{\rm 4F} = -{1\over 2880}\left[
(k^1\cdot k^3)(\chi^1\g^m\chi^2)(\chi^3\g_m\chi^4)+
(k^3\cdot k^4)(\chi^1\g^m\chi^3)(\chi^2\g_m\chi^4)
\right].
\ee
Both terms of \eqref{kint} contribute in the $K_0^{\rm 2B2F} \equiv K_0(f_1f_2b_3b_4)$
kinematic factor,
$$
K_0^{\rm 2B2F} =
-{1\over 36} k^m_1k^n_2 F^4_{mn} e^3_p \langle
(\l\g^t \t)(\l\g^u\t)(\l\g^p\t)(\t\g_t\chi^1)(\chi^2\g_u\t)
\rangle
$$
$$
-{1\over 24}(k^1\cdot k^3)F^4_{mn} e^3_p \langle
(\l\g^t \t)(\l\g^p\t)(\l\g^q\g^{mn}\t)(\t\g_q\chi^1)(\chi^2\g_t\t)
\rangle + (1\leftrightarrow 2)
$$
\be
\label{lado}
= {1\over 5760}F^4_{mn}e^3_p\left[
k^m_1k_2^n(\chi^1\g^p\chi^2) + \half(k^1\cdot k^3) (\chi^1\g^{mn}\g^{p}\chi^2)
\right] + (1\leftrightarrow 2)
\ee
It is worth noticing that the explicit computation of $K_0^{\rm 2B2F}$ becomes 
easier if we use the identity \eqref{mapocha} with a convenient choice for the labels
in the right hand side, namely
$K_0 = -\langle (\l A^1)(\l\g^m W^3)(\l\g^n W^4){\cal F}^2_{mn}\rangle$, because
now one can check that only one term contributes
$$
K_0^{\rm 2B2F} = {1\over 24}\langle
(\l\g^p \t)(\l\g^{[m|}\gamma^{rs}\t)(\l\g^{|n]}\gamma^{tu}\t)(\t\g_p\chi^1)(\chi^2\g_n\t)
\rangle k^2_m F^3_{rs}F^4_{tu}
$$
\be
\label{outro}
=  {1\over 5760}F^3_{mn}F^4_{rs}\left[
- i(\chi^1\g^r \chi^2)\eta^{sm}  k_2^n 
+ {i\over 2} (\chi^1\g^{mnr} \chi^2)k_2^s
\right] 
+ (3\leftrightarrow 4).
\ee
One can verify that \eqref{lado} and \eqref{outro} are in fact 
equal and equivalent to the RNS
result (see for example \cite{siegel_lee}).
This equality can also be regarded as a check of identity \eqref{mapocha}, which
is reassuring.

The computation of $K_0^{\rm 4B}$ is straightforward (and can also be deduced
from the one-loop result of \cite{mafra_one}). One can in fact check that
$$
K_0^{\rm 4B} = {1\over 5760}\left[
- \half (e^1\cdot e^3)(e^2\cdot e^4)ts
- \half (e^1\cdot e^4)(e^2\cdot e^3)us
- \half (e^1\cdot e^2)(e^3\cdot e^4)tu \right.
$$
$$
+(k^4\cdot e^1)(k^2\cdot e^3)(e^2\cdot e^4)s+
(k^3\cdot e^2)(k^1\cdot e^4)(e^1\cdot e^3)s
$$
$$
+(k^3\cdot e^1)(k^2\cdot e^4)(e^2\cdot e^3)s+
(k^4\cdot e^2)(k^1\cdot e^3)(e^1\cdot e^4)s
$$
$$
+(k^1\cdot e^2)(k^3\cdot e^4)(e^1\cdot e^3)t+
(k^4\cdot e^3)(k^2\cdot e^1)(e^2\cdot e^4)t
$$
$$
+(k^4\cdot e^2)(k^3\cdot e^1)(e^3\cdot e^4)t+
(k^1\cdot e^3)(k^2\cdot e^4)(e^1\cdot e^2)t
$$
$$
+(k^2\cdot e^1)(k^3\cdot e^4)(e^2\cdot e^3)u+
(k^4\cdot e^3)(k^1\cdot e^2)(e^1\cdot e^4)u
$$
$$
\left.
+ (k^4\cdot e^1)(k^3\cdot e^2)(e^3\cdot e^4)u+
(k^2\cdot e^3)(k^1\cdot e^4)(e^1\cdot e^2)u {\phantom\half}\right]
$$
\be
\label{final}
 = {1\over 2880} t_8^{m_1n_1m_2n_2m_3n_3m_4n_4}
 F^1_{m_1n_1}F^2_{m_2n_2}F^3_{m_3n_3}F^4_{m_4n_4},
\ee
where we used the $t_8$ tensor definition 
of \cite{gswI}\cite{gswII}.

%**********************************
\section{Anomaly kinematic factor}
\label{anomaly_sec}
%**********************************

Type-I superstring theory is defined in ten dimensional space and contains
gluinos of only one chirality. So it could be plagued by an anomaly which would
reveal itself as the failure 
of the massless six-point amplitude to be gauge invariant, therefore making
the theory inconsistent at the quantum level.
One could even expect type-I theory to be anomalous, considering the fact that
$N=1$ super-Yang-Mills in $D=10$ definitely has an anomaly and it is recovered
as the low-energy limit of the type-I superstring. But Green and Schwarz showed
in a seminal 1984 paper that superstring theory is anomaly-free \cite{green_anomalia}, setting free
the spark which lighted the wondrous fire of the so-called First Superstring Revolution.

One way to understand the absence of the anomaly is to note that to compute
the massless six-point amplitude in type-I superstring theory one has to sum over three
different manifolds. That is because at genus one the worldsheet can be the planar (or non-planar) cylinder
or the M\"obius strip. So there is a chance of getting a vanishing gauge transformation
if the sum of the non-vanishing partial amplitudes cancel out. And that is exactly what happens. So
it is the extended nature of the string -- its propagation through space-time
sweeps a Riemann surface -- which allows the cancellation of the gauge variation. And now one
can understand why SYM is anomalous; being the low-energy limit of type-I superstring theory
it looses the information about the size of the string to become a point, and the three
different worldsheet configurations become only one Feynman graph along the way, making
a cancellation impossible.

One can compute the superstring scattering for each worldsheet configuration 
individually and the kinematic factor turns out to be the same for all of them and is
given by
\be
\label{eps5}
K = \e^{m_1n_1{\ldots} m_5n_5} k^1_{m_1}e^1_{n_1}{\ldots}k^5_{m_5}e^5_{n_5}.
\ee

In the following sections it will be shown that the non-minimal 
pure spinor formalism
computation of the hexagon gauge anomaly in the Type-I superstring 
is equivalent to the RNS result of \cite{green_anomalia}. As will be
proved below, the kinematic factor of the hexagon gauge variation can be
written as the pure spinor superspace integral
$$
K = \vev{(\l\g^mW^2)(\l\g^nW^3)(\l\g^pW^4)(W^5\g_{mnp}W^6)},
$$
whose bosonic part will be computed to
demonstrate that it is the well-known $\e_{10}F^5$ RNS result of \eqref{eps5}.

As discussed in \cite{PolchinskiTU}\cite{FMS}, the anomaly can be
computed as a surface term which contributes
at the boundary of moduli space.
The result can be separated in two parts: the kinematic factor
depending only on momenta and polarizations,
and the moduli space part which depends on the worldsheet surface. 
We will treat them separately, first we will compute the kinematic
factor and afterwards we will say a few words about the moduli
space part.

%******************************************
\subsection{Kinematic factor computation}
\label{anoma_exp}
%******************************************

In the type-I superstring theory with gauge group SO(N),
the massless open string six-point 
one-loop amplitude is given by
\be
\label{sixpt}
{\cal A} = \sum_{top=P,NP,N}
G_{top}\int_0^{\infty} dt \langle{ {\cal N} \int dw b(w)
    (\l A_1)\prod_{r=2}^6 \int dz_r U_r (z_r)}\rangle
\ee
where 
$P,NP,N$
denotes the three possible
different world-sheet topologies, each of which has a different
group-factor $G_{top}$ \cite{gswII}. When 
all particles are attached to one boundary, we have a cylinder with
$G_P=N\tr{(t^{a_1}t^{a_2}t^{a_3}t^{a_4}t^{a_5}t^{a_6})}$.
When particles are attached to both
boundaries, the diagram is a non-planar 
cylinder, where $G_{NP}=\tr{(t^{a_1}t^{a_2})}\tr{(t^{a_3}t^{a_4}t^{a_5}t^{a_6})}$.
And finally, there is the non-orientable M\"obius strip where
$G_{N}=-\tr{(t^{a_1}t^{a_2}t^{a_3}t^{a_4}t^{a_5}t^{a_6})}$.

We will be interested in the amplitude when all external states are massless
gluons with polarization $e^r_m$ {\it i.e.}, $a^r_m(x) = e_m^r e^{ik\cdot x}$, 
where $m=0,{\ldots} 9$ is the
space-time vector index and $r$ is the particle label
\footnote{We will omit the adjoint gauge group index
from the polarizations and field-strengths for the rest of this section.}.
To probe the anomaly, 
one can compute \eqref{sixpt} and substitute one of the external polarizations
for its respective momentum. However,
instead of first computing the six-point amplitude and substituting $e_m\rightarrow k_m$
in the answer, we will first make the gauge transformation
in \eqref{sixpt} and then compute the 
resulting correlation function. 
This will give us the anomaly kinematic factor directly.

Under the super-Yang-Mills gauge transformation \eqref{gaugetr}
\be\label{gauges}
\d A_\a = D_\a \Omega, \quad \d A_m = \p_m \Omega,
\ee
the integrated vertex operator $\int dz U$ changes by the surface term
$\int dz \d U = \int dz \p \Omega$, and the unintegrated vertex operator
changes by the BRST-trivial quantity $\d(\l A) = \l^\a D_\a \Omega = Q\Omega.$
Choosing $\Omega(x,\t) =e^{ik\cdot x}$ has the same effect as changing $e^m \to k^m$,
which is the desired gauge transformation to probe the anomaly.

To compute the gauge anomaly, it will be convenient to choose
the gauge transformation to act on the polarization $e_m^1$ in the
unintegrated vertex operator, so that
the gauge variation of \eqref{sixpt} is 
\be
\label{variation}
\delta {\cal A} = \sum_{top=P,N,NP}G_{top}\int_0^{\infty} dt 
                  \langle {\cal N} \int dw b(w)
                  (Q\Omega(z_1))\prod_{r=2}^6 \int dz_r U_r(z_r) \rangle.
\ee
Integrating $Q$ by parts inside the 
correlation function will only get
a contribution from the BRST variation of the $b$-ghost,
which is
a derivative with respect to the modulus \cite{dhoker_geom}\cite{FMS}. 
This is due to the fact that $\{Q, b\} = T$ together with
the definition of the Beltrami differential as parametrizing the violation
of the conformal gauge under a variation of the metric tensor
\[
\d g = \d g_{z{\bar z}}dzd{\bar z} + \d g_{zz}dzdz
\]
where if $\tau$ is the modulus parameter then
\[
\d g_{zz} = \mu_z^{\phantom{z}{\bar z}} g_{z{\bar z}}\d\tau.
\]
In this way, having the insertion of $\int d^2z T_{z{\bar z}}\mu_z^{\phantom{z}{\bar z}}$
in a correlation function is equivalent to derive it with respect to the modulus because
\[
\frac{\d S}{\d \tau} = \frac{\d S}{\d g^{zz}}\frac{\d g^{zz}}{\d \tau} = \int d^2z T\cdot \mu.
\]
So \eqref{variation} becomes
\begin{align}
\label{vartwo}
\d{\cal A} = & -\sum_{top}G_{top} \int_0^{\infty} dt {d\over dt}
                    \langle\Omega(z_1) {\cal N} \prod_{r=2}^6
\int dz_r U_r(z_r)\rangle \\
\label{modul}
      \equiv & - K\sum_{top}G_{top} \Big[
	         B_{top}(\infty)-B_{top}(0)
		 \Big], 
\end{align}	   
where the moduli space part of the anomaly is encoded in the 
function  
$$
B_{top}(t) \equiv \int_0^{t} dz_6\int_0^{z_6} 
dz_5\int_0^{z_5} dz_4\int_0^{z_4} dz_3
                  \int_0^{z_3} dz_2 \,
                  \langle \prod_{r=1}^6 :e^{ik_r\cdot x_r}:\rangle_{top},
$$
and $K=\langle{{\cal N}U_2U_3U_4U_5U_6}\rangle$. 
From \eqref{vartwo}, it is clear that the
anomaly comes from the boundary of moduli space.

To compute the kinematic factor $K$, 
observe that there is an unique way to absorb the 16 zero 
modes of $d_{\alpha}$, 11 of $s^{\a}$ and 11 of $r_{\a}$. 
The regularization factor ${\cal N}$
must provide 11 $d_{\alpha}$,
11 $s^{\alpha}$ and 11 $r_{\a}$ zero modes. The five remaining $d_{\a}$ zero modes
must come from the external vertices\footnote{It follows from this
zero mode counting that the anomaly trivially vanishes for amplitudes with less
than six external massless particles.}
through $(d {\cal W})^5$. As in the computations of the previous section,
the kinematic factor is thus given by a pure spinor superspace
integral involving 3 $\l$'s
and 5 $W$'s, as can be easily verified by integrating all the zero mode measures
except $[d\l],[d\lb]$ and $[dr]$.
To find out how the indices are contracted in $K$,
choose the reference frame where only $\l^+ \neq 0$. Then one can easily check
that the unique $U(1)\times SU(5)$-invariant contraction is 
$$
K= \vev{ (\l^+)^3\e_{abcde}W_2^aW_3^bW_4^cW_5^dW_6^e
   },
$$
which in SO(10)-covariant notation translates into
\be
\label{fator}
 K = \langle (\l\g^m W_2)(\l\g^n W_3)(\l\g^p W_4)(W_5\g_{mnp} W_6)\rangle.
\ee

Now one needs to check if the pure spinor expression \eqref{fator}
reproduces the known $\e_{10}F^5$ contraction computed with the
RNS formalism \cite{green_anomalia}.

When all external states are gluons, there is
only one possibility to obtain the required five $\t$'s from
the pure spinor measure $\vev{\l^3\t^5}$. Using the superfield 
expansions from appendix \ref{sym_ap} we see that each
superfield $W^{\a}(\t)$ must contribute one $\t$ through
the term $-{1\over 4}(\g^{mn}\t)^{\a}F_{mn}$. 
Therefore the kinematic factor \eqref{fator} is equal to
\be
\label{kinem}
= -\frac{1}{1024}
\vev{ (\l\g^p\g^{m_2n_2}\t)(\l\g^q\g^{m_3n_3}\t)(\l\g^r\g^{m_4n_4}\t)
          (\t\gamma^{m_5n_5}\g_{pqr}\g^{m_6n_6}\t)}F^2_{m_2n_2}{\ldots} 
	  F^6_{m_6n_6}.
\ee
We will now demonstrate the equivalence with the RNS anomaly result of 
\cite{green_anomalia} by proving that 
\be
\label{impor}
\langle (\l\g^p\g^{m_1n_1}\t)(\l\g^q\g^{m_2n_2}\t)(\l\g^r\g^{m_3n_3}\t)
     (\t\g^{m_4n_4}\g_{pqr}\g^{m_5n_5}\t) \rangle
     = {1\over 45} \e^{m_1n_1{\ldots}m_5n_5}.
\ee

We will first
show that the correlation in \eqref{impor} is proportional
to $\e_{10}$ by checking its behavior
under a parity transformation.
Using the language of \cite{nathan_multiloop}, we can rewrite \eqref{impor} as
\be
\label{compl}
(T^{-1})^{(\a\b\g)[\rho_1\rho_2\rho_3\rho_4\rho_5]}
T_{(\a\b\g)[\d_1\d_2\d_3\d_4\d_5]}(\g^{m_1n_1})^{\d_1}_{\,\,\rho_1}
(\g^{m_2n_2})^{\d_2}_{\,\,\rho_2}(\g^{m_3n_3})^{\d_3}_{\,\,\rho_3}
(\g^{m_4n_4})^{\d_4}_{\,\,\rho_4}(\g^{m_5n_5})^{\d_5}_{\,\,\rho_5},
\ee
where $T$ and $T^{-1}$ are defined by
\be
\label{deftt}
(T^{-1})^{(\a_1\a_2\a_3)[\d_1\d_2\d_3\d_4\d_5]}=
 (\g^m)^{\a_1\d_1}(\g^n)^{\a_2\d_2}(\g^p)^{\a_3\d_3}(\g_{mnp})^{\d_4\d_5}
\ee
$$
T_{(\a_1\a_2\a_3)[\d_1\d_2\d_3\d_4\d_5]} = 
 \g^m_{\a_1\d_1}\g^n_{\a_2\d_2}\g^p_{\a_3\d_3}(\g_{mnp})_{\d_4\d_5},
$$
and the $\a$-indices are symmetric and gamma matrix traceless, and the
$\d$-indices are antisymmetric. Since
a parity transformation has the effect of changing
a Weyl spinor $\psi^{\a}$ to an anti-Weyl spinor $\psi_{\a}$,  
it follows from the definitions of \eqref{deftt} that a parity transformation 
exchanges $T \leftrightarrow T^{-1}$. Furthermore, since a
parity transformation also changes
$$
(\g^{mn})^{\d}_{\,\,\rho}\rightarrow (\g^{mn})_{\d}^{\,\,\rho} = 
- (\g^{mn})_{\,\,\d}^{\rho},
$$
it readily follows that the kinematic factor \eqref{compl}
is odd under parity, 
so it is proportional to $\e_{10}$. Finally, 
the proportionality constant of ${1\over {45}}$ in \eqref{impor} can be 
explicitly computed using the identities
listed in Appendix \ref{sym_ap}. 

%****************************************************
\subsection{The evaluation of the moduli space part}
%****************************************************

In this section we compute the limits of $B_{top}(t)$ when the topology of 
the open string world-sheet is a cylinder, where
$B_{top}(t)$ is given by
\be
\label{topological}
B_{top}(t) \equiv \int_0^{\t} dw_5\int_0^{w_5} dw_4\int_0^{w_4} dw_3\int_0^{w_3} dw_2
                  \int_0^{w_2} dw_1 \,
                  \langle \prod_{r=1}^6 :e^{ik_r\cdot X_r}:\rangle_{top},
\ee
and $\t=it$.
Note that to obtain \eqref{topological} we did not need to
introduce any regularization factor such as the Pauli-Villars 
used in \cite{gswII}.
Furthermore, in the 
RNS computations of \cite{gswII} the fact 
that the anomaly is an effect from the boundary of moduli space is not evident, and was 
first pointed
out in \cite{FMS}. This was later verified to be the case in \cite{PolchinskiTU} and
\cite{liu}. Using the non-minimal pure spinor formalism it is clear
from \eqref{modul} that
the anomaly comes from the boundary of moduli space. 
Although reference \cite{liu} uses a different formalism, the
evaluation of the moduli space limits of \eqref{topological} require almost identical
manipulations as the ones described there, so this subsection is heavily based
on it. The final result is obviously the same, concluding the proof that
the non-minimal pure spinor formalism analysis of the gauge anomaly in 
type-I superstring is equivalent to the RNS.

When the topology is a cylinder (P) we have \cite{liu}
$$
\langle 
	\prod_{r=1}^6 :e^{ik_r\cdot X_r}:
\rangle_{P} \sim 
   {1\over t^5} \prod_{i< j}\Big[\psi_C(w_{ij}|\tau)\Big]^{k_i\cdot k_j},
$$
where
$$
\psi_C(w_{ij}|\tau) = -ie^{(i\pi w_{ij}^2/\tau)} 
                      {\vartheta_1(w_{ij}|\tau)\over \eta^3(\tau)},
$$
and 
$$
\vartheta_1(\nu|\tau) = 2q^{1/8}\sin(\pi\nu)\prod_{n=1}^{\infty}
                       (1-q^n)(1-q^ne^{2\pi i\nu})
                       (1-q^ne^{-2\pi i\nu}),
$$
$$
\eta(\tau) = q^{1/24}\prod_{n=1}^{\infty}(1-q^n),
$$
with $q=\exp{(2\pi i \tau)}$. Note that
\be
\label{prev}
    \lim_{t\to\infty}\psi_C(w_{ij}|\tau) = 2\sin(\pi w_{ij}),
\ee
implies the vanishing of $\lim_{t\to\infty}B(t)$. To obtain the limit when $t\to 0$
we make the following change of variables $\tilde{w}={w\over\tau}$ and
$\tilde{\tau}\equiv i \tilde{t}=
-{1\over \tau}$ to get 
$$
B_P(t)\equiv 
\int_0^1 d\tilde{w}_5\int_0^{\tilde{w}_5} 
d\tilde{w}_4\int_0^{\tilde{w}_4} d\tilde{w}_3
\int_0^{\tilde{w}_3} d\tilde{w}_2\int_0^{\tilde{w}_2} d\tilde{w}_1
\prod_{i\le j}\Big[\psi_C(\tilde{w}_{ij}|\tilde{\tau})\Big]^{k_i\cdot k_j},
$$
where
$$
\psi_C(\tilde{w}_{ij}|\tilde{\tau})=-i\exp{(-i\pi\tilde{w}_{ij}/\tilde{\tau})}
{ 
\vartheta_1(-{\tilde{w}_{ij}\over \tilde{\tau}}|-{1\over\tilde{\tau}})
\over
\eta^3(-{1\over\tilde{\tau}})
}
$$
\be
\label{expl}
= {i\over \tilde{\tau}} {\vartheta_1(\tilde{w}_{ij}|\tilde{\tau})\over
\eta^3(\tilde{\tau})}.
\ee
To obtain \eqref{expl} we used the well-known modular transformation properties,
$$
\vartheta_1(-{\nu \over \tau}|-{1\over \tau}) = -i\sqrt{-i\tau}\exp{(i\pi\nu^2/\tau)}
\vartheta_1(-\nu|\tau)
$$
$$
\eta(-{1\over\tau})=\sqrt{-i\tau}\eta(\tau),
$$
and $\vartheta_1(-\nu|\tau)=-\vartheta_1(\nu|\tau)$. Note that the factor of
$i/\tilde{\tau}$ in \eqref{expl} will not contribute because of momentum conservation
and masslessness of the external particles,
$$
\prod_{i< j} \left({i\over \tilde{\tau}}\right)^{k_i\cdot k_j} =
\left({i\over \tilde{\tau}}\right)^{\sum_{i< j}k_i\cdot k_j} = 1.
$$
Noting that the limit is now $\tilde{t}\to\infty$, we can use the previous
result \eqref{prev} to finally obtain 
\be
\label{partedois}
\lim_{t\to 0}B_P(t) = 
\int_0^1 d\tilde{w}_5\int_0^{\tilde{w}_5} 
d\tilde{w}_4\int_0^{\tilde{w}_4} d\tilde{w}_3
\int_0^{\tilde{w}_3} d\tilde{w}_2\int_0^{\tilde{w}_2} d\tilde{w}_1
\prod_{i<j}\Big[2\sin(\pi w_{ij})\Big]^{k_i\cdot k_j}.
\ee

The final expression for the gauge variation of the amplitude in the planar cylinder
is finally
\be
\label{fim}
\delta{\cal A} =  K G_P
\int_0^1 d\tilde{w}_5\int_0^{\tilde{w}_5} 
d\tilde{w}_4\int_0^{\tilde{w}_4} d\tilde{w}_3
\int_0^{\tilde{w}_3} d\tilde{w}_2\int_0^{\tilde{w}_2} d\tilde{w}_1
\prod_{i<j}\Big[\sin(\pi \tilde{w}_{ij})\Big]^{k_i\cdot k_j},
\ee
where $K$ is given by \eqref{anomaly}. When the world-sheet is a M\"obius strip the
kinematical factor is exactly the same. The limits of \eqref{topological}
in the boundary of moduli space are however different.
We will not repeat the computation here and merely
quote the result of \cite{gswII}\cite{liu} that $\lim_{t\to 0}B_{N}(t)=
32\lim_{t\to 0} B_P(t)$. From this it follows that the sum of the anomalies
for the planar and non-orientable diagrams vanish if $N=32$, i.e., if the
gauge group is SO(32). For the non-planar cylinder one can show that
\eqref{topological} vanishes \cite{gswII}. 

%***************************************************
\subsection{The fermionic expansion}
\label{total_deriv}
%***************************************************

From the explicit bosonic computation  
of the anomaly kinematic factor of section \ref{anoma_exp} we know
that it is a total derivative, because
\[
\e_{10}^{mnpqrstuvx}F_{mn}F_{pq}F_{rs}F_{tu}F_{vx} = 
2 \p_m \big[\e_{10}^{mnpqrstuvx}A_n F_{pq}F_{rs}F_{tu}F_{vx}\big].
\]
What about the fermionic components? Based on supersymmetry we
expect them to be total derivatives also, and in this section
we explicitly compute the $2F3B$ components of \eqref{fator} to
show that this is indeed true.

If we choose the particle distribution to be $f_0f_1b_2b_3b_4$
then we obtain for
\[
K = 
\langle 
(\l\g^m W_2)(\l\g^n W_3)(\l\g^p W_4)(W_0\g_{mnp} W_1)
\rangle 
\]
the following distribution of thetas
\begin{center}
\begin{tabular}{|c|c|c|c|c|} \hline
$W^2(\t)$ & $W^3(\t)$ & $W^4(\t)$ & $W^0(\t)$ & $W^1(\t)$ \\ \hline
3 & 1 & 1 & 0 & 0 \\ \hline
1 & 1 & 3 & 0 & 0 \\ \hline
1 & 3 & 1 & 0 & 0 \\ \hline
1 & 1 & 1 & 2 & 0 \\ \hline
1 & 1 & 1 & 0 & 2 \\ \hline
\end{tabular}
\end{center}
Using the component expansion of Appendix \ref{sym_ap} and the identity
\[
\psi^\a \chi^\b = \frac{1}{16}\g_m^{\a\b}(\psi \g^m \chi)
+ \frac{1}{96}\g_{mnp}^{\a\b}(\psi \g^{mnp} \chi)
+ \frac{1}{3840}\g_{mnpqr}^{\a\b}(\psi \g^{mnpqr} \chi)
\]
the above table translates to
\[
256 K = 
\]
\[
+ \Big[ \frac{1}{3}\langle 
(\l\g^m \g^{rt}\t)(\l\g^n \g^{m_3n_3}\t)(\l\g^p \g^{m_4n_4}\t)
(\t\g^{tm_2n_2}\t)
\rangle \; (\chi^0\g_{mnp}\chi^1)k^2_r 
\]
\[
+ \frac{1}{3}\langle 
(\l\g^m \g^{rt}\t)(\l\g^n \g^{m_2n_2}\t)(\l\g^p \g^{m_4n_4}\t)
(\t\g^{tm_3n_3}\t)
\rangle\; (\chi^0\g_{mnp}\chi^1)k^3_r
\]
\[
+ \frac{1}{3}\langle 
(\l\g^m \g^{rt}\t)(\l\g^n \g^{m_2n_2}\t)(\l\g^p \g^{m_3n_3}\t)
(\t\g^{tm_4n_4}\t)
\rangle\; (\chi^0\g_{mnp}\chi^1)k^4_r
\]
\[
- \frac{1}{16}\langle 
(\l\g^m \g^{m_2n_2}\t)(\l\g^n \g^{m_3n_3}\t)(\l\g^p \g^{m_4n_4}\t)
(\t\g^{ka}\g_{mnp}\g^r\g_a\t)
\rangle\; k^0_k(\chi^0\g_{r}\chi^1)
\]
\[
+ \frac{1}{16}\langle
(\l\g^m \g^{m_2n_2}\t)(\l\g^n \g^{m_3n_3}\t)(\l\g^p \g^{m_4n_4}\t)
(\t\g^{ka}\g_{mnp}\g^r\g_a\t)
\rangle\; k^1_k(\chi^0\g_{r}\chi^1)
\]
\[
+ \frac{1}{96}\langle 
(\l\g^m \g^{m_2n_2}\t)(\l\g^n \g^{m_3n_3}\t)(\l\g^p \g^{m_4n_4}\t)
(\t\g^{ka}\g_{mnp}\g^{rst}\g_a\t)
\rangle\; k^0_k(\chi^0\g_{rst}\chi^1)
\]
\[
+ \frac{1}{96}\langle 
(\l\g^m \g^{m_2n_2}\t)(\l\g^n \g^{m_3n_3}\t)(\l\g^p \g^{m_4n_4}\t)
(\t\g^{ka}\g_{mnp}\g^{rst}\g_a\t)
\rangle\; k^1_k(\chi^0\g_{rst}\chi^1)
\]
\[
- \frac{1}{3840}\langle 
(\l\g^m \g^{m_2n_2}\t)(\l\g^n \g^{m_3n_3}\t)(\l\g^p \g^{m_4n_4}\t)
(\t\g^{ka}\g_{mnp}\g^{rstuv}\g_a\t)
\rangle\; k^0_k(\chi^0\g_{rstuv}\chi^1)
\]
\[
+ \frac{1}{3840}\langle 
(\l\g^m \g^{m_2n_2}\t)(\l\g^n \g^{m_3n_3}\t)(\l\g^p \g^{m_4n_4}\t)
(\t\g^{ka}\g_{mnp}\g^{rstuv}\g_a\t)
\rangle\; k^1_k(\chi^0\g_{rstuv}\chi^1)\Big]
\]
\be
\label{2f3b}
\times F^2_{m_2n_2}F^3_{m_3n_3}F^4_{m_4n_4}.
\ee
It will be instructive to compute the correlator by parts, separating
them according to the factors of $(\chi^0\g_{r}\chi^1)$, $(\chi^0\g_{rst}\chi^1)$ or
$(\chi^0\g_{rstuv}\chi^1)$.

The terms containing $(\chi^0\g_{r}\chi^1)$ are given by $K_r = $
\[
-\frac{1}{16}(\chi^0\g_{r}\chi^1)(k^0_k - k^1_k)F^2F^3F^4
\langle 
(\l\g^m \g^{m_2n_2}\t)(\l\g^n \g^{m_3n_3}\t)(\l\g^p \g^{m_4n_4}\t)
(\t\g^{ka}\g_{mnp}\g^r\g_a\t)
\rangle 
\]
which can be rewritten using the tensor \eqref{fratres}
computed in section \ref{t8_e10},
\[
K_r = + \frac{1}{16}(\chi^0\g_{r}\chi^1)(k^0_k - k^1_k) 
t_{10}^{arm_2n_2m_3n_3m_4n_4kb}\,\d_b^a
\]
\be
\label{gr}
 = -\frac{1}{360}t_8^{krm_2n_2m_3n_3m_4n_4}
(k^0_k - k^1_k) (\chi^0\g_{r}\chi^1)
F^2_{m_2n_2}F^3_{m_3n_3}F^4_{m_4n_4}.
\ee
To compare \eqref{gr} with the terms proportional to 
$(\chi^0\g_{rst}\chi^1)$ in \eqref{2f3b} it is convenient
to use Dirac's equation such that
\[
k^0_k (\chi^0 \g^{kmn} \chi^1) = - 2 k^0_{[m}(\chi^0\g_{n]}\chi^1),
\quad
k^1_k (\chi^0 \g^{kmn} \chi^1) = + 2 k^1_{[m}(\chi^0\g_{n]}\chi^1)
\]
which allows one 
to rewrite \eqref{gr} as follows
\[
256K_r = + \frac{1}{720}t_8^{mnm_2n_2m_3n_3m_4n_4}
(k^0_k + k^1_k) (\chi^0\g_{kmn}\chi^1)
F^2_{m_2n_2}F^3_{m_3n_3}F^4_{m_4n_4}.
\]
The terms in \eqref{2f3b} proportional to $(\chi^0\g_{rst}\chi^1)$
can be computed similarly and we obtain
\[
256K_{rst} = + \frac{1}{720}t_8^{mnm_2n_2m_3n_3m_4n_4}
(k^2_k + k^3_k + + k^4_k) (\chi^0\g_{kmn}\chi^1)
F^2_{m_2n_2}F^3_{m_3n_3}F^4_{m_4n_4},
\]
and therefore 
\be
\label{parcial1}
256(K_r + K_{rst}) =
 + \frac{1}{720}t_8^{mnm_2n_2m_3n_3m_4n_4}\sum_{i=0}^4 k^i_k
% (k^0_k + k^1_k + k^2_k + k^3_k + k^4_k)
 (\chi^0\g_{kmn}\chi^1)
F^2_{m_2n_2}F^3_{m_3n_3}F^4_{m_4n_4}.
\ee
Another long computation for the terms proportional 
to $(\chi^0\g_{rstuv}\chi^1)$ results in the following
\[
256K_{rstuv} = -\frac{53}{25200}\sum_{i=0}^4 k^i_k
(\chi^0 \g^{km_2n_2m_3n_3m_4n_4}\chi^1)F^2_{m_2n_2}F^3_{m_3n_3}F^4_{m_4n_4},
\]
therefore we have show that the terms containing two fermions and
three bosons in \eqref{2f3b} combine into a total derivative.

%****************************************************
%***************************************************
\section{$t_8$ and $\e_{10}$ from pure spinor superspace}
\label{t8_e10}
%***************************************************
%***************************************************

In this section we digress about an interesting identity which
involves both
the $t_8$ and $\e_{10}$ tensors, showing how closely related they are
when obtained from pure spinor superspace integrals.
This is different from computations in the RNS formalism where 
$t_8$ and $\e_{10}$ come from correlation functions with different
spin structures.

Since the one-loop $t_8 F^4$ and $\e_{10} B F^4$
terms are expected to be related by non-linear supersymmetry,
there might be a common superspace origin for the $t_8$ and $\e_{10}$
tensors.
This suggests looking for a BRST-closed pure spinor superspace
integral involving four super-Yang-Mills superfields whose bosonic
part involves both the $t_8$ and $\e_{10}$ tensors. One such BRST-closed 
expression is 
\be
\label{found}
\langle(\l\g^rW^1)(\l\g^sW^2)(\l\g^tW^3)
(\t\g^m\g^n\g_{rst}W^4)\rangle.
\ee
Although \eqref{found} is not spacetime supersymmetric because of the explicit
$\t$, it might be related to a supersymmetric expression
in a constant background where the $N=1$ supergravity superfield
$G_{m\a}$ satisfies $G_{m\a}=\g_{m\a\b}\t^\b + b_{mn} (\g^n\t)_\a$
for constant $b_{mn}$.

When restricted to its purely bosonic part,
\eqref{found}
defines the following 10-dimensional
tensor:
\be
\label{fratres}
t_{10}^{mnm_1n_1m_2n_2m_3n_3m_4n_4} = \vev{ (\l\g^a \g^{m_1n_1}\t)(\l\g^b \g^{m_2n_2}\t)
	          (\l\g^c \g^{m_3n_3}\t)(\t\g^m\g^n\g_{abc}\g^{m_4n_4}\t)}.
\ee
Using $\g^m\g^n = \g^{mn}+\eta^{mn}$ we obtain
\[
t_{10}^{mnm_1n_1m_2n_2m_3n_3m_4n_4} =
+ \langle (\l\g^a \g^{m_1n_1}\t)(\l\g^b \g^{m_2n_2}\t)
	          (\l\g^c \g^{m_3n_3}\t)(\t\g^{mn}\g_{abc}\g^{m_4n_4}\t)\rangle
\]		  
\be
\label{idt}
+\eta^{mn}\langle (\l\g^a \g^{m_1n_1}\t)(\l\g^b \g^{m_2n_2}\t)
	          (\l\g^c \g^{m_3n_3}\t)(\t\g_{abc}\g^{m_4n_4}\t) \rangle.
\ee
And using the identities listed in appendix \ref{sym_ap}, 
one can check that\footnote{The sign
in front of $\e_{10}$ depends on the chirality of $\t$. For an anti-Weyl $\t_{\a}$,
the sign is ``+''.}
\be
\label{ident}
t_{10}^{mnm_1n_1m_2n_2m_3n_3m_4n_4} = 
       -  {2\over 45}\Big[ \eta^{mn}t_8^{m_1n_1m_2n_2m_3n_3m_4n_4}
                      - {1\over 2}\e^{mnm_1n_1m_2n_2m_3n_3m_4n_4}\Big].
\ee

It is also interesting to contrast the similarity between $\e_{10}$ and
$t_8$ when written in terms of the $T$ and $T^{-1}$ tensors:
\[
\epsilon^{mnm_1n_1{\ldots}m_4n_4}  \propto
(T^{-1})^{(\a\b\g)[\rho_0\rho_1\rho_2\rho_3\rho_4]}
T_{(\a\b\g)[\d_0\d_1\d_2\d_3\d_4]}(\g^{mn})^{\d_0}_{\,\,\rho_0}
{\ldots} 
(\g^{m_4n_4})^{\d_4}_{\,\,\rho_4} 
\]
\[
t_8^{m_1n_1{\ldots}m_4n_4}  \propto
(T^{-1})^{(\a\b\g)[\kappa\rho_1\rho_2\rho_3\rho_4]}
T_{(\a\b\g)[\kappa\d_1\d_2\d_3\d_4]}
(\g^{m_1n_1})^{\d_1}_{\,\,\rho_1}{\ldots} 
(\g^{m_4n_4})^{\d_4}_{\,\,\rho_4},
\]
which shows, in a pure spinor superspace language, how one can ``obtain''
the $t_8$ tensor from $\e_{10}$: it is a matter of removing 
$(\g^{mn})^{\d_0}_{\,\,\rho_0}$ and contracting the associated spinorial
indices in $T$ and $T^{-1}$. 
So when using pure spinors, there is a close relation between these
two different-looking tensors.

%*****************************************************************
\section{Massless open string five-point amplitude at one-loop}
\label{fracas}
%*****************************************************************

In this section we will describe some of the ongoing work with Christian
Stahn to obtain
the kinematic factor for the massless five-point amplitude for 
open strings\footnote{This is the status of the computation as of September 2008.
It was later completely finished and the agreement with bosonic computations
obtained with other formalisms demonstrated \cite{5pt}.}.

\subsection{Modified NMPS computation}

First of all I will show how one can get the correct answers for
the bosonic components using the following substitution rule
\be
\label{bad}
\int d^2z b\cdot\mu V^1 \longrightarrow (d^0\t)\int dz U^1(z)
\ee
where $d_{\a}^0$ means only the zero mode.
Using \eqref{bad} the amplitude prescription of
\be
\label{unmod}
{\cal A} = \int dt \langle
{\cal N} (\int b\cdot \mu) V^1 U^2  U^3 U^4 U^5 \rangle
\ee
becomes
\be
\label{modif}
{\cal A} = \int dt \langle
{\cal N}  (d\t) U^1 U^2  U^3 U^4 U^5 \rangle
\ee
where the regularization factor is given by \eqref{reg_loop}
\[
{\cal N} = \exp\big[-(\l\lb) - (w{\bar w}) -(r\theta) + (sd)\big].
\]
To saturate the 16 zero-modes of $d_{\a}$ and 11 of $s^{\a}$ which are present at
the one-loop level there is only one possibility using the
modified prescription of \eqref{modif}. The regularization factor ${\cal N}$ provides 
11 zero-modes of
$s^{\a}$ and $d_{\a}$. Note that there is already a zero mode in the factor $(d^0 \t)$ so
the remaining four $d_{\a}$ zero modes can come from the five external vertices through
two different ways, depending on which term of the integrated vertex is chosen to act
by its OPE over the others. The kinematic factor then becomes
\be
\label{inic}
K = \langle (d\t)(d W^1)(d W^2)(d W^3)(d W^4)\(A^5_m\Pi^m + {\hat d}_{\a}W_5^{\a}\)\rangle
+ \text{perm}(2345)
\ee
\[
= K^a + K^b,
\]
where, following the notation of \cite{siegel_lee}, we defined
\be
\label{Kb}
K^a = \sum_{i=1}^4 G^{\t}(z_5,z_2)\langle (\l\g^m\t)(\l\g^nW^1)(\l\g^p \g^{rs}W^5)(W^3\g_{mnp} W^4)F^2_{rs}\rangle
\ee
\be
\label{Ka}
K^b = \sum_{i=1}^4 G^X(z_5,z_i)  k^r_i\langle A^5_r (\l\g^m\t)(\l\g^nW^1)(\l\g^p W^2)
(W^3\g_{mnp} W^4)\rangle
\ee
To go from \eqref{inic} to arrive at \eqref{Kb} and \eqref{Ka} we used the following
substitution
\be\label{trickt}
d_{\d_1}d_{\d_2}d_{\d_3}d_{\d_4}d_{\d_5} \rightarrow 
(\l\g^m)_{[\d_1}(\l\g^n)_{\d_2}(\l\g^p)_{\d_3}(\g_{mnp})_{\d_4\d_5]},
\ee
which can be justified by a ghost number conservation argument
using the integrations over the measures $[dw],[d\wb]$ and
$[ds]$, as explained in section \ref{multiloop_NMPS}.

To evaluate the bosonic components of \eqref{Ka} one notices that the superfield
$A^5_r$ doesn't contribute any $\t$'s, so it can be moved out of the pure spinor
brackets $\langle \; \rangle$. By making use of identity \eqref{pil} we get
\[
K^b = 8 \sum_{i=1}^4 G^X(z_5,z_i)  (k^i\cdot e^5)\langle  (\l A^1)(\l\g^mW^2)(\l\g^n W^3){\cal F}_{mn}^4
\rangle,
\]
which is equivalent to equation (5.5.2) of \cite{siegel_lee}.

In the kinematic factor $K^b$ we will restrict our attention to the
term proportional to $G^{\t}(z_5,z_2)$, which is antisymmetric in its 
arguments. Consequently,
\[
K^a = G^{\t}(z_5,z_2)\Big[
\langle (\l\g^m\t)(\l\g^nW^1)(\l\g^p \g^{rs}W^5)(W^3\g_{mnp} W^4)F^2_{rs}\rangle
- (2\leftrightarrow 5)\Big].
\]
After checking that the following is true for the bosonic components
\[
\langle k^1_t e^1_u (\l\g^m W^3)(\l\g^n W^4)(\l\g^p\g^{rs}W^5)(\t\g^t\g_{mnp}\g^u\t)
{\cal F}_{rs}^2\rangle - (2\leftrightarrow 5) = 0,
\]
one can repeat the steps in the proof of section \ref{crazy_four_sub}
to arrive at the following 
\be
\label{fimKb}
K^a = - G^{\t}(z_5,z_2)\Big[
\langle (D\g_{mnp}A^1)(\l\g^m W^3)(\l\g^n W^4)(\l\g^p \g^{rs}W^5)F^2_{rs}\rangle
- (2\leftrightarrow 5)\Big].
\ee
The computation of the bosonic components of \eqref{fimKb} is straightforward,
\[
K^a = - \frac{1}{256}F^1_{m_1n_1}{\ldots}F^5_{m_5n_5}\Big[
\]
\[
\langle
(\l\g^m \t)(\l\g^n\g^{m_1n_1}\t)(\l\g^p \g^{m_2n_2}\g^{m_5n_5}\t)(\t\g^{m_3n_3}\g_{mnp}\g^{m_4n_4}\t)
\rangle - (2\leftrightarrow 5)\Big]
\]
\begin{align}
= \frac{1}{360}\Big[    & + 2F^1_{mn}F^2_{mp}F^3_{nq}F^4_{qr}F^5_{pr}
       + 2F^1_{mn}F^2_{mp}F^3_{qr}F^4_{nq}F^5_{pr}
       - F^1_{mn}F^2_{mp}F^3_{qr}F^4_{qr}F^5_{np} \nonumber \\
     &  + F^1_{mn}F^2_{pq}F^3_{mn}F^4_{pr}F^5_{qr}
       - 2F^1_{mn}F^2_{pq}F^3_{mp}F^4_{nr}F^5_{qr}
       + 2F^1_{mn}F^2_{pq}F^3_{mr}F^4_{np}F^5_{qr} \nonumber \\
     &  + 2F^1_{mn}F^2_{pq}F^3_{mr}F^4_{pr}F^5_{nq}
       + F^1_{mn}F^2_{pq}F^3_{pr}F^4_{mn}F^5_{qr}
       + 2F^1_{mn}F^2_{pq}F^3_{pr}F^4_{mr}F^5_{nq}
       \Big].
\end{align}
This result is equivalent to equation (5.5.1) of \cite{siegel_lee}.
So the modified  prescription of \eqref{modif} is able to quickly
reproduce the same \emph{bosonic} results while being simpler to
obtain than using the ``unmodified'' non-minimal pure spinor prescription
of \ref{multiloop_NMPS}, as we shall see in the next subsection.
However much work remains to be done in order to understand the procedure
described in the previous paragraphs and to relate it to following 
(incomplete as of September 2008)\footnote{It was completed
in February 2009 and agreement was found with RNS/GS/LS \cite{5pt}.}
computation using the ``unmodified''
NMPS formalism.

%*************************************************
\subsection{Unmodified NMPS computation}
%*************************************************
\label{crazy_5}

Using the unmodified NMPS prescription of \eqref{unmod}, with the b-ghost being
\be\label{bghostf}
b = s^{\a}\p\lb_{\a} + {2\Pi^m(\lb\g_m d)-N_{mn}(\lb\g^{mn}\p\t)-J(\lb\p\t)
-(\lb\p^2\t) \over 4(\lb\l)}
\ee
\[
+{ (\lb\g^{mnp}r)(d\g_{mnp}d+24N_{mn}\Pi_p) \over 192(\lb\l)^2}
-{(r\g^{mnp}r)(\lb\g_m d)N_{np} \over 16(\lb\l)^3}
+{(r\g^{mnp}r)(\lb\g_{pqr}r)N_{mn}N^{qr} \over 128(\lb\l)^4},
\]
the saturation of the 16 $d_{\a}$ and 11 $s^{\a}$ zero modes implies these
four different expressions for the kinematic 
factor\footnote{The hat notation means that the variable acts via its OPE.}.
\be\label{pone}
I_1= {1\over 2}\langle {\Pi^m \over (\l\lb)}(\lb \g_m d) (\l A^1)(dW^2)(dW^3)(dW^4)(dW^5)
\rangle
\ee
\be\label{ptwo}
I_2=-{1\over 16}\langle {(r\g_{mnp}r) \over(\l\lb)^3}(\lb\g^m d)N^{np}
(\l A^1)(dW^2)(dW^3)(dW^4)(dW^5)
\rangle
\ee
\be\label{pthree}
I_3= {1\over 96}\langle {(\lb\g_{mnp}r) \over(\l\lb)^2}(d\g^{mnp} {\hat d})
(\l A^1)(dW^2)(dW^3)(dW^4)(dW^5)
\rangle
\ee
\[
I_4 = \langle
{(\lb\g_{mnp}r) \over 192(\l\lb)^2}(d\g^{mnp}d)(\l A^1)
(dW^2)(dW^3)(dW^4)(A^5_q\Pi^q
+ ({\hat d}W^5) + {1\over 2}N\cdot F^5)\rangle
\]
\be\label{pfour}
+ \text{perm(2345)}
\rangle
\ee
We will restrict out attention to $I_4$, which is the only case where the
b-ghost doesn't act through OPE's.

The term proportional to $G^{\t}(z_5,z_2)$ in \eqref{pfour} is
\[
G^{\t}(z_5,z_2)\langle (\lb\g_{mnp}r)(d\g^{mnp}d)(\l A^1)
(d\g^{rs}W^5)(dW^3)(dW^4){\cal F}^2_{rs} \rangle - (2\leftrightarrow 5),
\]
which becomes
\[
K^a = G^{\t}(z_5,z_2)\langle (\lb\g_{mnp}D)\Big[ (\l A^1)
(\l\g^m\g^{rs} W^5)(\l\g^n W^3)(\l\g^p W^4){\cal F}^2_{rs}\Big]\rangle
- (2\leftrightarrow 5) + (3\leftrightarrow 4),
\]
where we summed over $(3\leftrightarrow 4)$ for reasons which will soon become
clear.
After a long and tedious computation we get
\[
K^a/G^{\t}(z_5,z_2) = \langle(\lb\g_{mnp}D)\big[ (\l A^1){\cal F}^2_{rs}\big]
(\l\g^m\g^{rs} W^5)(\l\g^n W^3)(\l\g^p W^4)\rangle
\]
\[
-48(\l\lb)\langle(\l A^1)(\l\g^{[m}W^3)(\l\g^{n]}W^4){\cal F}^2_{mu}F^5_{nu}\rangle
\]
\[
+ 12 (\l\lb)\langle(\l A^1)(\l\g^m \g^{rs}W^5)(\l\g^n W^4){\cal F}^2_{rs}{\cal F}^3_{mn}\rangle
\]
\[
+ 12 (\l\lb)\langle(\l A^1)(\l\g^m \g^{rs}W^5)(\l\g^n W^3){\cal F}^2_{rs}{\cal F}^4_{mn}\rangle
\]
\[
+ 4\langle(\l\g^{rs}\lb)(\l A^1)(\l\g^m W^3)(\l\g^n W^4){\cal F}^2_{rs}{\cal F}^5_{mn}\rangle
\]
\[
+ 4\langle(\l\g^{rs}\lb)(\l A^1)(\l\g^m W^3)(\l\g^n W^4){\cal F}^5_{rs}{\cal F}^2_{mn}\rangle
\]
\[
- 16\langle(\l\g^{tu}\lb)(\l A^1)(\l\g^{[m} W^3)(\l\g^{n]} W^4){\cal F}^2_{mt}{\cal F}^5_{nu}\rangle
\]
\be
\label{grande}
- (2\leftrightarrow 5) + (3\leftrightarrow 4).
\ee
However the last three lines of \eqref{grande} vanish after antisymmetrization over $[25]$.
The second line can be rewritten more conveniently as
\[
-48(\l\lb)\langle(\l A^1)(\l\g^{[m}W^3)(\l\g^{n]}W^4){\cal F}^2_{mu}F^5_{nu}\rangle =
\frac{1}{2}(\l\lb)
\langle (\l\g^{mnpqr}\l)(\l A^1)(W^3\g_{pqr} W^4){\cal F}^2_{mu}{\cal F}^5_{nu}
\rangle.
\]
Therefore \eqref{grande} becomes, after explicitly summing over $(3\leftrightarrow 4)$,
\[
K^a/G^{\t}(z_5,z_2) = 
\]
\[
+ \langle(\lb\g_{mnp}D)\big[ (\l A^1){\cal F}^2_{rs}\big]
(\l\g^m\g^{rs} W^5)(\l\g^n W^3)(\l\g^p W^4)\rangle
\]
\[
+ \langle(\lb\g_{mnp}D)\big[ (\l A^1){\cal F}^2_{rs}\big]
(\l\g^m\g^{rs} W^5)(\l\g^n W^4)(\l\g^p W^3)\rangle
\]
\[
+ (\l\lb)
\langle (\l\g^{mnpqr}\l)(\l A^1)(W^3\g_{pqr} W^4){\cal F}^2_{mu}{\cal F}^5_{nu}
\rangle
\]
\[
+ 24 (\l\lb)\langle(\l A^1)(\l\g^m \g^{rs}W^5)(\l\g^n W^4){\cal F}^2_{rs}{\cal F}^3_{mn}\rangle
\]
\[
+ 24 (\l\lb)\langle(\l A^1)(\l\g^m \g^{rs}W^5)(\l\g^n W^3){\cal F}^2_{rs}{\cal F}^4_{mn}\rangle
\]
\[
- (2\leftrightarrow 5).
\]
The first two lines can be rewritten using $\g^n_{\a(\b}(\g_n)_{\g\d)}=0$, for example
\be
\label{obter}
(\lb\g_{mnp}D)\big[ (\l A^1){\cal F}^2_{rs}\big](\l\g^m\g^{rs} W^5)(\l\g^n W^3)(\l\g^p W^4) =
\ee
\[
= -D_{\s}\big[(\l A^1){\cal F}^2_{rs}\big](\lb\g^m\g^n\g^p)^{\s}(\l\g^n W^3)
(\l\g^m \g^{rs}W^5)(\l\g^p W^4) 
\]
\[
= (\l\g_n\g^p D)\big[(\l A^1){\cal F}^2_{rs}\big](\lb\g^m\g^n W^3)(\l\g^m\g^{rs} W^5)(\l\g^p W^4)
\]
\[
+ D_{\s}\big[(\l A^1){\cal F}^2_{rs}\big](\lb\g_m\g^n\l)(W^3\g^n\g^p)^{\s}
(\l\g^m\g^{rs}W^5)(\l\g_p W^4)
\]
\be
\label{dif1}
= 2 (\l D)\big[(\l A^1){\cal F}^2_{rs}\big](\lb\g^m\g^n W^3)(\l\g^m\g^{rs} W^5)(\l\g^n W^4)
\ee
\be
\label{dif2}
+ 2 (\l\lb)(W^3\g_m\g_p D)\big[(\l A^1){\cal F}^2_{rs}\big](\l\g^m\g^{rs}W^5)(\l\g^p W^4).
\ee
Using the equation of motion $(\l D)(\l A)=0$ and a few gamma matrix identities we
obtain for \eqref{dif1} (and its permutation over $(3\leftrightarrow 4)$)
\be
\label{corri}
2 \langle(\l D)\big[(\l A^1){\cal F}^2_{rs}\big](\lb\g^m\g^n W^3)
(\l\g^m\g^{rs} W^5)(\l\g^n W^4)\rangle + (3\leftrightarrow 4) =
\ee
\[
= - 8 k^2_r\langle(\l A^1)(\l\g^m W^2)(\l\g^r W^5)\big[
(\lb\g_{mn} W^3)(\l\g^n W^4) + (\lb\g_{mn} W^4)(\l\g^n W^3)
\big] \rangle
\]
\[
= - 4 k^2_r\langle(\l\g_{ab}\lb)
(\l A^1)(\l\g_m W^2)(\l\g^r W^5)(W^3 \g^{mab} W^4)\rangle,
\]
where to arrive at the last line we used
\[
W_3^\a W_4^\b = \frac{1}{16}(W^3 \g^m W^4)\g_m^{\a\b} +
\frac{1}{96}(W^3 \g^{mnp} W^4)\g_{mnp}^{\a\b} +
\frac{1}{3840}(W^3 \g^{mnpqr} W^4)\g_{mnpqr}^{\a\b}.
\]
To rewrite \eqref{dif2} in a more convenient way we use $D_{\a}(\l A^1) = -(\l D)A^1
+ (\l\g^q)_{\a}A_q$ and \eqref{eqW},
\be
\label{af}
2 (\l\lb)(W^3\g_m\g_p D)\big[(\l A^1){\cal F}^2_{rs}\big](\l\g^m\g^{rs}W^5)(\l\g^p W^4) =
\ee
\[
= - 2 (\l\lb)(W^3\g_m\g_p)^{\a}(Q A^1_{\a}){\cal F}^2_{rs}(\l\g^m\g^{rs}W^5)(\l\g^p W^4)
\]
\[
+ 2(\l\lb)A^1_q(W^3\g^m\g^p\g^q\l){\cal F}^2_{rs}(\l\g^m\g^{rs}W^5)(\l\g^p W^4)
\]
\[
+ 4 k^2_r (\l\lb)(\l A^1)(\l\g^m\g^{rs}W^5)(\l\g^p W^4)
(W^3\g^m\g^p \g^s W^2).
\]
The second line is zero due to the pure spinor condition. Integrating the BRST-charge by
parts \eqref{af} becomes
\be
\label{ze1}
 = -\frac{1}{2}(\l\lb)\langle (\l\g^{tu}\g^m\g^p A^1)F^3_{tu}F^2_{rs}(\l\g^m\g^{rs}W^5)(\l\g^p W^4)\rangle
\ee
\be
\label{ze2}
+ \frac{1}{2}(\l\lb)\langle (W^3\g_m\g_p A^1)(\l\g^{mrstu}\l)(\l\g^p W^4){\cal F}^2_{rs}{\cal F}^5_{tu}
\rangle
\ee
\[
+ 8 k^2_r (\l\lb)\langle(W^3\g_m\g_p A^1)(\l\g^m W^2)(\l\g^r W^5)(\l\g^p W^4)\rangle
\]
\[
+ 4 k^2_r(\l\lb)\langle (\l A^1)(\l\g^m\g^{rs}W^5)(\l\g^p W^4)
(W^3\g^m\g^p \g^s W^2)\rangle.
\]
The first line can be rewritten as
\[
-\frac{1}{2}(\l\lb)\langle (\l\g^{tu}\g^m\g^p A^1)F^3_{tu}F^2_{rs}
(\l\g^m\g^{rs}W^5)(\l\g^p W^4)\rangle =
\]
\[
=  4(\l\lb)(\l A^1)(\l\g^m\g^{rs}W^5)(\l\g^n W^4){\cal F}^2_{rs}{\cal F}^3_{mn},
\]
and \eqref{ze2}
vanishes after antisymmetrization in $[25]$. Therefore
\be
\label{corr1}
2 (\l\lb)(W^3\g_m\g_p D)\big[(\l A^1){\cal F}^2_{rs}\big](\l\g^m\g^{rs}W^5)(\l\g^p W^4) 
+ (3\leftrightarrow 4)=
\ee
\[
+4(\l\lb)(\l A^1)(\l\g^m\g^{rs}W^5)(\l\g^n W^4){\cal F}^2_{rs}{\cal F}^3_{mn}
\]
\[
+4(\l\lb)(\l A^1)(\l\g^m\g^{rs}W^5)(\l\g^n W^3){\cal F}^2_{rs}{\cal F}^4_{mn}
\]
\[
- 8 k^2_r (\l\lb)\langle (\l\g^r W^5)(\l\g^m W^2)(\l\g^n W^4)
(W^3\g_{mn} A^1)
\rangle
\]
\[
- 8 k^2_r (\l\lb)\langle (\l\g^r W^5)(\l\g^m W^2)(\l\g^n W^3)
(W^4\g_{mn} A^1)
\rangle
\]
\[
+ 4 k^2_r(\l\lb)\langle (\l A^1)(\l\g^m\g^{rs}W^5)(\l\g^p W^4)
(W^3\g^{mp} \g^s W^2)\rangle.
\]
\[
+ 4 k^2_r(\l\lb)\langle (\l A^1)(\l\g^m\g^{rs}W^5)(\l\g^p W^3)
(W^4\g^{mp} \g^s W^2)\rangle.
\]
So that finally from \eqref{obter}, \eqref{corri} and \eqref{corr1} we get
\[
\langle(\lb\g_{mnp}D)\Big[ (\l A^1){\cal F}^2_{rs}(\l\g^m\g^{rs} W^5)(\l\g^n W^3)
(\l\g^p W^4)\Big]\rangle  - (2\leftrightarrow 5) + (3\leftrightarrow 4) =
\]
\[
+ (\l\lb)
\langle (\l\g^{mnpqr}\l)(\l A^1)(W^3\g_{pqr} W^4){\cal F}^2_{mu}{\cal F}^5_{nu}
\rangle
\]
\[
+ 28 (\l\lb)\langle(\l A^1)(\l\g^m \g^{rs}W^5)(\l\g^n W^4){\cal F}^2_{rs}{\cal F}^3_{mn}\rangle
\]
\[
+ 28 (\l\lb)\langle(\l A^1)(\l\g^m \g^{rs}W^5)(\l\g^n W^3){\cal F}^2_{rs}{\cal F}^4_{mn}\rangle
\]
\[
- 8 k^2_r (\l\lb)\langle(\l\g^r W^5)(\l\g^m W^2)(\l\g^n W^4)(W^3\g_{mn} A^1)\rangle
\]
\[
- 8 k^2_r (\l\lb)\langle(\l\g^r W^5)(\l\g^m W^2)(\l\g^n W^3)(W^4\g_{mn} A^1)\rangle
\]
\[
+ 4 k^2_r(\l\lb)\langle (\l A^1)(\l\g^m\g^{rs}W^5)(\l\g^n W^4)
(W^3\g_{mn} \g_s W^2)\rangle
\]
\[
+ 4 k^2_r(\l\lb)\langle (\l A^1)(\l\g^m\g^{rs}W^5)(\l\g^n W^3)
(W^4\g_{mn} \g_s W^2)\rangle
\]
\be
\label{cpt}
- 4 k^2_r\langle(\l\g_{ab}\lb)
(\l A^1)(\l\g_m W^2)(\l\g^r W^5)(W^3 \g^{abm} W^4)\rangle  - (2\leftrightarrow 5),
\ee
where we used  $(\l\g^m)_{\a}(\l\g_m)_{\b} = 0$ to convert $\g_m\g_p$ into
$\g_{mp}$ in the fourth and fifth lines.

Using the tensor \eqref{T41} we obtain
\[
\langle (\l\g_{ab}\lb)
(\l A^1)(\l\g_m W^2)(\l\g^r W^5)(W^3 \g^{abm} W^4)\rangle =
\]
\[
+\frac{3}{2}\langle(\l\g^{ab} A^1)(\l\g^m W^2)(\l\g^r W^5)(W^3\g_{abm}W^4)\rangle
\]
\[
+\frac{3}{2}\langle (\l A^1)(\l\g^{abm}W^2)(\l\g^r W^5)(W^3\g_{abm}W^4)\rangle
\]
\[
+\frac{3}{2}\langle (\l A^1)(\l\g^m W^2)(\l\g^{ab}\g^r W^5)(W^3\g_{abm}W^4)\rangle,
\]
and therefore the last line in \eqref{cpt} can be evaluated using the
standard methods.

Using the identity $\g^n_{\a(\b}(\g_n)_{\g\d)}=0$ one can show
that \eqref{fimKb} -- which gives us the right answer for the kinematic
factor -- can be rewritten like
\be
\label{propto}
\langle (D\g_{mnp}A^1)(\l\g^m W^3)(\l\g^n W^4)(\l\g^p \g^{rs}W^5)F^2_{rs}\rangle 
- (2\leftrightarrow 5) =
\ee
\[
= -8\langle (\l A^1)(\l\g^m\g^{rs} W^5)(\l\g^n W^4){\cal F}^2_{rs}{\cal F}^3_{mn}
\]
\[
+ 16k^2_r\langle (\l\g^r W^5)(\l\g^m W^2)(\l\g^n W^4)(A^1\g_{mn} W^3)
- (2\leftrightarrow 5),
\]
therefore we conclude that if 
\[
+ (\l\lb)
\langle (\l\g^{mnpqr}\l)(\l A^1)(W^3\g_{pqr} W^4){\cal F}^2_{mu}{\cal F}^5_{nu}
\rangle
\]
\[
+ 24 (\l\lb)\langle(\l A^1)(\l\g^m \g^{rs}W^5)(\l\g^n W^4){\cal F}^2_{rs}{\cal F}^3_{mn}\rangle
\]
\[
+ 24 (\l\lb)\langle(\l A^1)(\l\g^m \g^{rs}W^5)(\l\g^n W^3){\cal F}^2_{rs}{\cal F}^4_{mn}\rangle
\]
\[
+ 4 k^2_r(\l\lb)\langle (\l A^1)(\l\g^m\g^{rs}W^5)(\l\g^n W^4)
(W^3\g_{mn} \g_s W^2)\rangle
\]
\[
+ 4 k^2_r(\l\lb)\langle (\l A^1)(\l\g^m\g^{rs}W^5)(\l\g^n W^3)
(W^4\g_{mn} \g_s W^2)\rangle
\]
\be
\label{cpt2}
- 4 k^2_r\langle(\l\g_{ab}\lb)
(\l A^1)(\l\g_m W^2)(\l\g^r W^5)(W^3 \g^{abm} W^4)\rangle  - (2\leftrightarrow 5),
\ee
is proportional to \eqref{propto} then the NMPS kinematic factor computation
described above is equivalent\footnote{Note added: The story is a bit more
complex. See the complete result in \cite{5pt}.}
to the RNS result listed in \cite{siegel_lee}.

%***************************************
\chapter{Conclusions}
\label{conclu}
%***************************************

In this thesis we have shown how the pure spinor formalism can be used
to obtain in a manifestly SO(1,9)-covariant and supersymmetric way various
scattering amplitudes for massless particles. We emphasized the 
simple pure spinor superspace
expressions for their kinematic factors and how they encode the
complete results for all possible external state combination related
by supersymmetry, and computed them explicitly in components. The
results which were known from the computations of the RNS and
GS formalisms were shown to be correctly reproduced. But also new results
were obtained, namely the explicit computation at two-loops 
involving external fermionic states \cite{stahn}\cite{mafra_tree}.

What was also accomplished is a simple proof which 
explicitly relates
the massless four-point kinematic factors at tree-level
with those at one- and two-loops. The proof can be summarized 
as follows,
\be
\label{art}
K_{\rm 0} = -{1\over 3} K_{\rm 1}, \quad
K_{\rm 2} = -32 K_0\, {\cal Y}(s,t,u),
\ee
where $K_0$, $K_1$ and $K_2$ denote the \emph{complete}
supersymmetric massless four-point kinematic factors at
tree-level, one-loop and two-loops, respectively. The
function ${\cal Y}(s,t,u)$ is called quadri-holomorphic 1-form and is
quadratic in momenta. It is given by
\[
{\cal Y}(s,t,u) = \left[(u-t)\Delta(1,2)\Delta(3,4) 
+(s-t)\Delta(1,3)\Delta(2,4) + (s-u)\Delta(1,4)\Delta(2,3)
\right].
\]
But there is more to the scattering amplitude than just
its kinematic factor. The amplitudes also depend
on the moduli space of the Riemann surface. Considering
the amplitudes for closed strings they read
\[
{\cal A}_0 \propto {\rm e}^{-2\phi} K_0{\bar K}_0
{\Gamma(-{\displaystyle \a' t\over 4})\Gamma(-{\displaystyle \a' u\over 4})
\Gamma(-{\displaystyle \a' s\over 4}) 
\over \Gamma(1+{\displaystyle \a' t\over 4})\Gamma(1+{\displaystyle \a' u\over 4})
\Gamma(1+{\displaystyle \a' s\over 4})},
\]
\[
{\cal A}_1 \propto K_1{\bar K}_1 \int_{{\cal M}_1}
{|d\tau|^2 \over ({\rm Im}\tau)^2} F_1(\tau),
\]
\be
\label{art2}
{\cal A}_2 \propto {\rm e}^{2\phi} {\tilde K}_2{\bar {\tilde K}}_2 \int_{{\cal M}_2}
{|d^3\Omega|^2 \over ({\rm det}\,{\rm Im}\Omega)^5} F_2(\Omega, {\cal Y}),
\ee
where $F_1(\tau)$ and $F_2(\Omega, {\cal Y})$ are modular invariant functions given
by
\[
F_1(\tau)={1\over ({\rm Im}\tau)^3}\int d^2z_2\int d^2z_3\int d^2z_4
\prod_{i<j}G_1(z_i,z_j)^{k_i\cdot k_j}.
\]
\[
F_2(\Omega, {\cal Y}) = \int |{\cal Y}|^2 \prod_{i<j}G_2(z_i,z_j)^{k_i\cdot k_j}
\]
and for convenience we defined $K_2 = {\tilde K}_2 {\cal Y}$.
It is not hard to say that, up to overall coefficients, 
equations \eqref{art} -- \eqref{art2} encode the state-of-the-art
of what is known about massless four-point scattering amplitudes in superstring 
theory after more than two decades of heavy development.

Much more work remains to be done. For example, 
the coefficients in the above amplitudes can be derived by
factorization \cite{gutperle}\cite{xiao},
but it is still unclear how to obtain them directly
from the pure spinor amplitude prescription\footnote{The analogous
task has not yet been done with the RNS at two-loops.}. Another project worth
pursuing is the computation of higher-point scattering amplitudes, because
not so much is known about these amplitudes involving fermionic fields \cite{vanhove}\cite{vanhove2}.
At tree-level, it would be interesting to find the pure spinor superspace expression which
reproduces the five-point results of \cite{medina}, therefore also obtaining the fermionic amplitudes
along the way.
One could expect some previously unknown hidden simplicities will become apparent
using pure spinor superspace. For example there may be a possible generalization of the
identities of \cite{mafra_tree} to the case of five-point amplitudes (as one could
expect from the conjecture in \cite{richards}).

Therefore there remains some avenues of research exploring
the pure spinor formalism's aptitude to compute superstring
scattering amplitudes.

\appendix

%********************************************************
\chapter{Evaluating Pure Spinor Superspace Expressions}
\label{chap_t}
%********************************************************

In the previous chapter we have encountered many pure spinor superspace
expressions of the form
\be
\label{comp}
\langle \l^{\a}\l^{\b}\l^{\g}\t^{\d_1}\t^{\d_2}\t^{\d_3}
\t^{\d_4}\t^{\d_5}f_{\a\b\g}(\t)\rangle
\ee
where $f_{\a\b\g}(\t)$ was composed by some combination of 
super-Yang-Mills superfields and the angle brackets $\langle\quad\rangle$
is defined in such a way that the only non-vanishing component is 
proportional to 
\be
\label{rule}
\langle (\l\g^m\t)(\l\g^n\t)(\l\g^p\t)(\t\g_{mnp}\t)\rangle = 1.
\ee
We will now proceed to show how they can be explicitly computed, obtaining
as a result an expression which depends only on polarizations $e_m$, $\xi^{\a}$ 
and momenta $k^m$. That this is possible can be seen by checking that
there is only one scalar built out of three pure spinors $\l^{\a}$ and
five unconstrained $\t$'s
\[
[0,0,0,0,3]\otimes \Big( [0,0,0,3,0]+[1,1,0,1,0]\Big) = 
1X[0,0,0,0,0] +2X[0,0,0,1,1] + {\ldots},
\]
so that an arbitrarily complicated pure spinor correlator
written in terms of three $\l$'s
and five $\t$'s can be written entirely in terms of Kronecker deltas
and  Levi-Civita $\e_{10}$ tensors.

After expanding the superfields appearing in the generic correlator
\eqref{comp} and taking the
terms which contain five $\t$'s one will get
\[
\t^{\d_1}\t^{\d_2}\t^{\d_3}
\t^{\d_4}\t^{\d_5}f_{\a\b\g\d_1\d_2\d_3\d_4\d_5}
\]
where $f_{\a\b\g\d_1\d_2\d_3\d_4\d_5}$ is composed by a string of gamma matrices
with several indices. Each one of those terms can be easily evaluated using the
rule \eqref{rule}.

% For example, if someone wants to compute the one-loop expression obtained
% for the scattering of four massless strings
% \[
% K_1 = \langle (\l A^1)(\l\g^m W^2)(\l\g^n W^3) {\cal F}^4_{mn}\rangle
% \]
% it can be checked using the superfield expansions of 
% appendix \ref{sym_ap} that there will be four distinct ways to obtain five $\t$'s,
% according to the following table
% \begin{center}
% \begin{tabular}{|c|c|c|c|} \hline
% $A^1_{\a}(\t)$ & $W_2^{\a}(\t)$ & $W_3^{\a}(\t)$ & ${\cal F}^4_{mn}(\t)$ \\ \hline
% 1 & 1 & 1 & 2 \\ \hline
% 1 & 1 & 3 & 0 \\ \hline
% 1 & 3 & 1 & 0 \\ \hline
% 3 & 1 & 1 & 0 \\ \hline
% \end{tabular}
% \end{center}
% or explicitly
% \[
% K_1= +{1\over 128}(k^4_me^1_n)F^2_{pq}F^3_{rs}F^4_{tu}\langle(\l\g^{[m}\g^{pq}\t)
% (\l\g^{a]}\g^{rs}\t)(\l\g^n\t)(\t\g_a\g^{tu}\t)\rangle 
% \]
% \[
% +{1\over 384}(k^3_me^1_n)F^2_{pq}F^3_{rs}F^4_{tu}\langle(\l\g^{[t}\g^{pq}\t)
% (\l\g^{u]}\g^{ma}\t)(\l\g^n\t)(\t\g_a\g^{rs}\t)\rangle.
% \]
% \[
% +{1\over 384}(k^2_me^1_n)F^2_{pq}F^3_{rs}F^4_{tu}\langle(\l\g^{[t}\g^{ma}\t)
% (\l\g^{u]}\g^{rs}\t)(\l\g^n\t)(\t\g_a\g^{pq}\t)\rangle 
% \]
% \be
% \label{kin1}
% +{1\over 512}F^1_{mn}F^2_{pq}F^3_{rs}F^4_{tu}
% \langle
% (\l\g^{[t}\g^{pq}\t)(\l\g^{u]}\g^{rs}\t)(\l\g_a\t)(\t\g^{mna}\t).
% \rangle 
% \ee
Suppose one wants to compute the following pure spinor correlation
\be
\label{now}
\langle (\l\g^m \t)(\l\g^n\g^{rs}\t)(\l\g^p\g^{tu}\t)(\t\g_{fgh}\t)\rangle.
\ee
To use the argument after \eqref{rule} it is better to write \eqref{now}
in a form in which the symmetries over the vector indices are manifest.
In this case we can do it by using the gamma matrix identity
\be
\label{mat1}
\g^m\g^{np} = \g^{mnp} + \eta^{mn}\g^p - \eta^{mp}\g^n.
\ee
to obtain 
\be
\label{tocompu}
\langle (\l\g^m \t)(\l\g^n\g^{rs}\t)(\l\g^p\g^{tu}\t)(\t\g_{fgh}\t)\rangle =
\langle (\l\g^m \t)(\l\g^{nrs} \t)(\l\g^{ptu} \t)(\t\g_{fgh}\t)
\ee
\[
+ 2 \langle(\l\g^{m} \t) \d^{[r}_n(\l\g^{s]} \t)(\l\g^{ptu} \t)(\t\g_{fgh}\t)\rangle
+ 4 \langle(\l\g^{m} \t) \d^{[r}_n(\l\g^{s]} \t)\d^{[t}_p(\l\g^{u]} \t)(\t\g_{fgh}\t)\rangle.
\]
And now we can easily use symmetry arguments to show that
\be
\label{eas}
\langle(\l\g^{m} \t) (\l\g^{s} \t)(\l\g^{u} \t)(\t\g_{fgh}\t)\rangle = \frac{1}{120}\d^{msu}_{fgh}
\ee
where $\d^{msu}_{fgh}$ is the antisymmetrized combination of Kronecker deltas beginning with
$\frac{1}{3!}\d^m_f\d^s_g\d^u_h$. To see this note that the right hand side of \eqref{eas} is the
only tensor which is antisymmetric in $[msu]$ and $[fgh]$ and which is normalized to one (because
$\d^{msu}_{msu} = 120$), therefore respecting the normalization imposed by the rule \eqref{rule}.
By the same token, using symmetry arguments alone one can show that
\be
\label{duo}
\langle(\l\g_{m}\t)(\l\g_{s}\t)(\l\g^{ptu}\t)(\t\g_{fgh}\t)\rangle=
       {1\over 70}\d^{[p}_{[m}\eta_{s][f}\d^{t}_g\d^{u]}_{h]}
\ee
\[
\langle(\l\g_{m}\t)(\l\g^{n rs}\t)(\l\g^{ptu}\t)(\t\g_{fgh}\t)\rangle=
{1\over 8400}\e^{fghmnprstu}+
\]
\be
\label{tres}
+{1\over 140}\Big[ 
	 \d^{[n}_m\d^r_{[f}\eta^{s][p}\d^t_g\d^{u]}_{h]}
	-\d^{[p}_m\d^t_{[f}\eta^{u][n}\d^r_g\d^{s]}_{h]}
\Big]
-{1\over 280}\Big[ 
	 \eta_{m[f}\eta^{v[p} \d^t_g\eta^{u][n}\d^r_{h]}\d^{s]}_v 
	-\eta_{m[f}\eta^{v[n} \d^r_g\eta^{s][p}\d^t_{h]}\d^{u]}_v
\Big].
\ee

In general, using several 
identities like $(\t\g^{abc}\g^{mn}\t)=(\t\g^{r_1r_2r_3}\t)K^{abcmn}_{r_1r_2r_3}$,
where
$$
K^{abcmn}_{r_1r_2r_3} = 
-\eta^{cn}\d^{abm}_{r_1r_2r_3} + 
 \eta^{cm}\d^{abn}_{r_1r_2r_3} + 
 \eta^{bn}\d^{acm}_{r_1r_2r_3} 
-\eta^{bm}\d^{acn}_{r_1r_2r_3} - 
 \eta^{an}\d^{bcm}_{r_1r_2r_3} + 
 \eta^{am}\d^{bcn}_{r_1r_2r_3}
$$
or
\[
          (\l\g^{mnp}\t)(\l\g^{qrs}\t) 
           = -{1\over 96}(\t\g^{tuv}\t)(\l\g^{mnp}\g_{tuv}\g^{qrs}\l)
\]
\be
\label{ftensor}
           \equiv - {1\over 96}(\l\g^{abcde}\l)(\t\g^{tuv}\t)f^{mnpqrs}_{abcdetuv},
\ee
where
\[
f^{mnpqrs}_{abcdetuv} = 18 (
\d^{rs}_{uv}\d^{abcde}_{mnpqt}-\d^{np}_{uv}\d^{abcde}_{qrsmt}) +
54(\d^{ps}_{tv}\d^{abcde}_{mnqru}+\d^{mn}_{qr}\d^{abcde}_{pstuv})
\]
\[
+54(\d^{nv}_{rs}\d^{abcde}_{mpqtu}-\d^{rv}_{np}\d^{abcde}_{qsmtu})+[mnp]+[qrs]+[tuv]
\]
and gamma matrix identities like \eqref{mat1} together with
$$
(\l\g^{abc}\g^{de}\t) = + (\l\g^{abcde}\t)
-2\d^{bc}_{de}(\l\g^{a}\t) 
+ 2\d^{ac}_{de}(\l\g^{b}\t) 
- 2\d^{ab}_{de}(\l\g^{c}\t)
$$
$$
- \d^{c}_{e}(\l\g^{abd}\t) + 
 \d^{c}_{d}(\l\g^{abe}\t) + 
 \d^{b}_{e}(\l\g^{acd}\t) - 
 \d^{b}_{d}(\l\g^{ace}\t) - 
 \d^{a}_{e}(\l\g^{bcd}\t) + 
 \d^{a}_{d}(\l\g^{bce}\t) 
$$
(and many others)
all pure spinor superspace expressions can be written in terms
as a linear combination of 
the basic ones \eqref{eas}, \eqref{duo}, \eqref{tres} and
\be
\label{quattuor}
\langle (\l\g^{mnpqr}\t)(\l\g_{stu}\t)(\l\g^v\t)(\t\g_{fgh}\t)=
{1\over 35}
\eta^{v[m}\d^n_{[s}\d^p_t \eta_{u][f}\d^q_g\d^{r]}_{h]}
-{2\over 35}
\d^{[m}_{[s}\d^n_t\d^p_{u]}\d^q_{[f}\d^{r]}_g\d^v_{h]}
\ee
\[
+{1\over 120}\e^{mnpqr}_{\qquad\;\; abcde}\left(
{1\over 35}
\eta^{v[a}\d^b_{[s}\d^c_t \eta_{u][f}\d^d_g\d^{e]}_{h]}
-{2\over 35}
\d^{[a}_{[s}\d^b_t\d^c_{u]}\d^d_{[f}\d^{e]}_g\d^v_{h]}
\right)
\]
\be
\label{quinque}
\langle (\l\g^{mnpqr}\t)(\l\g_{d}\t)(\l\g_e\t)(\t\g_{fgh}\t)=
 - {1\over 42}\delta^{mnpqr}_{defgh} 
 - {1\over 5040}\e^{mnpqr}_{\qquad\;\; defgh}
\ee
\be
\label{septem}
\vev{(\l\g^{mnp}\t)(\l\g^{qrs}\t)(\l\g_{tuv}\t)(\t\g_{ijk}\t)}=
\ee
$$
-{3\over 175}\Big[- \d^{[i}_a\d^{j}_{[q}\d^{k]}_{r}\d^{[m}_{s]}\d^{n}_{[t}\d^{p]}_{u}\d^a_{v]}
+ \d^{[i}_a\d^{j}_{[t}\d^{k]}_{u}\d^{[m}_{v]}\d^{n}_{[q}\d^{p]}_{r}\d^a_{s]} 
+ \d^{[i}_{[q}\d^{j}_{r}\eta^{k][m}\eta_{s][t}\d^{n}_{u}\d^{p]}_{v]}
$$
$$
+\d^a_{[t}\eta^{b[i}\d^j_u\eta^{k][m}\eta_{v][q}\d^n_r\eta_{s]a}\d^{p]}_b
- \d^{a}_{[q}\eta^{b[i}\d^{j}_{r}\eta^{k][m}\eta_{s][t}\d^{n}_{u}\eta_{v]a}\d^{p]}_b 
- \d^{[i}_{[t}\d^{j}_{u}\eta^{k][m}\eta_{v][q}\d^{n}_{r}\d^{p]}_{s]}
\Big]
$$
$$
+{1\over 33600}\e^{abcde}_{\qquad \, a_1a_2a_3a_4a_5}
f^{mnpqrs}_{abcdefgh}\Big[
	\d^{[t}_{[f}\d^u_g \eta^{v][a_1}\d^{a_2}_{h]}\d^{a_3}_{[i}\d^{a_4}_j\d^{a_5]}_{k]}
      + \d^{[t}_{[i}\d^u_j \eta^{v][a_1}\d^{a_2}_{k]}\d^{a_3}_{[f}\d^{a_4}_g\d^{a_5]}_{h]} 
$$
$$
      - \eta^{z[t}\d^u_{[f}\eta^{v][a_1}\d_g^{a_2}\eta_{h][i}\d^{a_3}_j
        \d^{a_4}_{k]}\d_z^{a_5]}
      - \eta^{z[t}\d^u_{[i}\eta^{v][a_1}\d_j^{a_2}\eta_{k][f}\d^{a_3}_g
        \d^{a_4}_{h]}\d_z^{a_5]} 
\Big].
$$
\be
\label{uind}
\langle(\l\g^{mnpqr}\l)(\l\g^{u}\t)(\t\g_{fgh}\t)
(\t\g_{jkl}\t)\rangle = 
\ee
\[
-{4\over 35}{\Big[}\d^{[m}_{[j} \d^n_k \d^p_{l]}\d^q_{[f}\d^{r]}_g \d^u_{h]}
+\d^{[m}_{[f} \d^n_g \d^p_{h]}\d^q_{[j}\d^{r]}_k \d^u_{l]}
-{1\over 2}\d^{[m}_{[j}\d^n_k \eta_{l][f}\d^p_g\d^q_{h]}\eta^{r]u}
-{1\over 2}\d^{[m}_{[f}\d^n_g \eta_{h][j}\d^p_k\d^q_{l]}\eta^{r]u}
\Big]
\]
\[
-{1\over 1050}\e^{mnpqr}_{\qquad\;\; abcde}{\Big[}
\d^{[a}_{[j} \d^b_k \d^c_{l]}\d^d_{[f}\d^{e]}_g \d^u_{h]}
+\d^{[a}_{[f} \d^b_g \d^c_{h]}\d^d_{[j}\d^{e]}_k \d^u_{l]}
\]
\[
-{1\over 2}\d^{[a}_{[j}\d^b_k \eta_{l][f}\d^c_g\d^d_{h]}\eta^{e]u}
-{1\over 2}\d^{[a}_{[f}\d^b_g \eta_{h][j}\d^c_k\d^d_{l]}\eta^{e]u}
\Big]
\]
\be
\label{tind}
\langle(\l\g^{mnpqr}\l)(\l\g^{stu}\t)(\t\g_{fgh}\t)
(\t\g_{jkl}\t)\rangle=
\ee
\[
-{12\over 35}\Big[
	\d^{[s}_{[f}\d^t_g \eta^{u][m}\d^n_{h]}\d^p_{[j}
	\d^q_k\d^{r]}_{l]}
      + \d^{[s}_{[j}\d^t_k \eta^{u][m}\d^n_{l]}\d^p_{[f}
        \d^q_g\d^{r]}_{h]} 
\]
\[
      - \eta^{v[s}\d^t_{[f}\eta^{u][m}\d_g^n\eta_{h][j}
        \d^p_k\d^q_{l]}\d_v^{r]}
      - \eta^{v[s}\d^t_{[j}\eta^{u][m}\d_k^n\eta_{l][f}
        \d^p_g\d^q_{h]}\d_v^{r]} 
\Big]
\]
\[
-{1\over 350}\e^{mnpqr}_{\qquad \, abcde}\Big[
	\d^{[s}_{[f}\d^t_g \eta^{u][a}\d^b_{h]}\d^c_{[j}\d^d_k\d^{e]}_{l]}
      + \d^{[s}_{[j}\d^t_k \eta^{u][a}\d^b_{l]}\d^c_{[f}\d^d_g\d^{e]}_{h]} 
\]
\[
      - \eta^{v[s}\d^t_{[f}\eta^{u][a}\d_g^b\eta_{h][j}\d^c_k
        \d^d_{l]}\d_v^{e]}
      - \eta^{v[s}\d^t_{[j}\eta^{u][a}\d_k^b\eta_{l][f}\d^c_g
        \d^d_{h]}\d_v^{e]} 
\Big]
\]
and finally, the only ``important'' correlator which was missing in the catalog of \cite{nmps_two}
was obtained by Stahn in \cite{stahn}
$$
 \langle (\lambda\gamma^{mnpqr}\lambda) (\lambda\gamma^{abcde}\theta) 
 (\theta\gamma^{fgh}\theta) (\theta\gamma^{jkl}\theta)\rangle =
 \tfrac{16}{7} \left(\delta^{mnpqr}_{\bar m\bar n\bar p\bar q\bar r} 
 + \tfrac{1}{5!}\varepsilon^{mnpqr}{}_{\bar m\bar n\bar p\bar q\bar r}\right) 
$$ 
\be
\label{cind}
\times \Big[ \delta^{\bar m\bar n\bar p}_{abc} \delta^f_j\delta^{d}_{g}\delta^{\bar q}_{k} ( - \delta^{e}_{h}\delta^{\bar r}_{l} 
+ 2 \delta^{e}_{l}\delta^{\bar r}_{h}) 
+ \delta^{\bar m\bar n}_{ab} \delta^{cd}_{fg}\delta^{\bar p\bar q}_{jk}( \delta^{e}_{h}\delta^{\bar r}_{l} 
- 3  \delta^{e}_{l}\delta^{\bar r}_{h}) \Big]_{[abcde][fgh][jkl](fgh\leftrightarrow jkl)}
\ee
where, for convenience, \eqref{cind} has been taken \emph{ipsis litteris} from \cite{stahn}.

\section{Obtaining epsilon terms}
In the correlations above we obtained the epsilons terms by considering
the duality of the gamma matrices
\be
\label{cincind}
\(\g^{m_1m_2m_3m_4m_5}\)_{\a\b}=+{1\over 5!}
\e^{m_1m_2m_3m_4m_5n_1n_2n_3n_4n_5}\(\g_{n_1n_2n_3n_4n_5}\)_{\a\b},
\ee
$$
\(\g^{m_1m_2m_3m_4m_5m_6}\)^{\,\,\b}_{\a}=+{1\over 4!}
\e^{m_1m_2m_3m_4m_5m_6n_1n_2n_3n_4}\(\g_{n_1n_2n_3n_4}\)^{\,\,\b}_{\a},
$$
$$
\(\g^{m_1m_2m_3m_4m_5m_6m_7}\)_{\a\b}=-{1\over 3!}
\e^{m_1m_2m_3m_4m_5m_6m_7n_1n_2n_3}\(\g_{n_1n_2n_3}\)_{\a\b},
$$
$$
\(\g^{m_1m_2m_3m_4m_5m_6m_7m_8}\)^{\,\,\b}_{\a}=-{1\over 2!}
\e^{m_1m_2m_3m_4m_5m_6m_7m_8n_1n_2}\(\g_{n_1n_2}\)^{\,\,\b}_{\a}.
$$
For example, to obtain the epsilon term of \eqref{tres} we used the identity
\eqref{ftensor} to relate \eqref{tres} with \eqref{uind}, whereas the 
epsilon terms in \eqref{uind} were found by first computing its Kronecker
delta terms and then using \eqref{cincind} to obtain the
epsilon terms.

In fact, due to the pure spinor property of $(\l\g^m\l) = 0$ there are
only three different correlations which need to be taken care of. That is because
one can always use the properties of
\[
\l^{\a}\l^{\b} = \frac{1}{3840}(\l\g^{mnpqr}\l)\g_{mnpqr}^{\a\b}
\]
\[
\t^{\a}\t^{\b} = \frac{1}{96}(\t\g^{mnp}\t)\g_{mnp}^{\a\b}
\]
to write any\footnote{This was explicited in \cite{stahn}.} correlation in terms of \eqref{uind}, \eqref{tind} and
\eqref{cind}. And as the epsilon 
terms of these three fundamental
building blocks are easily found through the use of \eqref{cincind},
all epsilon terms of an arbitrary pure spinor expression are easily\footnote{If
one has a computer to do the tedious calculations, of course \cite{FORM}\cite{tform}\cite{ulf}.}
determined.

%********************************************************
\chapter{$N=1$ Super Yang-Mills Theory in $D=10$}
\label{sym_ap}
%********************************************************

The basic reference for this appendix is \cite{SYM} (see also \cite{siegel_sym}).
One can also look at the 1991 derivation of the SYM on-shell 
constraint $F_{\a\b}=0$ using ten-dimensional pure spinors by Howe \cite{howe_ps1},
whose work also contains the operator $u^\a D_\a$ (the pure spinor $\l^\a$
was denoted by $u^\a$ in \cite{howe_ps1}), which is essentially the BRST operator
\eqref{Q_brst} of the pure spinor formalism\footnote{As a curiosity it is interesting
to mention that in reference \cite{howe_ps1} the author makes reference to the
``pure spinor formalism'' in the context of \cite{howe_ps2}.}.

\begin{def.} \label{defu} The supercovariant derivatives are
\begin{align}
\nabla_m & = \p_m + A_m \\
\nabla_{\alpha} & = D_{\alpha} + A_{\alpha} \\
D_{\a} & = \frac{\p}{\p\t^{\a}} + \frac{1}{2}(\g^m\t)_{\a}\p_m,
\end{align}
and the field-strengths are defined by
\begin{align}
\label{fab}
F_{\a\b} & = \{\nabla_{\a},\nabla_{\b} \} - \g^m_{\a\b}\nabla_m \\
\label{defi}
F_{\a m} & = [\nabla_{\a},\nabla_m] \\
F_{mn} & = [\nabla_m, \nabla_n].
\end{align}
\end{def.}
One can easily check the above field-strengths to be invariant under
the gauge transformations of
\be
\label{gaugetr}
\d A_m = \p_m\Omega, \quad \d A_{\a} = D_{\a}\Omega.
\ee
\begin{lema}
The fermionic supercovariant derivative satisfies
\[
\{D_{\a},D_{\b} \} = \g^m_{\a\b}\p_m.
\]
\end{lema}

\begin{prop.}
$F_{\a\b}=0$ if and only if
\begin{equation}
\label{symeq}
\g_{mnpqr}^{\a\b}\(D_{\a}A_{\b}+D_{\b}A_{\a} + \{A_{\a},A_{\b}\}\) = 0.
\end{equation}
\end{prop.}
By straighforward computation using definition \ref{defu} we get
\[
F_{\a\b} = \{D_{\a},D_{\b}\}+\{D_{\a},A_{\b}\}+\{A_{\a},D_{\b}\} + \{A_{\a},A_{\b}\}
-\g^m_{\a\b}(\p_m+A_m) = 0,
\]
so
\[
D_{\a}A_{\b}+D_{\b}A_{\a} + \{A_{\a},A_{\b}\} = \g^s_{\a\b}A_s.
\]
Multipling both sides by $\g_{mnpqr}^{\a\b}$ and using the 
identity $\tr{(\g_{mnpqr}\g_s)}=0$, equation \eqref{symeq} follows. By
reversing the above steps, the converse can 
also be proved. 
\begin{prop.}
\label{prop2}
If $\g_{mnpqr}^{\a\b}\(D_{\a}A_{\b}+D_{\b}A_{\a} + \{A_{\a},A_{\b}\}\) = 0$ (or
equivalently $\{\nabla_{\a},\nabla_{\b}\} = \g^m_{\a\b}\nabla_m$)
then
\be
\label{equivW}
F_{\a m} \equiv (\g_mW)_{\a}
\ee
\be
\label{eqW}
\nabla_{\a}{W^{\b}} = -\frac{1}{4}(\g^{mn})_{\a}^{\phantom{m}\b}F_{mn},
\ee
\be
\label{eqF}
\nabla_{\a}F_{mn} = \nabla_m(\g_nW)_{\a}-\nabla_n(\g_mW)_{\a}.
\ee
\end{prop.}
The proof follows from the use of Bianchi and gamma matrix identities.
The Bianchi identity 
\[
[\{\nabla_{\a},\nabla_{\b}\},\nabla_{\g}] + [\{\nabla_{\g},\nabla_{\a}\},\nabla_{\b}]
+ [\{\nabla_{\b},\nabla_{\g}\},\nabla_{\a}] = 0,
\]
implies that
\[
\g^m_{\a\b}[\nabla_m,\nabla_{\g}]+\g^m_{\g\a}[\nabla_m,\nabla_{\b}]
+\g^m_{\b\g}[\nabla_m,\nabla_{\a}] = 0.
\]
Consequentely,
\begin{equation}
\label{ver}
\g^m_{\a\b}F_{\g m}+\g^m_{\g\a}F_{\b m}
+\g^m_{\b\g}F_{\a m} = 0.
\end{equation}
The identity $\eta_{mn}\g^m_{\a(\b}\g^n_{\g\delta)}=0$ implies that
\eqref{ver} is trivially satisfied if $F_{\a m}=(\g_m W)_{a}$.
Similarly, from 
\[
[[\nabla_m,\nabla_n],\nabla_{\a}]+[[\nabla_{\a},\nabla_m],\nabla_n]
+[[\nabla_n,\nabla_{\a}],\nabla_m] = 0
\]
we obtain \eqref{eqF}.
From 
\[
[\{\nabla_{\a},\nabla_{\b}\},\nabla_m]+\{[\nabla_m,\nabla_{\a}],\nabla_{\b}\}
-\{[\nabla_{\b},\nabla_m],\nabla_{\a}\} = 0,
\]
we get
\be
\label{po}
(\g_m)_{\b\delta}\nabla_{\a}W^{\delta} + (\g_m)_{\a\delta}\nabla_{\b}W^{\delta} =
\g^n_{\a\b}F_{nm}.
\ee
The identity $\nabla_{\a}W^{\a}=0$ follows if we multiply \eqref{po} by $(g^m)^{\a\b}$.
Multiplication by $(\g^m)^{\b\sigma}$ results in,
\be
\label{prim}
10\nabla_{\a}W^{\sigma} + \g^m_{\a\delta}\g_m^{\b\sigma}\nabla_{\b}W^{\delta} =
-(\g^{mn})_{\a}^{\phantom{m}\sigma}F_{mn}.
\ee
Multiplying \eqref{prim} by $\g^p_{\sigma\kappa}\g^{\a\rho}_p$ and using
the identities $(\g^p\g^{mn}\g_p)^{\rho}_{\phantom{n}\kappa}
=-6(\g^{mn})^{\rho}_{\phantom{n}\kappa}$ and
\[
(\g_m)_{\a\d}(\g^m)^{\b\sigma}(\g^p)_{\sigma\kappa}(\g_p)^{\a\rho}=
-4\g_r^{\b\rho}\g^r_{\d\kappa} + 12\d^{\b}_{\kappa}\d^{\rho}_{\delta}
+ 8 \d^{\b}_{\d}\d^{\rho}_{\kappa},
\]
it follows that
\be
\label{seg}
12\nabla_{\kappa}W^{\rho}+6\g^p_{\sigma\kappa}\g_p^{\a\rho}\nabla_{\a}W^{\sigma}
= 6 (\g^{mn})_{\kappa}^{\phantom{n}\rho}F_{mn}.
\ee
Finally, from \eqref{prim} and \eqref{seg} we get $\nabla_{\a}W^{\b} = 
-\frac{1}{4}(\g^{mn})_{\kappa}^{\phantom{n}\rho}F_{mn}$.

\begin{lema}
The following identity holds true (written with the pure spinor $\l^{a}$
for convenience)
\be
\label{cst}
\l^{\a}D_{\a}A_m = (\l\g^m W) + \p^m (\l A).
\ee
\end{lema}
The proof follows trivially from \eqref{defi} and \eqref{equivW}.
\begin{lema}
The superfields $A_m$, $W^{\a}$ and $F_{mn}$ can be written as
\begin{align}
A_m &=\frac{1}{8}\g_m^{\a\b}(D_{\a}A_{\b}+A_{\a}A_{\b})\\
W^{\a} & = \frac{1}{10}(\g^m)^{\a\b}\(D_{\b}A_m - \p_mA_{\b}+[A_{\b},A_m]\)\\
F_{mn} & = \nabla_m A_n - \nabla_n A_m 
\end{align}
\end{lema}
\begin{prop.}
The constraint equation 
\be
\label{const}
\g_{mnpqr}^{\a\b}\(D_{\a}A_{\b}+D_{\b}A_{\a} + \{A_{\a},A_{\b}\}\) = 0
\ee
is equivalent to the super Yang-Mills equations
\be
\label{eom1}
\g^m_{\a\b}\nabla_mW^{\b} = 0
\ee
\be
\label{eom2}
\nabla_mF^{mn} + \frac{1}{2}\g^n_{\a\b}\{W^{\a},W^{\b}\} = 0.
\ee
\end{prop.}

From Proposition \ref{prop2} we know that the constraint \eqref{const}
implies the equations \eqref{eqW} and \eqref{eqF}. Now we will show that
\eqref{eom1} and \eqref{eom2} follow from those two equations, which 
proves the above proposition.

To prove \eqref{eom1} we act with the derivative $\nabla_{\g}$ over \eqref{eqW}
and symmetrize over the spinor indices $(\g\a)$ to obtain
\[
\(\nabla_{\a}\nabla_{\g} + \nabla_{\g}\nabla_{\a}\)W^{\b} = 
-\frac{1}{4}(\g^{mn})_{\g}^{\phantom{\g}\b}\nabla_{\a}F_{mn}
-\frac{1}{4}(\g^{mn})_{\a}^{\phantom{\g}\b}\nabla_{\g}F_{mn}
\]
Using \eqref{fab} in the left hand side and \eqref{eqF} on the right we
get
\[
\g^p_{\a\b}\nabla_p W^{\b} = 
-\frac{1}{2}(\g^{mn})_{\g}^{\phantom{\g}\b}(\g_m\nabla_n W)_{\a}
-\frac{1}{2}(\g^{mn})_{\a}^{\phantom{\g}\b}(\g_m\nabla_n W)_{\g},
\]
from which we obtain upon multiplication by $\d_{\b}^{\g}$ on both
sides and using $\g^{mn}_m = 7\g^n$ that
\be
\label{pas1}
\g^p_{\a\b}\nabla_p W^{\b} = 0,
\ee
which proves \eqref{eom1}. To obtain \eqref{eom2} we multiply
\eqref{pas1} by $\g^{\a\d}_n\nabla_{\d}$ to get
\[
\g^m_{\a\b}\g^{\a\d}_n\nabla_{\d}\nabla_m W^{\b} = 0.
\]
Using 
$\nabla_{\d}\nabla_m = [\nabla_{\d},\nabla_m] + \nabla_m\nabla_{\d}$, 
$[\nabla_{\d},\nabla_m] = (\g_m W)_{\d}$ and the equation of motion
\eqref{eqW} we arrive at
\[
\frac{1}{4}\text{tr}(\g_n\g^{pq}\g_m)\nabla_m F_{pq} = (\g^m\g_n\g_m)_{\b\kappa} W^{\kappa}W^{\b}
\]
which implies
\[
\nabla^m F_{mn} = (\g_n)_{\a\b}W^{\b}W^{\a} 
= -\frac{1}{2}(\g_n)_{\a\b}\{W^{\a},W^{\b}\}.
\]

%************************************************************
\section{The $\t$-expansion of Super Yang-Mills superfields}
%************************************************************

In this thesis we use the following ${\cal N}=1$ super-Yang-Mills $\t$ 
expansions \cite{thetaSYM}\cite{theta2}\cite{theta3}
$$
A_{\a}(x,\t)={1\over 2}a_m(\g^m\t)_\a -{1\over 3}(\xi\g_m\t)(\g^m\t)_\a
-{1\over 32}F_{mn}(\g_p\t)_\a (\t\g^{mnp}\t) 
$$
$$
+ {1\over 60}(\g_m\t)_{\a}(\t\g^{mnp}\t)(\p_n\xi\g_p\t) 
+ {1\over 1152}(\g^m\t)_\a(\t\g^{mrs}\t)(\t\g^{spq}\t)\p_r F_{pq} 
+ \ldots
$$
$$
A_{m}(x,\t) = a_m - (\xi\g_m\t) - {1\over 8}(\t\g_m\g^{pq}\t)F_{pq}
         + {1\over 12}(\t\g_m\g^{pq}\t)(\p_p\xi\g_q\t) +
$$
$$
+ {1\over 192}(\t\g^{mrs}\t)(\t\g^{spq}\t)\p_r F_{pq} + {\ldots} 
$$
$$
W^{\a}(x,\t) = \xi^{\a} - {1\over 4}(\g^{mn}\t)^{\a} F_{mn}
           + {1\over 4}(\g^{mn}\t)^{\a}(\p_m\xi\g_n\t)
	   + {1\over 48}(\g^{mn}\t)^{\a}(\t\g_n\g^{pq}\t)\p_m F_{pq} 	  
$$
$$
- {1\over 96}(\g^{mn}\t)^\a(\t\g^{npq}\t)\p_m\p_p(\xi \g_q \t)
- {1\over 1536}(\g^{mn}\t)^\a(\t\g^{nrs}\t)(\t\g^{spq}\t)\p_m\p_r F_{pq} + {\ldots} 
$$
$$
{\cal F}_{mn}(x,\t) = F_{mn} - 2(\p_{[m}\xi\g_{n]}\t) + {1\over
4}(\t\g_{[m}\g^{pq}\t)\p_{n]}F_{pq} -{1\over 6}(\t\g_{[m}\g^{pq}\t)\p_{n]}\p_p(\xi\g_q\t)
$$
\be
\label{sym_exp}
-{1\over 96}(\t\g_{[m}\g^{rs}\t)(\t\g^{spq}\t)\p_{n]}\p_r F_{pq} + {\ldots} 
\ee
Here $\xi^{\a}(x) =
\chi^{\a}{\rm e}^{ik\cdot x}$ and $a_m(x) = e_m {\rm e}^{ik\cdot x}$ describe the gluino and
gluon respectively, while $F_{mn} = 2\p_{[m} a_{n]}$ is the gluon field-strength.

%********************************************************
\chapter{The $t_8$ tensor}
\label{t8_ap}
%********************************************************

The famous $t_8$ tensor is defined by \cite{schwarz_phys_rep}\cite{gswII}
\be
\label{t_tensor}
t_8^{m_1n_1m_2n_2m_3n_3m_4n_4} =
-{1\over 2}\Big[ \left(\d^{m_1m_2}\d^{n_1n_2} - \d^{m_1n_2}\d^{n_1m_2}\right)
\left(\d^{m_3m_4}\d^{n_3n_4} - \d^{m_3n_4}\d^{n_3m_4}\right) 
\ee
\[
+ \left(\d^{m_2m_3}\d^{n_2n_3} - \d^{m_2n_3}\d^{n_2m_3}\right)
\left(\d^{m_4m_1}\d^{n_4n_1} - \d^{m_4n_1}\d^{n_4m_1}\right) 
\]
\[
+ \left(\d^{m_1m_3}\d^{n_1n_3} - \d^{m_1n_3}\d^{n_1m_3}\right)
\left(\d^{m_2m_4}\d^{n_2n_4} - \d^{m_2n_4}\d^{n_2m_4}\right) \Big]
\]
\[
+{1\over 2}\Big[\d^{n_1m_2}\d^{n_2m_3}\d^{n_3m_4}\d^{n_4m_1}
+\d^{n_1m_3}\d^{n_3m_2}\d^{n_2m_4}\d^{n_4m_1}
+\d^{n_1m_3}\d^{n_3m_4}\d^{n_4m_2}\d^{n_2m_1}
\]
\[
+ 45{\rm\, terms\, obtained\, by\, antisymmetrizing} 
{\rm\, on\, each\, pair\, of\, indices}\Big].
\]
One can check that its contraction with four field-strengths $F_{mn}$
gives the following expression
\[
t_8F^4 =
  8(F^1F^2F^3F^4) + 8(F^1F^3F^2F^4) + 8(F^1F^3F^4F^2)
\]
\[
 -2(F^1F^2)(F^3F^4) - 2(F^2F^3)(F^4F^1) - 2(F^1F^3)(F^2F^4),
\]
which is a convenient way of summarizing the $t_8$ tensor. One can 
also check that in terms of components
\[
t_8F^4 =
{1\over 2}\left[
- (k^2\cdot e^3)(k^2\cdot e^4)(e^1\cdot e^2)t
- (k^2\cdot e^4)(k^4\cdot e^3)(e^1\cdot e^2)t \right.
\]
\[
+ (k^2\cdot e^4)(k^3\cdot e^2)(e^1\cdot e^3)t
- (k^3\cdot e^4)(k^4\cdot e^2)(e^1\cdot e^3)t
+ (k^2\cdot e^3)(k^4\cdot e^2)(e^1\cdot e^4)t
\]
\[
+ (k^4\cdot e^2)(k^4\cdot e^3)(e^1\cdot e^4)t
- (k^2\cdot e^4)(k^3\cdot e^1)(e^2\cdot e^3)t 
- (k^2\cdot e^3)(k^4\cdot e^1)(e^2\cdot e^4)t
\]
\[
- (k^3\cdot e^1)(k^4\cdot e^3)(e^2\cdot e^4)t
- (k^4\cdot e^1)(k^4\cdot e^3)(e^2\cdot e^4)t 
+ (k^3\cdot e^1)(k^4\cdot e^2)(e^3\cdot e^4)t
\]
\[
- (k^2\cdot e^3)(k^2\cdot e^4)(e^1\cdot e^2)u
- (k^2\cdot e^3)(k^3\cdot e^4)(e^1\cdot e^2)u 
+ (k^2\cdot e^4)(k^3\cdot e^2)(e^1\cdot e^3)u
\]
\[
+ (k^3\cdot e^2)(k^3\cdot e^4)(e^1\cdot e^3)u
+ (k^2\cdot e^3)(k^4\cdot e^2)(e^1\cdot e^4)u 
- (k^3\cdot e^2)(k^4\cdot e^3)(e^1\cdot e^4)u
\]
\[
- (k^2\cdot e^4)(k^3\cdot e^1)(e^2\cdot e^3)u
- (k^3\cdot e^1)(k^3\cdot e^4)(e^2\cdot e^3)u 
- (k^3\cdot e^4)(k^4\cdot e^1)(e^2\cdot e^3)u
\]
\[
- (k^2\cdot e^3)(k^4\cdot e^1)(e^2\cdot e^4)u
+ (k^3\cdot e^2)(k^4\cdot e^1)(e^3\cdot e^4)u 
+ \half (e^1\cdot e^3)(e^2\cdot e^4)t^2 +
\]
\[ \left.
 \half (e^1\cdot e^4)(e^2\cdot e^3)u^2
+ \half (e^1\cdot e^4)(e^2\cdot e^3)tu
+ \half (e^1\cdot e^3)(e^2\cdot e^4)tu
- \half (e^1\cdot e^2)(e^3\cdot e^4)tu  \right],
\]
which is a useful representation when comparing against scattering 
amplitude computations.

The $t_8$ tensor can also be represented in a $U(5)$-covariant fashion
by taking \eqref{crazy_one} and going to the $\l^+$-frame
\[
t_8 F^4 = \e_{abcde}\langle (\l^+)^3 \t^a W^b W^c W^d W^e \rangle
= \e_{abcde}\langle (\l^+)^3  F^b_f F^c_g F^d_h F^e_i \t^a \t^f \t^g \t^h \t^i \rangle
\]
\[
= \d^{fghi}_{bcde} F^b_f F^c_g F^d_h F^e_i.
\]

%********************************************************

\end{document}